%% file: main.tex
\documentclass[a4paper,11pt]{article}
\usepackage{settings_custom}
\begin{document}

\title{Perturbative construction of mean-field equations \\ in extensive-rank matrix factorization and denoising}
\date{\today}
\author{Antoine Maillard$^{\star,\dagger,\diamond}$, Florent Krzakala$^{\star,\oplus}$, 
Marc M\'ezard$^{\star,\ddagger}$, Lenka Zdeborov\'a$^{\otimes}$}
\maketitle
{\let\thefootnote\relax\footnote{
$\star$ Laboratoire de Physique de l'\'Ecole Normale Sup\'erieure, PSL University, CNRS, Sorbonne Universit\'e, Universit\'e de Paris, France.\\
$\dagger$ Institute for Mathematical Research (FIM), ETH Zurich, Switzerland. \\
$\oplus$ IdePHICS laboratory, EPFL, Switzerland. \\
$\ddagger$ Department of Computing Sciences, Bocconi University, Italy.\\
$\otimes$ SPOC laboratory, EPFL, Switzerland. \\
$\diamond$ To whom correspondence shall be sent: \href{mailto:antoine.maillard@math.ethz.ch}{antoine.maillard@math.ethz.ch}.
}}

\begin{abstract}
    \noindent
    Factorization of matrices where the rank of the two factors diverges linearly with their sizes has many applications in diverse areas such as unsupervised representation learning, dictionary learning or sparse coding.
   We consider a setting where the two factors are generated from known component-wise independent prior distributions, and the statistician observes a (possibly noisy) component-wise function of their matrix product.
   In the limit where the dimensions of the matrices tend to infinity, but their ratios remain fixed, we expect to be able to derive closed form expressions for the optimal mean squared error on the estimation of the two factors.
	However, this remains a very involved mathematical and algorithmic problem.
    A related, but simpler, problem is extensive-rank matrix denoising, where one aims to reconstruct a matrix with extensive but usually small rank from noisy measurements.
    In this paper, we approach both these problems using high-temperature expansions at fixed order parameters. 
    This allows to clarify how previous attempts at solving these problems failed at finding an asymptotically exact solution.
    We provide a systematic way to derive the corrections to these existing approximations, taking into account the structure of correlations particular to the problem.
    Finally, we illustrate our approach in detail on the case of extensive-rank matrix denoising.
    We compare our results with known optimal rotationally-invariant estimators, and show how exact asymptotic calculations of the minimal error can be performed using extensive-rank matrix integrals.
\end{abstract}

\tableofcontents

\input{introduction.tex}
\input{plefka.tex}
\input{denoising.tex}
\input{conclusion.tex}

\section*{Acknowledgements}

We are grateful to Laura Foini, Christian Schmidt, Alice Guionnet, Jean-Philippe Bouchaud, and Yoshiyuki Kabashima for many insightful discussions throughout different stages of this work.
A.M. is particularly thankful to Laura Foini for collaboration on the annealed calculations and discussions on the PGY expansion.
We acknowledge funding from the ERC under the European Union’s Horizon 2020 Research and Innovation
Program Grant Agreement 714608-SMiLe.
Additional funding is acknowledged by A.M. from ‘Chaire de
recherche sur les modèles et sciences des données’, Fondation CFM pour la Recherche-ENS.

\bibliographystyle{alpha}
\bibliography{refs}

\appendix
\addtocontents{toc}{\protect\setcounter{tocdepth}{1}} 
\addcontentsline{toc}{part}{Appendices}

\input{appendix_rmt.tex}
\input{appendix_plefka_technicalities.tex}
\input{appendix_solution_matytsin.tex}
\input{appendix_small_alpha_RIE.tex}
\input{appendix_pgy_denoising.tex}
\input{appendix_annealed.tex}
\input{appendix_misc.tex}

\end{document}

%% file: introduction.tex
\section{Introduction}\label{sec:introduction}

\subsection{Motivation}\label{subsec:motivation}

The matrix factorization problem
is a very generic problem that appears in a range of  contexts. The problem is to reconstruct two matrices $\bF^\star$ and $\bX^\star $, knowing only a noisy component-wise function of their product $\bF^\star\bX^\star$. In the Bayesian setting that we adopt here, one assumes some prior knowledge of the probability distributions from which $\bF^\star$ and $\bX^\star $ have been generated. Depending on the context, this prior will enforce  some specific
requirements like sparsity, or non-negativity. Applications that can be formulated as matrix factorization include dictionary learning or sparse coding
\cite{olshausen1996emergence,olshausen1997sparse,kreutz2003dictionary,mairal2009online},
sparse principal component analysis \cite{zou2006sparse},
blind source separation \cite{belouchrani1997blind}, low rank matrix
completion \cite{candes2009exact,candes2010power}, or robust
principal component analysis \cite{candes2011robust}. 

\medskip\noindent
Here, following exactly the problem setting of \cite{kabashima2016phase}, we shall focus on the case of synthetic data, and inquire about the statistical limit of recovering matrices $\bF^\star$ and $\bX^\star$ from the observation of a function of their noisy product.
While this problem has been solved in detail in the case in which $\bF^\star, \bX^\star$ are low-rank (see e.g.\ \cite{baik2005phase,el2018estimation,lesieur2017constrained,donoho2018optimal}),
the case of extensive rank, where all dimensions of the matrices $\bF^\star$ and $\bX^\star $ go to infinity (with fixed ratios), is a much more challenging problem.
In this limit, we are interested in two types of questions:
\begin{itemize}
    \item[$1.$] Finding the theoretical limits of the recovery of $\bF^\star$ and $\bX^\star $, in particular understand when it is possible.
    \item[$2.$] Designing algorithms for this inference, and understanding their own limits.
\end{itemize}
These challenges were addressed in \cite{kabashima2016phase}, and related message-passing algorithms were developed in \cite{parker2014bilinear,parker2014bilinear2,zou2021multi,lucibello2021deep}.
We shall argue, however, that these previous works actually neglected some relevant correlations and do not provide a solution that is exact in the limit mentioned above.
In this work we propose a systematic way to approach this problem with a high-temperature type of expansion at fixed order parameter.
While the second order of this expansion gives back the previous result of \cite{kabashima2016phase}, we show that higher order corrections are relevant and cannot be neglected.

\medskip\noindent
We validate our approach on the special case corresponding to Gaussian priors on $\bF^\star$ and $\bX^\star$ and Gaussian additive noise on the observation of the product.
This setting can be solved exactly using high-dimensional ``HCIZ'' matrix integrals \cite{harish1957differential,itzykson1980planar,matytsin1994large}, or by leveraging known results on the denoising of rotationally invariant matrices \cite{bun2016rotational}.

\medskip\noindent
The high-temperature expansions presented in this paper open the way to a systematic construction of the mean-field ``TAP-like'' equations for this problem, which could lead, with a proper iteration scheme, to an efficient message-passing algorithm.
In future studies, it will also be interesting to make the contact between our high-temperature approach and the mean-field replica approach or the cavity method, which studies the statistical properties of the solutions to the TAP equations.
All these approaches should eventually converge to a unified understanding of extensive-rank matrix factorization, and hopefully to efficient message-passing algorithms.

\subsection{Setting of the problem}\label{subsec:setting}

We consider the matrix factorization problem in the extensive-rank setting, which we define as the following inference problem:
\begin{custommodel}{$\bF \bX$}[Extensive-rank matrix factorization]\label{model:extensive_factorization_fx}
    \noindent
    Consider $n,m,p\geq 1$. Extensive-rank matrix factorization is the inference
    of the matrices $\bF^\star \in \bbR^{m \times n}$ and $\bX^\star \in \bbR^{n \times p}$ from the observation of $\bY \in \bbR^{m \times p}$, generated as:
    \begin{align*}
        Y_{\mu l} &\sim P_{\rm out}\Big(\cdot\Big|\frac{1}{\sqrt{n}} \sum_{i=1}^n F^\star_{\mu i} X^\star_{il} \Big) , \qquad 1 \leq \mu \leq m, \quad 1 \leq l \leq p.
    \end{align*}
    We assume that the matrix elements of $\bF^\star$ and $\bX^\star$
    are both generated as i.i.d.\ random variables, with respective prior distributions $P_F$ and $P_X$, both having zero mean and all moments finite.

    \medskip\noindent
    We consider this inference problem in the high-dimensional (or \emph{thermodynamic}) limit, i.e.\ we assume 
$n,m,p \to \infty$ with finite limit ratios $m/n \to \alpha > 0$ and $p/n \to \psi > 0$. 
\end{custommodel}
\noindent
In order to estimate $\bF^\star$ and $\bX^\star$ the statistician can use  the \emph{posterior distribution}, also referred to as \emph{Gibbs measure} in the statistical physics literature\footnote{We will generically use greek indices $\mu,\nu$ for indices between $1$ and $m$, while latin $i,j,k$ indices will run from $1$ to $n$, and the $l$ index between $1$ and $p$.}:
\begin{align}\label{eq:gibbs_extensive_factorization_fx}
    P_{\bY,n}(\rd \bF,\rd \bX) &\equiv \frac{1}{\mathcal{Z}_{\bY,n}} \prod_{\mu,i} P_F(\rd F_{\mu i})\prod_{i,l} P_X(\rd X_{il}) \prod_{\mu,l} P_{\rm out}\Big(Y_{\mu l}\Big| \frac{1}{\sqrt{n}} \sum_i F_{\mu i} X_{il}\Big).
\end{align}
We assume that she/he knows the distributions $P_\mathrm{out},P_F,P_X$; this is known as the \emph{Bayes-optimal setting}.
In this case, it is well known that the mean under the posterior distribution of eq.~\eqref{eq:gibbs_extensive_factorization_fx} is the information-theoretically optimal estimator of $\bF^\star$ and $\bX^\star$. 

\medskip\noindent
In the thermodynamic limit, we can define the \emph{single-instance free entropy} $\Phi_{\bY,n}$\footnote{We normalize $\Phi_{\bY,n}$ by the total number $n(m+p)$ of variables to infer, while the normalization of \cite{kabashima2016phase} is $n^2$.}
of the system as:
\begin{align}\label{eq:def_phi_fx}
    \Phi_{\bY,n} &\equiv \frac{1}{n(m+p)} \ln \int \prod_{\mu,i} P_F(\rd F_{\mu i})\prod_{i,l} P_X(\rd X_{il}) \prod_{\mu,l} P_{\rm out}\Big(Y_{\mu l}\Big| \frac{1}{\sqrt{n}} \sum_i F_{\mu i} X_{il}\Big).
\end{align}
The averaged limit free entropy is denoted $\Phi \equiv \lim_{n \to \infty} \EE_\bY \, \Phi_{\bY,n}$.
We finally define the asymptotic minimal mean squared error (or MMSE) of the matrices $\bF$ and $\bX$ as
:
\begin{align}\label{eq:def_mmses_factorization}
    \begin{dcases}
    \mathrm{MMSE}_\bF&\equiv \lim_{n\to\infty} 
    \frac{1}{nm} \sum_{\mu,i} \EE_{\bY} \int P_F(\rd\bF^\star) P_X(\rd\bX^\star) \Big( \langle 
    F_{\mu i}
    \rangle - F^\star_{\mu i}\Big)^2, \\
    \mathrm{MMSE}_\bX &\equiv \lim_{n\to\infty} 
    \frac{1}{np} \sum_{i,l} \EE_{\bY} \int P_F(\rd\bF^\star) P_X(\rd\bX^\star) \Big( \langle 
    X_{i l}
    \rangle - X^\star_{i l} \Big)^2,
    \end{dcases}
\end{align}
where the bracket denotes an average with respect to the measure  $P_{\bY,n}(\rd \bF,\rd \bX)$ of eq.~\eqref{eq:gibbs_extensive_factorization_fx}.

\medskip\noindent
\textbf{Permutation symmetry --}
Note that since $\bF^\star, \bX^\star$ are generated i.i.d., there is a natural symmetry of the problem, 
in the sense that one can only recover them up to a common permutation $\{i \to \sigma(i)\}$ of the columns of $\bF^\star$ and the rows of $\bX^\star$. 
Therefore a meaningful definition of the MMSEs of eq.~\eqref{eq:def_mmses_factorization} 
requires to add an infinitely small side information to the system, breaking this symmetry. 
For lightness of the presentation, and since we will not evaluate numerically the factorization MMSEs, we 
do not introduce this technicality here.
Note that this issue also arises in low-rank matrix factorization, and is treated in the same way \cite{lesieur2017statistical}.
Another way to deal with this symmetry is to compare quantities derived from $\bF^\star, \bX^\star$ that are invariant under said symmetry.

\medskip\noindent
We shall also consider a symmetric version of this problem, where the matrix has been generated as  an empirical covariance matrix.
\begin{custommodel}{$\bX \bX^\intercal$}[Symmetric extensive-rank matrix factorization]\label{model:extensive_factorization_xx}
    \noindent
    Consider $n,m \geq 1$. Symmetric extensive-rank matrix factorization is defined as the inference
    of the matrix $\bX^\star \in \bbR^{m \times n}$ from the observation of $\bY$ generated as:
    \begin{align*}
        Y_{\mu \nu} &\sim P_{\rm out}\Big(\cdot\Big|\frac{1}{\sqrt{n}} \sum_{i=1}^n X^\star_{\mu i} X^\star_{\nu i} \Big) , \qquad 1 \leq \mu < \nu \leq m.
   \end{align*}
    We also assume that the elements of $\bX^\star$ are generated i.i.d.\ from a prior distribution $P_X$.
    We consider this inference problem in the thermodynamic limit, where $n,m \to \infty$ with finite limit ratio $m/n \to \alpha > 0$.
\end{custommodel}
As in the non-symmetric case, we can define both the posterior distribution of $\bX$ and the single-graph free entropy as:
\begin{subnumcases}{}
\label{eq:def_gibbs_xx}
    P_{\bY,n}(\rd \bX) \equiv \frac{1}{\mathcal{Z}_{\bY,n}} \prod_{\mu,i} P_X(\rd X_{\mu i}) \prod_{\mu < \nu} P_{\rm out}\Big(Y_{\mu \nu}\Big| \frac{1}{\sqrt{n}} \sum_i X_{\mu i} X_{\nu i}\Big), &\\
    \label{eq:def_phi_xx}
    \Phi_{\bY,n} \equiv \frac{1}{nm} \ln \int \prod_{\mu,i} P_X(\rd X_{\mu i}) \prod_{\mu < \nu} P_{\rm out}\Big(Y_{\mu \nu}\Big| \frac{1}{\sqrt{n}} \sum_i X_{\mu i} X_{\nu i}\Big). &
\end{subnumcases}
And the asymptotic MMSE on the matrix $\bX$ is defined as in the non-symmetric model (again assuming that we add an infinitely small side-information 
breaking the permutation symmetry between the columns of $\bX^\star$):
\begin{align}
 \mathrm{MMSE}_\bX&\equiv \lim_{n\to\infty} 
 \frac{1}{nm} \sum_{\mu i} \EE_\bY \int P_X(\rd\bX^\star) \Big( \langle 
 X_{\mu i}
 \rangle - X^\star_{\mu i}\Big)^2,
 \end{align}
 where the bracket denotes an average with respect to the measure  $P_{\bY,n}(\rd \bX)$ of eq.~\eqref{eq:def_gibbs_xx}.
 
\medskip\noindent
\textbf{The Gaussian setting --}
In the remaining of this work, we will denote \emph{Gaussian setting} the specific models (both symmetric and non-symmetric) in which all the prior distributions and the channel distributions are zero-mean Gaussians.
We consider the priors to have unit variance and the channel to have variance $\Delta > 0$.

\subsection{A related problem: matrix denoising}\label{subsec:intro_denoising}

Matrix denoising is a fundamental problem as well, with deep connection to the estimation of large covariance matrices in statistics and data analysis \cite{bun2017cleaning}. 
In the problem of \emph{denoising}, one aims at reconstructing a matrix $\bY^\star$ given by
\begin{align}
   Y_{\mu l}^\star =  \frac{1}{\sqrt{n}} \sum_{i=1}^n F^\star_{\mu i} X^\star_{il}
\end{align}
from noisy measurements of its elements, $ Y_{\mu l} \sim P_{\rm out} ( Y_{\mu l}| Y_{\mu l}^\star )$.
The statistician knows that the matrix $\bY^\star$ was generated as a product of $\bF^\star$ and $\bX^\star $, and while she/he knows the prior distribution of $\bF^\star$ and $\bX^\star $, she/he is only interested in reconstructing their product $\bY^\star$.  
If we denote by $\bG$ the estimate of $\bY^\star$, the posterior distribution $P_{\bY,n}^{\rm den}$ for denoising is
\begin{align}\label{eq:gibbs_denoising_fx}
    P_{\bY,n}^{\rm den}(\rd \bG) &\equiv \Bigg\{\int \rd \bF \rd \bX 
     P_{\bY,n}(\rd \bF,\rd \bX) \; \prod_{\mu,l} \delta\Big(G_{\mu l}-\frac{1}{\sqrt{n}} \sum_i F_{\mu i} X_{il}\Big)\Bigg\} \rd \bG.
\end{align}
This posterior is closely related to  the posterior $P_{\bY,n}(\rd \bF,\rd \bX) $ of matrix factorization defined in 
eq.~\eqref{eq:gibbs_extensive_factorization_fx}. 
In particular, these two posterior distributions share the same partition function ${\mathcal{Z}_{\bY,n}}$ and free entropy $\Phi_{\bY,n}$.
In a sense they are governed by the same measure, the difference being that the denoising statistician considers only observables that are built from the product $\bF \bX$.

\medskip\noindent
Similarly, in the symmetric setting, the denoising problem aims at reconstructing the product $\bX^\star (\bX^\star)^\intercal$
so that the denoising posterior is:
\begin{align}\label{eq:gibbs_denoising_xx}
    P_{\bY,n}^{\rm den}(\rd \bG) &\equiv \Bigg\{\int \rd\bX 
     P_{\bY,n}(\rd \bX) \; \prod_{\mu < \nu} \delta\Big( G_{\mu \nu}-\frac{1}{\sqrt{n}} \sum_i X_{\mu i} X_{\nu i}\Big)\Bigg\} \rd \bG,
\end{align}
which again can be obtained from eq.~\eqref{eq:def_gibbs_xx}. In the random matrix literature $\bX^\star (\bX^\star)^\intercal$ would often be called the Wishart matrix and we will hence denote this problem as denoising of a Wishart matrix. This case is particularly important as it corresponds to the problem of "cleaning" empirical correlation matrices.

\subsection{Related works}\label{subsec:related_works}

\textbf{Previous approximations to extensive-rank matrix factorization --}  
Beyond matrix factorization, many inference problems can be formulated in the statistical physics language. 
The large range of tools developed in the framework of statistical physics of disordered systems can then be used to tackle these models in the high-dimensional (or thermodynamic) limit.
This has been attempted for the extensive rank matrix factorization problem with generic priors and output channels in a series of papers: \cite{sakata2013statistical,krzakala2013phase,kabashima2016phase} applying the replica method to the problem.
However, as we argue in detail in Section~\ref{subsec:previous_approaches}, these works only provide an approximation, that is actually not exact in the thermodynamic limit due to some correlations between the variables that have been  neglected in these works.

\medskip\noindent
A parallel line of work derived ``approximate message-passing types'' of algorithms for the extensive rank matrix factorization problem \cite{parker2014bilinear,parker2014bilinear2}.
By analogy with other graphical models where these algorithms are amenable to an exact analysis via an asymptotic description known as \emph{state evolution} \cite{bayati2011dynamics,bayati2015universality,javanmard2013state,zdeborova2016statistical,gerbelot2021graph}, the authors of \cite{kabashima2016phase} derived such a state evolution, but again making the same over-simplifying approximation.
This implies that the state evolution stated in \cite{kabashima2016phase} is actually not following the behavior of the algorithm in the thermodynamic limit.
Other recent works use related algorithms and their asymptotic analysis in a multi-layer setting \cite{zou2021multi,lucibello2021deep}, 
relying essentially on the same approximation, which indicates that their analysis does not yield exact predictions in the thermodynamic limit.

\medskip\noindent
\textbf{Towards the exact solution of extensive rank matrix factorization and denoising --} 
While the present authors realized the flaws in the solution presented in \cite{kabashima2016phase} back in 2017, works aiming at fixing this problem and finding the exact mean-field solution for the matrix factorization problems are scarce.

\medskip\noindent
An important theoretical progress was made in the Ph.D. thesis of C.~Schmidt \cite{schmidt2018statistical}, who realized that the approximation of \cite{kabashima2016phase} is not asymptotically exact when transposing it 
to the simpler case of matrix denoising with Gaussian additive noise and Gaussian prior on $\bX^\star$.
The optimal error of matrix denoising for that case follows from a series of works on extensive-rank Harish-Chandra-Itzykson-Zuber (HCIZ) integrals \cite{matytsin1994large,guionnet2002large} that are expressed via the solution of a hydrodynamical problem.
In Section~\ref{sec:denoising}, we show that in some cases this hydrodynamical problem can be solved exactly by the Dyson Brownian motion, and numerically evaluated, complementing the arguments of \cite{schmidt2018statistical}.
The author of \cite{schmidt2018statistical} also attempted to correct the replica calculation of \cite{kabashima2016phase} including the missing correlations, but his proposition was also not leading to an exact solution.

\medskip\noindent
An important algorithmic version of the optimal denoising has been worked out for rotationally-invariant estimators in \cite{bun2016rotational}.
In Section~\ref{sec:denoising} we confirm that indeed the algorithm of \cite{bun2016rotational} gives exactly the error predicted from the HCIZ matrix integrals.

\medskip\noindent
Finally, another very recent work, \cite{barbier2021statistical}, came out during the completion of the present paper. 
An attempt of correcting the replica solution of \cite{kabashima2016phase} has already been provided in \cite{schmidt2018statistical}. Authors of \cite{barbier2021statistical} take this attempt further to develop
a ``spectral replica'' calculation 
for extensive-rank matrix factorization with non-Gaussian prior and (only) Gaussian channel. 
Its main result is an expression of the asymptotic free entropy, which is expressed as the function of a hydrodynamical HCIZ-like problem, 
making critical use of the results of \cite{forrester2016hydrodynamical,guionnet2021large} on the high-dimensional limit of ``rectangular'' HCIZ integrals.
While the expression obtained in \cite{barbier2021statistical} is closed and (conjectured to be) exact, a numerical solution of this hydrodynamical problem for cases of interest is yet to be provided.

\medskip\noindent
\textbf{TAP equations and Plefka-Georges-Yedidia expansion --}
An important strategy of theoretical statistical physics, pioneered in the context of spin glasses by Thouless-Anderson-Palmer \cite{thouless1977solution}, is to write mean-field equations that are exact for weakly coupled systems with infinite-range interactions.
This was achieved  in the Sherrington-Kirpatrick model \cite{sherrington1975solvable} by finding a free energy at fixed order parameters: one constrains the average magnetization of each spin $i$ to take a given value $m_i$, and writes a TAP free-energy as function of all the $m_i$'s.
This can then be optimized in order to precisely describe the physics of the system.
The cavity method \cite{mezard1986sk} can be used both to derive the TAP free energy, and to perform a statistical analysis of its solutions (by going back to cavity fields).
In more general models where the variables are not binary, the generalization of the TAP approach \cite{mezard1989space} aims at computing the free entropy of the system (e.g.\ eqs.~\eqref{eq:def_phi_fx} and \eqref{eq:def_phi_xx}), while constraining the first and second moments of the fields (e.g.\ of $\{X_{\mu i}\}$ in eq.~\eqref{eq:def_gibbs_xx}).
These moments thus become parameters of the resulting ``TAP'' free entropy, and extremizing over these variables leads to the celebrated \emph{TAP equations}.
These equations can be turned into algorithms, in which one seeks a solution of the TAP equations by iterating them, using appropriate time-iteration schemes \cite{bolthausen2014iterative}.

\medskip\noindent
The TAP line of approach has led to prolific developments in increasingly complex models.
For instance, when considering correlated data structures, such improvements include adaptive TAP (adaTAP) \cite{opper2001adaptive,opper2001tractable}, Expectation-Consistency (EC) \cite{minka2001expectation,opper2005expectation}, as well as recent algorithmic approaches based on approximate message-passing algorithms \cite{rangan2011generalized,donoho2013information,rangan2017vector,schniter2016vector,opper2016theory}.

\medskip\noindent
In a recent work \cite{maillard2019high}, all these approaches have been shown to lead to equivalent approximation schemes for rotationally-invariant data structures in finite-rank models.
 Moreover, it was shown that they could be understood in terms of another technique in the theoretical physics toolbox, namely \emph{high-temperature expansions at fixed order parameters}. 
These expansions were introduced by Plefka \cite{plefka1982convergence} for the Sherrington-Kirkpatrick (SK) model, and then amply generalized by Georges and Yedidia \cite{georges1991expand}.
In the following of this paper we shall denote them as \emph{Plefka-Georges-Yedidia} (or PGY) expansions.
The core of this method is to systematically compute high-temperature expansions of the free entropy of the system, while constraining the first and second moments of the fields.
This provides a systematic and controllable way to derive the TAP free energy and the TAP equations.
An important aim of the present work is to apply these approaches to the extensive-rank matrix factorization problem.
Importantly, we will see that in Model~\ref{model:extensive_factorization_xx}, this PGY expansion turns out to be an expansion in powers of the parameter $\alpha = m/n$.

\subsection{Organization of the paper and main results}\label{subsec:organization_main_results}

Let us describe the structure of the paper.
\begin{itemize}[leftmargin=*]
    \item Section~\ref{sec:plefka} is dedicated to the detailed derivation of the PGY ``high-temperature'' expansion 
    at fixed order parameters for extensive-rank matrix factorization.
    We show that the second order of the expansion gives back the results of the previous approach of \cite{kabashima2016phase}, but relevant new terms appear already at order $3$ for the symmetric problem, 
    and order $4$ for the non-symmetric.
    We also comment in detail on the type of correlations that were neglected in \cite{kabashima2016phase}.
    In the majority of this section we focus on the symmetric model, and we end the section by detailing the PGY expansion for the non-symmetric matrix factorization problem.
    \item In Section~\ref{sec:denoising}, we consider the simpler problem of extensive-rank matrix denoising with Gaussian additive noise. 
    First, we show that exact calculations of the minimal error can be performed using extensive-rank spherical integrals \cite{harish1957differential,itzykson1980planar},
    and that they match the error achieved by the rotationally-invariant estimator of \cite{bun2016rotational}. As far as we know the explicit link between these approaches (one algorithmic, one asymptotic) has not been made previously. 
    Secondly, we compare these results to the predictions of the PGY expansion developed in Section~\ref{sec:plefka}, 
    adapted for the denoising problem. We show that the order-$3$ corrections found in the PGY approach and neglected in \cite{kabashima2016phase} 
    indeed match the optimal estimator, providing evidence for the correctness of our PGY-based approach. 
\end{itemize}
Finally, note that while we focus in most of the following on Model~\ref{model:extensive_factorization_xx}, 
in which many of our results can be stated in a lighter manner, 
all our techniques can be transposed to Model~\ref{model:extensive_factorization_fx}, as we show for instance 
in Section~\ref{subsec:non_symmetric}.

\medskip\noindent
The main contributions of this paper can be summarized as follows. 

\medskip\noindent
\textbf{Extensive-rank matrix factorization for generic prior and noise channel --}
We develop the Plefka-Georges-Yedidia expansion for Model~\ref{model:extensive_factorization_xx} and Model~\ref{model:extensive_factorization_fx}, 
for any priors $P_X,P_F$ and observation channel $P_\mathrm{out}$. For the symmetric model~\ref{model:extensive_factorization_xx},
the second order provides the approximation of \cite{kabashima2016phase}. However, we show that the third order of the PGY expansion provides additional non-negligible terms.
Going further, we conjecture that the PGY approach at all orders would provide exact solution to the extensive rank matrix factorization and denoising. 
The advantage of this approach with respect to the replica method stems from its relation to the TAP equations and message-passing algorithms: because of this connection, it is directly related to an algorithm that can be used to solve practical instances of the problem. 
On the other hand, the recent replica result of \cite{barbier2021statistical} might help understand the large-$n$ behavior of the solution to the TAP equations, what is called \emph{state evolution} 
in the literature \cite{bayati2011dynamics,bayati2015universality,javanmard2013state,zdeborova2016statistical,gerbelot2021graph}.
Another important aspect is that the TAP equations are conjectured for any type of output channel $P_\mathrm{out}$, while recent replica results \cite{barbier2021statistical} are restricted to Gaussian additive noise.

\medskip\noindent
\textbf{Extensive-rank matrix denoising --}
We probe our theoretical calculations on the simpler problem of matrix denoising described in Section~\ref{subsec:intro_denoising}.
Our contribution are as follows:
\begin{itemize}[leftmargin=*]
    \item We derive an analytical expression for the asymptotic free entropy of the problem by using 
    HCIZ integrals \cite{matytsin1994large}. We provide an explicit solution to the optimization problem resulting from these integrals, 
    and show both analytically and numerically that it agrees with the performance of the rotationally-invariant denoising algorithm of \cite{bun2016rotational}. 
    \item We evaluate the minimum mean squared error (MMSE) of matrix denoising,
    both from our asymptotic formula based on HCIZ integrals (red lines in Fig.~\ref{fig:denoising_intro}), and from the algorithmic performance of the rotationally-invariant estimator (RIE) of \cite{bun2016rotational} (green points in Fig.~\ref{fig:denoising_intro}). 
    We compare this MMSE to the solution of \cite{kabashima2016phase}, that agrees with the second order of the PGY expansion developed in the present paper (blue points in Fig.~\ref{fig:denoising_intro}), 
    and show that it falls short of the actual MMSE.
    \item We truncate PGY expansion at third order (orange points in Fig.~\ref{fig:denoising_intro}) and show that, while we do not reach the MMSE at this order, 
    at small $\alpha$ we obtain a significant improvement with respect to the second-order truncation.
\end{itemize}

\begin{figure}[H]
    \includegraphics[width=\textwidth]{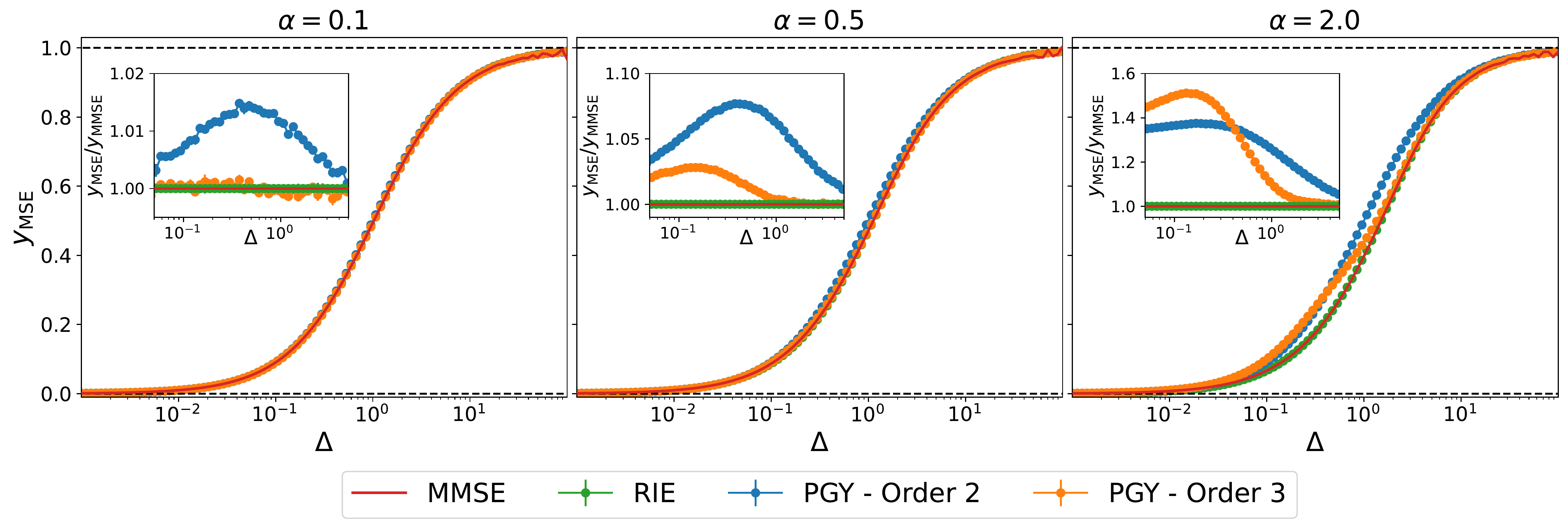}
	\caption{Denoising of a Wishart matrix as defined in eq.~\eqref{eq:gibbs_denoising_xx} observed with Gaussian noise of zero mean and variance $\Delta$. 
    The factors $\bX^\star$ are generated with i.i.d.\ ${\cal N}(0,1)$ components and $\alpha=m/n$ is their aspect ratio.
    We compare the MMSE obtained from the analytic prediction using HCIZ integrals, the RIE estimator, and the 
 estimator obtained from the PGY expansion truncated at orders $2$ and $3$. Note that 
 the PGY expansion is an expansion at small $\alpha$. In fact one sees that the third-order truncation works very well at small $\alpha$.
 The truncation at order $2$ corresponds to taking the limit $\alpha\to 0 $, and therefore to assuming
  that the underlying matrix is a Wigner matrix.
	\label{fig:denoising_intro}}
\end{figure}

\medskip\noindent
In Fig.~\ref{fig:denoising_intro} we illustrate some of the main results of this paper on the denoising problem.
The presented plots are for denoising of Wishart matrices as defined in eq.~\eqref{eq:gibbs_denoising_xx}, with Gaussian output channel of zero mean and variance $\Delta$. The factors $\bX^\star$ are generated with iid ${\cal N}(0,1)$ components. The three panels correspond to difference aspects ratios $\alpha = m/n$ of $\bX^\star$.
$\alpha \gg 1$ corresponds to denoising a small-rank matrix, a regime which we call \emph{undercomplete}.
On the other hand, the regime $\alpha \ll 1$ is called \emph{overcomplete}, and 
corresponds to denoising a matrix close to a Wigner matrix, i.e.\ with random independent components.
The four lines plotted in Fig.~\ref{fig:denoising_intro} are the following:
\begin{itemize}[leftmargin=*]
    \item \emph{MMSE}: The exact minimum mean squared error of the matrix denoising obtained analytically (as described in Section~\ref{subsec:mmse_denoising}, using the matrix integrals presented in Section \ref{subsec:fentropy_denoising}).
    For small $\alpha$ (in the overcomplete regime) the denoised matrix has almost independent elements, and thus the MMSE is very close to the MMSE of Wigner denoising.
    \item \emph{RIE}: The MSE obtained by the rotationally-invariant estimator of \cite{bun2016rotational}, presented in Section~\ref{subsec:Bouchaud_denoising}. 
    The algorithm was run on single instances of the problem with $m=3000$ and the data were averaged over 3 samples (error bars are invisible). The agreement with the MMSE is perfect. 
    \item \emph{PGY Order 2}: As explained in Section~\ref{subsec:previous_approaches} the second order of the PGY expansion for matrix denoising corresponds to ``Wigner denoising'', that is 
    it implicitly assumes that the components of the ground-truth matrix are independent.
    The approximation of \cite{kabashima2016phase} agrees exactly with the second order of the PGY expansion as we explain in Section~\ref{subsec:pgy_order2_extensive_fx}, it thus corresponds to Wigner (scalar) denoising.
    As in the overcomplete regime $\alpha \ll 1$ a Wishart matrix behaves like a Wigner matrix, this approximation is good in this domain, and deteriorates as $\alpha$ grows, as shown in Fig.~\ref{fig:denoising_intro}
    \item \emph{PGY Order 3}: This corresponds to truncating the PGY expansion at order $3$, as explained in Section~\ref{subsubsec:pgy_denoising_derivation}.
    It can be shown that increasing the order of the PGY expansion is equivalent to increasing the order considered in a small-$\alpha$ expansion of the denoiser, 
    which can be analytically computed, see eq.~\eqref{eq:gy_pgy_order_2}.
\end{itemize}

%% file: plefka.tex
\section{Derivation of the Plefka-Georges-Yedidia expansion}\label{sec:plefka}

\subsection{PGY expansion and TAP equations in the symmetric setting}\label{subsec:plefka_symmetric}

For a given instance of the problem, a generic way to compute the single-graph free entropy of eq.~\eqref{eq:def_phi_fx} is to follow the formalism of \cite{georges1991expand} to perform the high-temperature  expansion at fixed order parameters pioneered in \cite{plefka1982convergence}. 
Here we show how this PGY method can be applied the the symmetric Model~\ref{model:extensive_factorization_xx}.

\subsubsection{Sketch of the computation}\label{subsubsec:plefka_sketch}

In order to formulate our problem in a formalism adapted to high-temperature expansions, we first introduce an auxiliary field $\bHh \equiv \bX \bX^\intercal / \sqrt{n}$ in eq.~\eqref{eq:def_phi_xx}, which defines the free entropy:
\begin{align}\label{eq:introduction_Hhat_xx}
    e^{nm\Phi_{\bY,n}} &= \int \prod_{\mu,i} P(\rd X_{\mu i}) \prod_{\mu < \nu} \Big[ \int \rd \hat{H}_{\mu \nu} P_{\rm out}(Y_{\mu \nu}| \hat{H}_{\mu \nu}) \delta\Big(\hat{H}_{\mu \nu} - \frac{1}{\sqrt{n}} \sum_i X_{\mu i} X_{\nu i}\Big)\Big].
\end{align}
Using the Fourier transform of the delta function $\delta(x) = (2\pi)^{-1} \int \rd h \, e^{i h x}$, we reach
an effective free entropy in terms of two fields $\bX,\bH$, with $\bH \equiv \{H_{\mu \nu}\}_{\mu < \nu}$ a symmetric $m \times m$ matrix (whose diagonal can be taken equal to $0$):
\begin{align}\label{eq:phi_extensive_rank_xx_hfield}
    e^{nm\Phi_{\bY,n}} =  \int \prod_{\mu < \nu} P^{\mu \nu}_{H,Y}(\rd H_{\mu \nu}) \prod_{\mu,i} P_X(\rd X_{\mu i}) e^{-H_{\rm eff}[\bX,\bH]},
\end{align}
in which we introduced an \emph{effective Hamiltonian} $H_{\rm eff}[\bF,\bX,\bH]$ and un-normalized probability distributions $\{P_{H,Y}^{\mu l}\}$ defined as:
\begin{align}
\label{eq:def_H_eff}
    H_{\rm eff}[\bX,\bH] &\equiv \frac{1}{\sqrt{n}}\sum_{\mu <\nu} \sum_{i} (iH)_{\mu \nu} \,X_{\mu i} \,X_{\nu i} \ , \\
    \label{eq:def_Pmul_extensive_xx}
    P^{\mu \nu}_{H,Y} [\rd H] &\equiv \int \frac{\rd \hat{H}}{2 \pi} e^{i H \hat{H}} P_{\rm out}(Y_{\mu \nu}|\hat{H})\ .  
\end{align}
The PGY formalism allows to compute the free entropy of the system, constraining the means and variances of each 
variable $\{X_{\mu i},H_{\mu \nu}\}$ by ``tilting'' the Gibbs measure of eq.~\eqref{eq:phi_extensive_rank_xx_hfield}.
We thus impose:
\begin{align}\label{eq:def_means_variances_extensive_xx}
    \begin{cases}
    \langle X_{\mu i} \rangle = m_{\mu i}, & \langle (iH)_{\mu \nu} \rangle = -g_{\mu \nu}, \\
    \langle X_{\mu i}^2 \rangle = v_{\mu i} + (m_{\mu i})^2, & \langle (iH)_{\mu \nu}^2 \rangle = -r_{\mu \nu} + g_{\mu \nu}^2.
    \end{cases} 
\end{align}
Here, $\langle\cdot\rangle$ denotes an average over the (now tilted) Gibbs measure.
Note that we will sometimes symmetrize the quantities involved, e.g.\ we write $g_{\mu \nu} \equiv g_{\nu \mu}$ for $\mu > \nu$, and moreover we adopt the convention $g_{\mu \mu} = r_{\mu \mu} = 0$.
The resulting free entropy is a function of the means
and variances $\{\bmm,\bv,\bg,\br\}$, on which we will then have to extremize.
The conditions of eq.~\eqref{eq:def_means_variances_extensive_xx} will be imposed via Lagrange multipliers, which we denote respectively by 
$\{\blambda,\bgamma,\bomega,\bb\}$.
The free entropy can now be expressed as a function of all the parameters $\{\bmm,\bv,\bg,\br,\blambda,\bgamma,\bomega,\bb\}$ (we still denote it $\Phi_{\bY,n}$, with a slight abuse of notations):
\begin{align}\label{eq:phi_extensive_xx_before_plefka} 
    nm\Phi_{\bY,n} &= \sum_{\mu,i}\Big[\lambda_{\mu i} m_{\mu i} + \frac{\gamma_{\mu i}}{2}\Big(v_{\mu i} + (m_{\mu i})^2\Big) \Big] + \sum_{\mu < \nu}\Big[-\omega_{\mu \nu} g_{\mu \nu} - \frac{b_{\mu \nu}}{2}\Big(-r_{\mu \nu} + g_{\mu \nu}^2\Big) \Big] \\
    & + \ln \int P_{H,Y}(\rd \bH) P_X(\rd \bX) \, e^{-S_{\rm eff}[\bX,\bH]}, \nonumber 
\end{align}
in which we introduced an \emph{effective action} $S_\mathrm{eff}$: 
\begin{align}\label{eq:def_Seff_extensive_xx}
    S_{\rm eff}[\bX,\bH] &\equiv \sum_{\mu, i} \Big[\lambda_{\mu i} X_{\mu i} + \frac{\gamma_{\mu i}}{2} X_{\mu i}^2\Big]
    + \sum_{\mu<\nu} \Big[\omega_{\mu \nu} (iH)_{\mu \nu} - \frac{b_{\mu \nu}}{2} (iH_{\mu \nu})^2\Big] + H_{\rm eff}[\bX,\bH].  
\end{align}
We note that this effective action is a sum of decoupled terms, apart from a single "cubic interaction" term which is $H_{\rm eff} $ defined in eq.~\eqref{eq:def_H_eff}. 

\medskip\noindent
We now apply the standard PGY expansion by introducing a parameter $\eta > 0$ in front of the interaction Hamiltonian, i.e.\ replacing
$H_{\rm eff} \to \eta H_{\rm eff}$ in eq.~\eqref{eq:def_Seff_extensive_xx}. We shall then expand in powers of $\eta$, and  take $\eta = 1$  in the end.
We denote the corresponding free entropy $\Phi_{\bY,n}(\eta)$.
Importantly, for $\eta = 0$ all the fields $\{X_{\mu i},H_{\mu \nu}\}$ are independent, which allows for an efficient perturbative computation.

\medskip\noindent
In order to control the PGY expansion,
we introduce two natural assumptions on the structure of the correlations of the variables:
\begin{enumerate}[label=\textbf{H.\arabic*},ref=H.\arabic*]
    \item \label{hyp:means_uncorrelated_extensive_xx} At their physical value, the variables $\{m_{\mu i}\}$ should be uncorrelated, in coherence with the fact that the elements of $\bX^\star$ are drawn i.i.d. 
    Importantly, this is not true for $\bH$: although the elements $H_{\mu \nu}$ are independent by eq.~\eqref{eq:def_Pmul_extensive_xx}, their distribution depends on $(\mu,\nu)$ and therefore their statistics (e.g.\ the variables \{$g_{\mu \nu}$\}) might be correlated through the correlation of the variables $\{Y_{\mu \nu}\}$.
    \item \label{hyp:structure_g_extensive_xx_new}
     Recall that $g_{\mu \nu}$ is the average of $(iH)_{\mu \nu}$, the conjugate variable to $\hat H_{\mu \nu} \equiv \sum_i X_{\mu i} X_{\nu i} / \sqrt{n}$.
     In particular, as one can see from a simple calculation (see eqs.~\eqref{eq:interpretation_g_r_general},\eqref{eq:interpretation_gaussian_channel} that follow), 
     $\{g_{\mu \nu}\}$ takes the form of an independent-component estimator of the error achieved by $\bomega$, which is the estimator of the Wishart matrix $\bX \bX^\intercal / \sqrt{n}$.
     In order to estimate this error, we use the observed value of $Y_{\mu \nu}$, built as an independent component-wise non-linearity applied to the Wishart matrix $\bX^\star (\bX^\star)^\intercal / \sqrt{n}$.
    Similarly, we assume 
    that \emph{the correlations of the $\bg = \{g_{\mu \nu}\}$ scale as the ones of a matrix built as the $f(\bM \bM^\intercal / \sqrt{n})$, for $\bM$ a matrix with i.i.d.\ zero-mean components, and $f$ a -- possibly random -- component-wise function.}
\end{enumerate}
Assuming \ref{hyp:means_uncorrelated_extensive_xx} and \ref{hyp:structure_g_extensive_xx_new}, we derive
in Appendix~\ref{sec_app:technicalities_plefka} the first three terms of the small $\eta$ expansion. 
Note that \ref{hyp:structure_g_extensive_xx_new} is actually not necessary for Result~\ref{result:pgy_order_3_extensive_xx}, as we stop our computation of the series at order $3$.
However we believe that understanding the structure of $\{g_{\mu \nu}\}$ might prove 
critical when tackling higher-order terms in the expansion, similarly as \cite{maillard2019high} for finite-rank problems.
We leave for future work such a derivation of a general expression at all orders, and possibly a resummation.
\begin{result}[First three orders of the PGY expansion for Model~\ref{model:extensive_factorization_xx}]\label{result:pgy_order_3_extensive_xx}
\noindent 
At leading order as $n,m \to \infty$:
\begin{align}
&nm [\Phi_{\bY,n}(\eta) - \Phi_{\bY,n}(0)] = \frac{\eta}{\sqrt{n}} \sum_{\substack{\mu < \nu \\ i}} g_{\mu \nu} m_{\mu i} m_{\nu i} -\frac{\eta^2}{2 n} \sum_{\substack{\mu < \nu \\ i}} r_{\mu \nu}[v_{\mu i} v_{\nu i} + v_{\mu i} m_{\nu i}^2  + m_{\mu i}^2 v_{\nu i}]\nonumber \\
&+ \frac{\eta^2}{4 n} \sum_{\mu,\nu,i} g_{\mu \nu}^2 v_{\mu i} v_{\nu i} + \frac{\eta^3}{6 n^{3/2}} \sum_i \sum_{\substack{\mu_1,\mu_2,\mu_3 \\ \text{pairwise distinct}}} g_{\mu_1 \mu_2} g_{\mu_2 \mu_3} g_{\mu_3 \mu_1} v_{\mu_1 i} v_{\mu_2 i} v_{\mu_3 i} 
    + \mathcal{O}(\eta^4).
\end{align}
The zero-order term $\Phi_{\bY,n}(0)$ is given by (with implicit extremization over the Lagrange multipliers)
\begin{align}\label{eq:Phi_eta_0_extensive_xx}
    nm\Phi_{\bY,n}(0) &= \sum_{\mu,i}\Big[\lambda_{\mu i} m_{\mu i} + \frac{\gamma_{\mu i}}{2}\Big(v_{\mu i} + (m_{\mu i})^2\Big) + \ln \int P_X(\rd x) \, e^{-\frac{\gamma_{\mu i}}{2} x^2 - \lambda_{\mu i} x} \Big] \\
    &+ \sum_{\mu < \nu}\Big[-\omega_{\mu \nu} g_{\mu \nu} - \frac{b_{\mu \nu}}{2}\Big(-r_{\mu \nu} + g_{\mu \nu}^2\Big)  + \ln \int \rd z \, \frac{ e^{-\frac{1}{2 b_{\mu \nu}}(z - \omega_{\mu \nu})^2}}{\sqrt{2 \pi b_{\mu \nu}}} \, P_\mathrm{out}(Y_{\mu \nu}|z) \Big]. \nonumber
\end{align}
\end{result}
We now consider the fixed point equations that result from extremizing the free entropy of Result~\ref{result:pgy_order_3_extensive_xx}.
Note that the Lagrange multipliers only appear in the $\eta=0$ part of the free entropy (this is a general consequence of the formalism of \cite{georges1991expand}),
so we can easily write the maximization over these parameters:
\begin{align}\label{eq:fixedpoint_lagrange_extensive_xx}
    \begin{cases}
        m_{\mu i} = \EE_{P(\lambda_{\mu i}, \gamma_{\mu i})} [X], & v_{\mu i} = \EE_{P(\lambda_{\mu i}, \gamma_{\mu i})} [(X-m_{\mu i})^2], \\
        g_{\mu \nu} = g_{\rm out}(Y_{\mu \nu}, \omega_{\mu \nu}, b_{\mu \nu}), & r_{\mu \nu} = - \partial_\omega g_{\rm out}(Y_{\mu \nu}, \omega_{\mu \nu}, b_{\mu \nu}).
    \end{cases}
\end{align}
In eq.~\eqref{eq:fixedpoint_lagrange_extensive_xx} we defined $P_{\lambda,\gamma}$ and $g_\mathrm{out}$ as:
\begin{align}\label{eq:def_P_gout}
\begin{dcases}
    P(\lambda,\gamma)(x) &\equiv \frac{P_X(x) e^{- \frac{1}{2} \gamma x^2-\lambda x }}{\int \rd x' P_X(x') e^{- \frac{1}{2} \gamma x'^2-\lambda x'}}, \\
g_{\rm out}(y,\omega,b) &\equiv \frac{1}{b} \frac{\int \rd z \, P_{\rm out}(y|z) \, (z-\omega)\, e^{-\frac{(z-\omega)^2}{2b}}}{\int \rd z \, P_{\rm out}(y|z) \, e^{-\frac{(z-\omega)^2}{2b}}}. 
\end{dcases}
\end{align}
We shall then focus on the maximization over the parameters $\{\bmm, \bv, \bg, \br\}$. It is important to first understand the physical interpretation of $g_{\mu\nu}$ and $r_{\mu\nu}$.

\medskip\noindent
\textbf{General channels --}
Recall that we introduced the field $\hat{\bH}$ in eq.~\eqref{eq:introduction_Hhat_xx}, 
and then its conjugate field $\bH$ in eq.~\eqref{eq:phi_extensive_rank_xx_hfield}.
Recall as well that we defined $\{\bg,\br\}$ as the first and second moments that we constraint 
on the Gibbs measure, cf eq.~\eqref{eq:def_means_variances_extensive_xx}.
Because the variables $\{H_{\mu \nu}\}$ have been introduced as auxiliary parameters, the physical interpretation of $g_{\mu\nu}$ and $r_{\mu\nu}$ is not totally obvious, but it can be found through simple integration by parts, 
starting from eq.~\eqref{eq:def_means_variances_extensive_xx}.
We get, with $\hat{\bH} \equiv \bX \bX^\intercal / \sqrt{n}$:
\begin{align}\label{eq:interpretation_g_r_general}
g_{\mu \nu}=\Bigg\langle \frac{\partial_{\hat{H}} P_{\mathrm{out}}(Y_{\mu\nu}|\hat H_{\mu\nu})}{P_{\mathrm{out}}(Y_{\mu\nu}|\hat H_{\mu\nu})}\Bigg\rangle_{\bf Y}
\hspace{0.5cm}
;
\hspace{0.5cm}
r_{\mu \nu}=\Bigg\langle \frac{\partial^2_{\hat{H}} P_{\mathrm{out}} (Y_{\mu\nu}|\hat H_{\mu\nu})}{P_{\mathrm{out}}(Y_{\mu\nu}|\hat H_{\mu\nu})}\Bigg\rangle_{\bf Y}-g_{\mu\nu}^2.
\end{align}
\noindent
\textbf{Simplification for Gaussian channels --}
For a Gaussian channel (as is the case in the denoising problem), we have
$P_\mathrm{out}(Y| \cdot) = \mathcal{N}(Y, \Delta)$, and thus one gets:
\begin{align}\label{eq:interpretation_gaussian_channel}
    \begin{dcases}
        g_{\mu \nu} &= \frac{1}{\Delta} \Big[Y_{\mu \nu} - \frac{1}{\sqrt{n}} \Big \langle \sum_{i = 1}^n X_{\mu i} X_{\nu i} \Big\rangle\Big], \\	
        r_{\mu \nu} &= \frac{1}{\Delta} - \frac{1}{\Delta^2}  \Big[\Big \langle \Big(\frac{1}{\sqrt{n}}  \sum_{i = 1}^n X_{\mu i} X_{\nu i} \Big)^2 \Big\rangle - \Big \langle \frac{1}{\sqrt{n}}  \sum_{i = 1}^n X_{\mu i} X_{\nu i} \Big\rangle^2\Big].
    \end{dcases}
\end{align}
Notice that, from these equations, the PGY estimator for the denoising of $\bY$ takes a simple form in the case of a Gaussian channel:
\begin{align}
    \label{eq:TAP_estimator}
    \hat Y_{\mu\nu} \equiv Y_{\mu\nu}-\Delta g_{\mu\nu}.
\end{align}

\subsubsection{The series at order \texorpdfstring{$2$}{2} and the approximation of \texorpdfstring{\protect\cite{kabashima2016phase}}{[KKM16]}}\label{subsec:pgy_order2_extensive_fx}

We first examine the second order of the perturbation series of Result~\ref{result:pgy_order_3_extensive_xx}.
The free entropy reads:
\begin{align}\label{eq:phi_order2_extensive_fx}
&nm [\Phi_{\bY,n}(\eta) - \Phi_{\bY,n}(0)] = \frac{\eta}{\sqrt{n}} \sum_{\substack{\mu < \nu \\ i}} g_{\mu \nu} m_{\mu i} m_{\nu i} -\frac{\eta^2}{2 n} \sum_{\substack{\mu < \nu \\ i}} r_{\mu \nu}[v_{\mu i} v_{\nu i} + v_{\mu i} m_{\nu i}^2  + m_{\mu i}^2 v_{\nu i}] \nonumber \\
&+ \frac{\eta^2}{4 n} \sum_{\mu,\nu,i} g_{\mu \nu}^2 v_{\mu i} v_{\nu i} + \mathcal{O}(\eta^3). 
\end{align}
Since they are taken at $\eta = 0$, the fixed point equations of eq.~\eqref{eq:fixedpoint_lagrange_extensive_xx} are unchanged. 
The maximization over the physical order parameters (the means and variances) can be done and yields (we indicate on the right of the equation the corresponding parameter over which we maximized):
\begin{align}
\label{eq:tap_order_2_extensive_xx}
\begin{cases}
    b_{\mu \nu} =\frac{\eta^2}{n} \sum_{i} [v_{\mu i} v_{\nu i} + v_{\mu i} (m_{\nu i})^2 + (m_{\mu i})^2 v_{\nu i}], \qquad \qquad & (r_{\mu \nu}) \\
    \omega_{\mu \nu} = \frac{\eta}{\sqrt{n}} \sum_i m_{\mu i} m_{\nu i} -  g_{\mu \nu} [b_{\mu \nu} - \frac{\eta^2}{n} \sum_i v_{\mu i} v_{\nu i} ],  \qquad \qquad & (g_{\mu \nu}) \\
    \gamma_{\mu i} = \frac{\eta^2}{n} \sum_\nu [r_{\mu \nu} v_{i\nu} + r_{\mu \nu} (m_{\nu i})^2 - g_{\mu \nu}^2 v_{\nu i}], \qquad \qquad & (v_{\mu i}) \\
    \lambda_{\mu i} = -\frac{\eta}{\sqrt{n}} \sum_\nu g_{\mu \nu} m_{\nu i} + m_{\mu i} [-\gamma_{\mu i} + \frac{\eta^2}{n} \sum_\nu r_{\mu \nu} v_{\nu i} ]. \qquad \qquad & (m_{\mu i}) \\
\end{cases}
\end{align}
Taking $\eta = 1$, the reader can check easily that the combined equations~\eqref{eq:fixedpoint_lagrange_extensive_xx} and \eqref{eq:tap_order_2_extensive_xx} 
are actually completely equivalent to the GAMP equations derived in \cite{kabashima2016phase},  replacing the notations of the variables as follows\footnote{\cite{kabashima2016phase} considers primarily the non-symmetric Model~\ref{model:extensive_factorization_fx}, 
but one can very easily transfer their equations to Model~\ref{model:extensive_factorization_xx}.}:
\begin{align*}
    b_{\mu \nu} &\to V_{\mu \nu} & \omega_{\mu \nu} &\to \omega_{\mu \nu} & g_{\mu \nu} &\to g_{\rm out}(Y_{\mu \nu}, \omega_{\mu \nu}, V_{\mu \nu})  & r_{\mu \nu} &\to -\partial_\omega g_{\rm out}(Y_{\mu \nu}, \omega_{\mu \nu}, V_{\mu \nu}) \\
    \gamma_{\mu i} &\to Z^{-1}_{\mu i} & \lambda_{\mu i} &\to -\frac{W_{\mu i}}{Z_{\mu i}} & m_{\mu i} &\to \hat{f}_{\mu i} & v_{\mu i} &\to s_{\mu i}.
\end{align*}
In conclusion, the PGY expansion truncated to order $\eta^2$ gives back exactly the stationary limit of the GAMP algorithm of \cite{kabashima2016phase}.
However, as we will see below, the higher-order corrections of order $\eta^3$ (and most probably beyond) cannot be neglected: this shows explicitly how the GAMP equations of \cite{kabashima2016phase} (and also the BiGAMP equations of e.g.\ \cite{parker2014bilinear,parker2014bilinear2,zou2021multi}, which are based on the same approximation) are missing relevant terms.

\subsubsection{The series at third order: a non-trivial correction}\label{subsubsec:higher_orders_extensive_xx}

While the order $\eta^2$ truncation of Result~\ref{result:pgy_order_3_extensive_xx} yields back the approximation of \cite{kabashima2016phase}, 
we have computed via the PGY expansion the order $\eta^3$, which we recall:
\begin{align}\label{eq:phi_order3_extensive_xx}
    \frac{1}{3! n m} \frac{\partial^3 \bPhi_{\bY,n}}{\partial \eta^3}(\eta = 0) &= \frac{1}{6 m n^{5/2}} 
    \sum_i \sum_{\substack{\mu_1,\mu_2,\mu_3 \\ \text{pairwise distinct}}} g_{\mu_1 \mu_2} g_{\mu_2 \mu_3} g_{\mu_3 \mu_1} v_{\mu_1 i} v_{\mu_2 i} v_{\mu_3 i} + \smallO_n(1).
\end{align}
A crucial observation is that in general the RHS of eq.~\eqref{eq:phi_order3_extensive_xx} is not negligible as $n \to \infty$.
This is a consequence of the diagrammatic analysis of \cite{maillard2019high}.
As an example, consider the case of a Gaussian additive noise for the observation channel.
Given that $\bY$ is the sum of a Wishart and a Wigner matrix, it is natural to assume that $\bg$ behaves like a rotationally-invariant matrix, see \ref{hyp:structure_g_extensive_xx_new}.
Assuming moreover for simplicity that the variances are all equal $v_{\mu i} = v$, Theorem~1 of \cite{maillard2019high} (under the rotation-invariance hypothesis we described) yields:
\begin{align*}
    \frac{1}{3! n m} \frac{\partial^3 \bPhi_{\bY,n}}{\partial \eta^3}(\eta = 0) &= \frac{v^3 \alpha^{3/2}}{6} c_3(\bg / \sqrt{m}) + \smallO_n(1),
\end{align*}
in which the coefficients $\{c_p(\bg / \sqrt{m})\}$ are the \emph{free cumulants} of random matrix theory, and are functions of the spectrum of $\bg / \sqrt{m}$, cf.\ Appendix~\ref{sec_app:rmt}. 
Therefore, assuming that the bulk of eigenvalues of $\bg/\sqrt{m}$ stays of order $1$ as $n \to \infty$ (which is a natural scaling given eq.~\eqref{eq:TAP_estimator}), 
this order-$3$ term gives a non-negligible contribution to the free entropy.
This indicates that the approximation of \cite{kabashima2016phase} breaks down, even in the simple case of Gaussian noise.

\medskip\noindent
The third-order correction changes the order-2 TAP equations~\eqref{eq:tap_order_2_extensive_xx} as follows :
\begin{align}
\label{eq:tap_order_3_extensive_xx}
\begin{cases}
    b_{\mu \nu} =\frac{\eta^2}{n} \sum_{i} [v_{\mu i} v_{\nu i} + v_{\mu i} (m_{\nu i})^2 + (m_{\mu i})^2 v_{\nu i}], \qquad \qquad & (r_{\mu \nu}) \\
    \omega_{\mu \nu} = \frac{\eta}{\sqrt{n}} \sum_i m_{\mu i} m_{\nu i} -  g_{\mu \nu} [b_{\mu \nu} - \frac{\eta^2}{n} \sum_i v_{\mu i} v_{\nu i} ]+\frac{\eta^3}{2n^{3/2}} \sum_{\rho} g_{\nu\rho}g_{\rho\mu} \sum_i v_{\mu i}v_{\nu i}v_{\rho i} ,  \qquad \qquad & (g_{\mu \nu}) \\
    \gamma_{\mu i} = \frac{\eta^2}{n} \sum_\nu [r_{\mu \nu} v_{i\nu} + r_{\mu \nu} (m_{\nu i})^2 - g_{\mu \nu}^2 v_{\nu i}]-\frac{\eta^3}{n^{3/2}} \sum_{\nu,\rho} g_{\mu \nu} g_{\nu\rho}g_{\rho\mu} v_{\nu i}v_{\rho i}, \qquad \qquad & (v_{\mu i}) \\
    \lambda_{\mu i} = -\frac{\eta}{\sqrt{n}} \sum_\nu g_{\mu \nu} m_{\nu i} + m_{\mu i} [-\gamma_{\mu i} + \frac{\eta^2}{n} \sum_\nu r_{\mu \nu} v_{\nu i} ]. \qquad \qquad & (m_{\mu i}) \\
\end{cases}
\end{align}
It is interesting to see the typical scaling of the correction terms of order $\eta^3$ in the thermodynamic limit: 
\begin{itemize}[leftmargin=*]
    \item As we saw in the argument above, we expect the correction to the free entropy in eq.~\eqref{eq:phi_order3_extensive_xx} to scale as $\eta^3 \alpha^{3/2}$ in the thermodynamic limit.
    \item The correction to $\omega_{\mu\nu}$ scales as $(\eta^3/n^{3/2}) \sum_{\rho} g_{\nu\rho}g_{\rho\mu} \sum_i v_{\mu i}v_{\nu i}v_{\rho i}$, 
    while the correction to $\gamma_{\mu i}$ scales as $ (\eta^3/n^{3/2}) \sum_{\nu,\rho}g_{\mu \nu} g_{\nu\rho}g_{\rho\mu} v_{\nu i}v_{\rho i} $. 
    Using similar arguments as above in the case of Gaussian additive noise, we expect that these corrections scale 
    respectively as $\eta^3 \alpha^{1/2}$ and $\eta^3 \alpha^{3/2}$ in the thermodynamic limit.
\end{itemize}
\noindent
Therefore the corrections at order $\eta^3$ are relevant and cannot be neglected, except in the small $\alpha$ limit where they become small.
We end this discussion by a few remarks on the results presented here:
\begin{itemize}[leftmargin=*]
    \item \textbf{Similarity with finite-rank problems --}
    One may notice that the first orders of Result~\ref{result:pgy_order_3_extensive_xx} are very similar
    the the TAP free entropy of simple symmetric models with rotationally-invariant couplings,
    that is derived in \cite{maillard2019high}, see eq.~(25) of the mentioned paper.
    A crucial difference is that here the role of the coupling matrix is played by $\bg$, which is itself a parameter of the TAP free entropy.
    While this limits to the first order presented in Result~\ref{result:pgy_order_3_extensive_xx}, this similarity is already striking.
    Importantly, the resummation of all orders of the PGY expansion in \cite{maillard2019high} then suggests that all orders $p \geq 3$ of the free entropy might perhaps be expressible solely in terms of the 
    \emph{eigenvalue distribution} of $\bg/\sqrt{m}$.
    \item \textbf{Nature of the order parameter --}
    The previous remark suggests that the order parameter in extensive-rank matrix factorization 
    would be a probability distribution of eigenvalues\footnote{This hint is strengthened by recent  results in \cite{barbier2021statistical} using the spectral replica method: 
    in the special case of Gaussian channels $P_\mathrm{out}$, it analytically computes the asymptotic mean free energy with the replica method, 
    and its results strongly suggest that the proper order parameter is indeed a probability measure of eigenvalues.}.
    This is an important difference from the finite-rank case: for rank-$k$ matrix factorization (or in general for rank-$k$ recovery problems), the state of the system is governed by a $k \times k$ overlap matrix.
    \item \textbf{Small $\alpha$ series --} 
    As we have seen, the corrections at order $\eta^3$ to the free entropy are of order $\alpha^{3/2}$.
    Here we remind that $\alpha$ is the ratio between the size of the observed matrix $m$ and the size of the factor $n$.
    Large $\alpha$ corresponds to the low-rank matrix estimation (or undercomplete regime) and small $\alpha$ leads to the observed matrix having elements close to i.i.d. (or overcomplete regime).
    It seems to us that the expansion in powers of $\eta$, which is the core of the PGY approach, actually turns into an expansion of the TAP equations in powers of $\alpha$. 
    While this is easily seen up to order $3$ from Result~\ref{result:pgy_order_3_extensive_xx}, 
    we have no formal proof of this statement for all orders, and leave it as a conjecture.
    \item \textbf{Iterating the equations --} 
    As thoroughly discussed in \cite{maillard2019high}, the TAP equations are in general not sufficient to obtain an algorithm with good convergence properties. 
    A classical example is given by the Generalized Vector Approximate Passing algorithm (G-VAMP) \cite{schniter2016vector}, for which the corresponding TAP equations are derived in \cite{maillard2019high}
    using PGY expansions.
    As highlighted there, the TAP equations correspond to the stationary limit of G-VAMP, however, \emph{there is no obvious iterative resolution scheme of the TAP equations that automatically gives the GVAMP algorithm}. 
    This indicates that even if one obtains all the orders of perturbations in Result~\ref{result:pgy_order_3_extensive_xx}, turning them into an efficient algorithm 
    may require further work\footnote{We unsuccessfully tried naive iteration schemes to solve eqs.~\eqref{eq:fixedpoint_lagrange_extensive_xx},\eqref{eq:tap_order_3_extensive_xx}, as we found the algorithm to always either diverge or converge to a trivial solution 
    $\bmm = 0$. We leave a more precise investigation of the numerical properties of these iterations to a future work.}. 
\end{itemize}

\subsection{Nature of the approximation in previous approaches}\label{subsec:previous_approaches}

As we have found that the conjecture of \cite{kabashima2016phase} about exactness of their solution is incorrect, it is useful to go back to it and understand where that approach failed. 
In \cite{kabashima2016phase}
the free entropy of the matrix factorization problem is estimated in the Bayes-optimal setting, in two different ways, namely via the replica method and via 
belief propagation (BP) equations. In both cases there are somewhat hidden hypotheses which we think are not valid. 
These hypotheses are also present in the derivation of the BiGAMP (and BiG-VAMP) algorithm (cf.\ e.g.\ \cite{parker2014bilinear, parker2014bilinear2,zou2021multi}).
We therefore believe that these algorithms are also not able to give an exact asymptotic computation of the marginal probabilities in this problem. 

\medskip\noindent
Let us now describe both approaches taken in \cite{kabashima2016phase}, and explain how the assumptions behind them fail, 
focusing primarily on the replica analysis performed in Section~V.B of \cite{kabashima2016phase}.
The main idea behind the replica method is to compute the \emph{quenched} free entropy from the evaluation of the moments of the partition function, using the relation \cite{mezard1987spin}:
\begin{align*}
    \lim_{n \to \infty}\frac{1}{n^2}\EE_\bY \ln \mathcal{Z}_{\bY,n} &= \frac{\partial}{\partial r} \Big[\lim_{n \to \infty} \frac{1}{n^2} \ln \EE_\bY \mathcal{Z}_{\bY,n}^r \Big]_{r=0} . 
\end{align*}
The computation of the quenched free entropy then reduces to the evaluation of the integer moments of the partition function, by analytically expanding the expression of the $r$-th moment 
to any $r > 0$.
When writing out $\EE_\bY \mathcal{Z}_{\bY,n}^r$ there naturally appears $(r+1)$ \emph{replicas} of the system, that interact via the channel distribution term, as represented in the following equation (here we consider Model~\ref{model:extensive_factorization_fx}, to 
be in the same setting as \cite{kabashima2016phase}):
\begin{align*}
    \EE_\bY \mathcal{Z}_{\bY,n}^r &= \int \rd \bY \Big[\prod_{a=0}^{r} P_F(\rd \bF^a) P_X(\rd \bX^a)\Big] \prod_{a=0}^r \Big[\prod_{\mu,l} P_{\rm out} \Big(Y_{\mu l}\Big| \frac{1}{\sqrt{n}}\sum_i F^a_{\mu i} X^a_{il}\Big)\Big].
\end{align*}
A key step in the calculation of \cite{kabashima2016phase} is the assumption that 
\begin{align}\label{eq:def_Za_mul_extensive_fx}
  Z^{a}_{\mu l} \equiv \frac{1}{\sqrt{n}} \sum_i F^a_{\mu i} X^a_{il}  
\end{align}
are multivariate Gaussian random variables. However, although 
$\{\bF^a\}$ and $\{\bX^a\}$ follow statistically independent distributions, with zero mean and finite variances, this is not enough to guarantee the Gaussianity of the $Z^{a}_{\mu l}$ variables in the high-dimensional limit.
Indeed, there is a number $\mathcal{O}(n^2)$ of variables $Z^{a}_{\mu l}$, and therefore classical central limit results cannot conclude on the asymptotic Gaussianity of the joint distribution of such variables.
In general, this Gaussianity is actually false, as can be seen e.g.\ by considering the following quantity, for a single replica:
\begin{align}\label{eq:def_L4}
        L_4 &\equiv \lim_{n \to \infty}\EE \Bigg[\frac{1}{n^3} \sum_{\mu_1 \neq \mu_2}\sum_{l_1 \neq l_2} Z_{\mu_1 l_1} Z_{\mu_1 l_2} Z_{\mu_2 l_2} Z_{\mu_2 l_1} \Bigg].
\end{align}
Let us assume that both $P_F$ and $P_X$ are standard Gaussian distributions for the simplicity of the argument. The computation of $L_4$ then 
yields $L_4 = \alpha^2 \psi^2$. However, should the \emph{joint} distribution of the $\{Z_{\mu l}\}$ converge to a multivariate (zero-mean) Gaussian distribution, 
such a distribution would satisfy by definition
\begin{align}
  \EE[Z_{\mu_1 l_1} Z_{\mu_2 l_2}] = \frac{1}{n} \sum_{i,j} \EE [F_{\mu_1 i} F_{\mu_2 j}] \EE [X_{l_1 i} X_{l_2 j}] = \delta_{\mu_1 \mu_2} \delta_{l_1 l_2} \; ,
\end{align}
and  Wick's theorem would give, wrongly, $L_4 = 0$.
Therefore the joint distribution of $\{Z^a_{\mu l}\}$ is not Gaussian.

\medskip\noindent
\textbf{The message-passing approach --}
Another approach to the problem are
the Belief-Propagation (BP) equations \cite{mezard2009information}, also referred to as the \emph{cavity method} in the physics literature \cite{mezard1986sk}.
The goal of BP is to compute the \emph{marginal} distributions of each variable in the system, 
by solving iterative equations involving probability distributions over each single variable. These probability distributions are called \emph{messages} in the BP language, 
and the fixed point of the iterative equations yields an estimate of the marginal distributions. 
While a detailed treatment of the BP derivation of \cite{kabashima2016phase} is beyond the scope of this paper, 
we have seen in the previous Section~\ref{subsec:plefka_symmetric} that  the message-passing approach of \cite{kabashima2016phase} is equivalent to neglecting (wrongly) some higher order terms in the PGY expansion.

\subsection{Non-symmetric extensive-rank matrix factorization}\label{subsec:non_symmetric}

Performing the PGY expansion for Model~\ref{model:extensive_factorization_fx} is extremely similar to the calculation we have done for Model~\ref{model:extensive_factorization_xx}:
one can introduce a field $\tilde{\bH} \equiv \bF \bX / \sqrt{n}$, and then perform the same calculations via the Fourier transform of the delta function. 
In the following of Section~\ref{subsec:non_symmetric}, we give the results of our derivation, while more details are given in Appendix~\ref{subsec_app:pgy_fx}.
We adopt similar notations to the ones of Section~\ref{subsec:plefka_symmetric}, adding additional $X,F$ subscripts to differentiate the two fields. More precisely, we impose the first and second moment constraints:
\begin{align}\label{eq:def_means_variances_extensive_fx}
    \begin{cases}
    \langle F_{\mu i} \rangle = m^F_{\mu i}, & \langle X_{il} \rangle = m^X_{i l},  \\
    \langle F_{\mu i}^2 \rangle = v^F_{\mu i} + (m^F_{\mu i})^2, & \langle X_{il}^2 \rangle = v^X_{il} + (m^X_{il})^2, \\
    \langle (iH)_{\mu l} \rangle = -g_{\mu l}, & \langle (iH)_{\mu l}^2 \rangle = -r_{\mu l} + g_{\mu l}^2.
    \end{cases} 
\end{align}
\textbf{The PGY expansion at order $\eta^4$ --}
Similarly to Result~\ref{result:pgy_order_3_extensive_xx}, we obtain the first orders of the free entropy as:
\begin{result}[First orders of the PGY expansion for Model~\ref{model:extensive_factorization_fx}]\label{result:pgy_order_4_extensive_fx}
\noindent
We have, at leading order as $n \to \infty$:
\begin{align*}
&n(m+p) [\Phi_{\bY,n}(\eta) - \Phi_{\bY,n}(0)] = \frac{\eta}{\sqrt{n}} \sum_{\mu,i,l} g_{\mu l} m^F_{\mu i} m^X_{il} -\frac{\eta^2}{2 n} \sum_{\mu,i,l} r_{\mu l}[v^F_{\mu i} (m^X_{il})^2  + (m^F_{\mu i})^2 v^X_{il}] \\
& + \frac{\eta^2}{2n} \sum_{\mu,i,l} (g_{\mu l}^2 - r_{\mu l}) v^F_{\mu i} v^X_{il} + \frac{\eta^4}{4n^2} \sum_i \sum_{\mu_1 \neq \mu_2} \sum_{l_1 \neq l_2} g_{\mu_1 l_1} g_{\mu_2 l_1} g_{\mu_2 l_2} g_{\mu_1 l_2} v^F_{\mu_1 i} v^F_{\mu_2 i} v^X_{i l_1} v^X_{i l_2}+ \mathcal{O}(\eta^5).
\end{align*}
\end{result}
Recall that the term $\Phi_{\bY,n}(0)$ contains the dependency on the channel and priors contributions, as well as the Lagrange multipliers introduced to enforce the conditions of eq.~\eqref{eq:def_means_variances_extensive_fx}. 
Its precise form is: 
\begin{align}\label{eq:Phi_eta_0_extensive_xfx}
    &n(m+p)\Phi_{\bY,n}(0) = \sum_{\mu,i}\Big[\lambda^F_{\mu i} m^F_{\mu i} + \frac{\gamma^F_{\mu i}}{2}\Big(v^F_{\mu i} + (m^F_{\mu i})^2\Big) +\ln \int P_F(\rd f) \, e^{-\frac{\gamma^F_{\mu i}}{2} f^2 - \lambda^F_{\mu i} f} \Big] \nonumber\\
    &+\sum_{i,l}\Big[\lambda^X_{ i l} m^X_{i l} + \frac{\gamma^X_{i l}}{2}\Big(v^X_{i l} + (m^X_{i l})^2\Big) +\ln \int P_X(\rd x) \, e^{-\frac{\gamma^X_{i l}}{2} x^2 - \lambda^X_{i l} x} \Big]\nonumber \\
    &+ \sum_{\mu,l}\Big[-\omega_{\mu l} g_{\mu l} - \frac{b_{\mu l}}{2}\Big(-r_{\mu l} + g_{\mu l}^2\Big)  + \ln \int \rd z \, \frac{ e^{-\frac{1}{2 b_{\mu l}}(z - \omega_{\mu l})^2}}{\sqrt{2 \pi b_{\mu l}}} \, P_\mathrm{out}(Y_{\mu l}|z) \Big].
\end{align}

\medskip\noindent
\textbf{Higher-order terms, and breakdown of previous approximations --}
As in the symmetric case, the approximation of \cite{kabashima2016phase,parker2014bilinear,parker2014bilinear2,zou2021multi} amounts to truncating the perturbation series of Result~\ref{result:pgy_order_4_extensive_fx} 
at order $\eta^2$.
However, exactly as for Model~\ref{model:extensive_factorization_xx}, the higher-order terms in Result~\ref{result:pgy_order_4_extensive_fx} are in general non-negligible. 
Under a similar hypothesis as \ref{hyp:structure_g_extensive_xx_new}, they are related to the asymptotic singular value distribution of $\bg/\sqrt{n}$, similarly to what happened in the symmetric case.
The limit of such terms
has been worked out in \cite{maillard2019high}.
Assuming e.g.\ $v^F_{\mu i} = v^F$ and $v^{X}_{il} = v^X$, we have at $\eta = 0$:
\begin{align}\label{eq:order_4_pgy_extensive_fx}
    \frac{1}{4! n (m+p)} \frac{\partial^4 \bPhi_{\bY,n}}{\partial \eta^4}(\eta = 0) &= \frac{\psi}{\alpha+\psi} \frac{(v^F)^2 (v^X)^2}{4} \Gamma_2\Big(\frac{\alpha}{\psi},\frac{\bg^\intercal \bg}{n}\Big) + \smallO_n(1).
\end{align}
The functions $\Gamma_p$ are related to the rectangular free cumulants introduced in \cite{benaych2011rectangular}, and their 
precise definition is given in \cite{maillard2019high}.
Again, in general these coefficients are non-negligible in the limit $n \to \infty$.

\medskip\noindent
For completeness, let us write the TAP equations for this problem, with the corrections arising from the order $\eta^4$.
As in the symmetric case the extremization with respect to the Lagrange multipliers only depends on the zero-th order term, and yields:
\begin{align}\label{eq:TAP_non_symmetric_order_0}
    \begin{dcases}
        m^F_{\mu i} = \EE_{P_F(\lambda^F_{\mu i}, \gamma^F_{\mu i})} [f], & v^F_{\mu i} = \EE_{P_F(\lambda^F_{\mu i}, \gamma^F_{\mu i})} [(f-m^F_{\mu i})^2], \\
        m^X_{i l} = \EE_{P_X(\lambda^X_{i l}, \gamma^X_{i l})} [x], & v^X_{i l} = \EE_{P_X(\lambda^X_{i l}, \gamma^X_{i l})} [(x-m^X_{i l})^2], \\
        g_{\mu l} = g_{\rm out}(Y_{\mu l}, \omega_{\mu l}, b_{\mu l}), & r_{\mu l} = - \partial_\omega g_{\rm out}(Y_{\mu l}, \omega_{\mu l}, b_{\mu l}), \\
    \end{dcases}
\end{align}
The fourth-order corrections affect the remaining TAP equations, which read: 
\begin{align}\label{eq:TAP_non_symmetric_order_4}
    \begin{dcases}
        b_{\mu l} &= \frac{\eta^2}{n} \sum_{i} [v^F_{\mu i} v^X_{il} + v^F_{\mu i} (m^X_{il})^2 + (m^F_{\mu i})^2 v^X_{il}], \\
        \omega_{\mu l} &= \frac{\eta}{\sqrt{n}} \sum_i m^F_{\mu i} m^X_{il} -  g_{\mu l} [b_{\mu l} - \frac{\eta^2}{n} \sum_i v^F_{\mu i} v^X_{il} ] + \frac{\eta^4}{n^2} \sum_{i,\mu',l'} g_{\mu' l} g_{\mu' l'} g_{\mu l'} v^F_{\mu i} v^F_{\mu' i} v^X_{i l} v^X_{i l'}, \\
        \gamma^F_{\mu i} &= \frac{\eta^2}{n} \sum_l [r_{\mu l} v^X_{il} + r_{\mu l} (m^X_{il})^2 - g_{\mu l}^2 v^X_{il}] - \frac{\eta^4}{n^2} \sum_{\mu'} \sum_{l_1 \neq l_2}  g_{\mu l_1} g_{\mu' l_1} g_{\mu' l_2} g_{\mu l_2} v^F_{\mu' i} v^X_{i l_1} v^X_{i l_2}, \\
        \lambda^F_{\mu i} &= -\frac{\eta}{\sqrt{n}} \sum_l g_{\mu l} m^X_{il} + m^F_{\mu i} [-\gamma^F_{\mu i} + \frac{\eta^2}{n} \sum_l r_{\mu l} v^X_{il} ],\\
        \gamma^X_{i l} &= \frac{\eta^2}{n} \sum_\mu [r_{\mu l} v^F_{\mu i} + r_{\mu l} (m^F_{\mu i})^2 - g_{\mu l}^2 v^F_{\mu i}] - \frac{\eta^4}{n^2} \sum_{l'} \sum_{\mu_1 \neq \mu_2} g_{\mu_1 l} g_{\mu_2 l} g_{\mu_2 l'} g_{\mu_1 l'} v^F_{\mu_1 i} v^F_{\mu_2 i} v^X_{i l'}, \\
        \lambda^X_{i l} &= -\frac{\eta}{\sqrt{n}} \sum_\mu g_{\mu l} m^F_{\mu i} + m^X_{i l} [-\gamma^X_{i l} + \frac{\eta^2}{n}  \sum_\mu r_{\mu l} v^F_{\mu i} ].
    \end{dcases}
\end{align}

%% file: denoising.tex
\section{Probing our results in symmetric matrix denoising}\label{sec:denoising}

\subsection{The free entropy of factorization and denoising}\label{subsec:fentropy_denoising}

We consider the symmetric problem, specializing to the Gaussian setup in which the prior $P_X = \mathcal{N}(0,1)$, and the measurement channel $P_{\rm out} = \mathcal{N}(0,\Delta)$.
Then the partition function of both matrix factorization and denoising, defined in eqs.~\eqref{eq:def_gibbs_xx} and \eqref{eq:gibbs_denoising_xx}, is given by\footnote{Note that we added the diagonal terms $\mu = \nu$ with respect to eq.~\eqref{eq:def_gibbs_xx}. However, 
this number $\mathcal{O}(n)$ of terms does not affect the thermodynamic free entropy.}: 
\begin{align}
\label{eq:Zn_direct_quenched_start}
   \mathcal{Z}_{\bY,n} &= \int_{\bbR^{m \times n}} \frac{\rd \bX}{(2\pi)^{\frac{mn}{2}}(2\pi \Delta)^{\frac{m(m-1)}{4}}} e^{-\frac{1}{2} \underset{\mu,i}{\sum} X_{\mu i}^2} e^{-\frac{1}{4 \Delta}\underset{\mu,\nu}{\sum}\big(Y_{\mu \nu} - \frac{1}{\sqrt{n}} \underset{i}{\sum} X_{\mu i} X_{\nu i}\big)^2}.
\end{align}
The corresponding free entropy $\Phi_{\bY,n} \equiv (nm)^{-1}\ln \mathcal{Z}_{\bY,n}$ can be computed for a given signal $\bY$,  using random matrix techniques. 
We shall compute here its thermodynamic limit $\Phi_\bY \equiv \lim_{n\to\infty}\Phi_{\bY,n}$, in the generic case where the asymptotic eigenvalue distribution of $\bY / \sqrt{m}$ is well defined, and we denote it $ \rho_{\bY}$.
This computation relies on a sophisticated result on the celebrated Harish-Chandra-Itzykson-Zuber (HCIZ) integrals, 
that emerged in the theoretical physics literature \cite{matytsin1994large}, before being rigorously established \cite{guionnet2002large}.
This approach was first discussed in the context of extensive-rank matrix factorization in \cite{schmidt2018statistical}.
Here we elaborate on this work, filling in gaps in its arguments, and most importantly we work out how to evaluate analytically and numerically the corresponding equations, in order to provide explicit comparison to previous works and the PGY expansion proposed here. 
The recent and independent work \cite{barbier2021statistical} also derives asymptotic results in matrix factorization using Matytsin's results \cite{matytsin1994large}, but does not evaluate them or compare them to other approaches (beyond scalar denoising).

\myskip
\textbf{Beyond Wishart matrices --} While we focus on the denoising of Wishart matrices (since we are initially interested in the matrix factorization problem), 
we emphasize that one can easily generalize all the conclusions of Sections~\ref{subsec:fentropy_denoising}, \ref{subsec:mmse_denoising} and \ref{subsec:Bouchaud_denoising} to the more general 
denoising of a rotationally-invariant matrix.
This allows to compare our results with the general ones on symmetric matrix denoising of \cite{bun2016rotational,bun2017cleaning} 
in the case of Gaussian noise\footnote{Note that \cite{bun2016rotational,bun2017cleaning} also consider generic rotationally-invariant additive noise, beyond the Gaussian assumption. 
While this is beyond our scope here, adapting the spherical integral techniques presented hereafter to this more generic setting 
could be possible, and we leave it as an open direction of research.}.

\subsubsection{The case $\alpha\leq 1$}\label{subsubsec:alpha_leq_1}

We shall first study the case
 $\alpha \leq 1$ (i.e.\ $m \leq n$), so that the distribution of $\bS \equiv \bX \bX^\intercal / \sqrt{n m}$ is non-singular. 
We can rewrite eq.~\eqref{eq:Zn_direct_quenched_start} as:
\begin{align}\label{eq:Zn_direct_quenched_1}
\exp\{n m \Phi_{\bY,n}\} = \int P_S(\rd \bS) \frac{e^{-\frac{1}{4 \Delta}\underset{\mu,\nu}{\sum}\big(Y_{\mu \nu} - \sqrt{m} S_{\mu \nu}\big)^2}}{(2\pi \Delta)^{\frac{m(m-1)}{4}}},
\end{align}
in which $P_S(\rd \bS)$ is the Wishart distribution given by (with $\Gamma_p$ the multivariate $\Gamma$ function) :
\begin{align}
	P_S(\rd \bS) &\equiv \frac{m^{\frac{mn}{2}}}{(2\sqrt{\alpha})^{\frac{mn}{2}} \Gamma_m(n/2)} (\det \bS)^{\frac{n - m - 1}{2}} e^{-\frac{\sqrt{nm}}{2} \mathrm{Tr} \, \bS} \rd \bS.
\end{align}
We can decompose the integration in eq.~\eqref{eq:Zn_direct_quenched_1} on the eigenvalues and eigenvectors of $\bS = \bO \bL \bO^\intercal$, where the diagonal matrix $\bL$ contains the $m$ eigenvalues $\{l_\mu\}$.
This change of variables yields a combinatorial factor, that can be found e.g.\ using eq.~(C.4) of \cite{nicolaescu2014complexity} or Propositions~4.1.1 and 4.1.14 of \cite{anderson2010introduction}.
All in all, eq.~\eqref{eq:Zn_direct_quenched_1} turns into:
\begin{align}\label{eq:Zn_direct_quenched_2}
\exp\{n m \Phi_{\bY,n} \}&= \frac{m^{\frac{mn}{2}} e^{-\frac{1}{4 \Delta} \sum_{\mu,\nu} Y_{\mu \nu}^2}}{(2\sqrt{\alpha})^{\frac{mn}{2}} \Gamma_m(n/2) (2\pi \Delta)^{\frac{m(m-1)}{4}}} \times \frac{\pi^{\frac{m(m+1)}{4}}}{2^{m/2}\Gamma(m+1) \prod_{\mu=1}^m \Gamma(\mu/2) }  \\ 
& \times \int_{\bbR_+^m} \rd \bL \prod_{\mu < \nu} |l_\mu - l_\nu| e^{-\frac{m}{2} \sum\limits_{\mu=1}^m \big(\frac{l_\mu}{\sqrt{\alpha}} + \frac{l_\mu^2}{2 \Delta}\big) + \frac{n-m-1}{2}\sum\limits_{\mu=1}^m \ln l_\mu} \int_{\mathcal{O}(m)} \mathcal{D}\bO \, e^{\frac{\sqrt{m}}{2 \Delta} \mathrm{Tr}[\bY \bO \bL \bO^\intercal]} . \nonumber
\end{align}
The last integral is given by the large-$n$ limit of the extensive-rank HCIZ integral \cite{matytsin1994large,guionnet2002large}:
\begin{align*}
	\int_{\mathcal{O}(m)} \mathcal{D}\bO \, e^{\frac{\sqrt{m}}{2 \Delta} \mathrm{Tr}[\bY \bO \bL \bO^\intercal]} \simeq \exp\Big\{\frac{m^2}{2} I_\Delta[\rho_{\bY},\rho_\bS]\Big\}.
\end{align*}
The function $I_\Delta$ depends only on the asymptotic eigenvalue distribution  $\rho_{\bY}$ of $\bY / \sqrt{m}$, and on the asymptotic distribution of the eigenvalues of $\bS$.
In principle, using the Laplace method, the result should be expressed as a supremum over the asymptotic distribution of these eigenvalues, and one would need to assume the existence of this asymptotic distribution (this is justified as there is a number $\mathcal{O}(n)$ of eigenvalues, and we consider asymptotics in the scale $\exp\{\Theta(n^2)\}$).
It turns out that in our case we do not need to perform this complicated optimization on the spectrum of $\bS$.
This is a consequence of \emph{Bayes optimality}.
Indeed, $\rho_\bS$ is the limit of the eigenvalue distribution of $\bS$ under the posterior measure. 
But by the Nishimori identity \cite{nishimori2001statistical}, for any bounded function $\phi$ we have:
\begin{align}\label{eq:nishimori_matytsin}
	\EE \Big\langle \int \phi(\lambda) \rho_\bS(\rd \lambda) \Big\rangle = \EE \int \phi(\lambda) \rho_\bS^\star(\rd \lambda),
\end{align}
in which $\rho_\bS^\star$ is the limit eigenvalue distribution of the ground-truth signal $\bS^\star$.
In other terms, the Nishimori identity implies that the eigenvalue distribution of $\bS$ under the distribution of eq.~\eqref{eq:Zn_direct_quenched_1} concentrates on the Marchenko-Pastur law. 

\medskip\noindent
Therefore the function $I_\Delta$ depends on two known spectra, $\rho_\bS$ which is the Marcenko-Pastur spectrum of the matrix $\bX^\star (\bX^\star)^\intercal/\sqrt{nm}$, and $\rho_\bY$ which is the spectrum of the signal, equal to a Wishart matrix plus a Gaussian noise. 
We moreover have, from \cite{matytsin1994large,guionnet2002large}:
\begin{align}\label{eq:def_matytsin_limit}
	I_\Delta[\rho_{\bY},\rho_\bS] &= \frac{\ln \Delta}{2} - \frac{3}{4} + \frac{1}{2 \Delta} \Big[\int \rho_{\bY} (\rd x) x^2 + \int \rho_\bS(\rd x) x^2 \Big] - \frac{1}{2} \int \rho_{\bY}(\rd x) \rho_{\bY}(\rd y) \ln |x-y| \\ 
	& - \frac{1}{2}\int \rho_\bS(\rd x) \rho_\bS(\rd y) \ln |x-y| -\frac{1}{2} \Big\{\int_0^\Delta \rd t \int \rd x \rho(x,t) \Big[\frac{\pi^2}{3}\rho(x,t)^2 + v(x,t)^2\Big] \Big\}\ , \nonumber
\end{align}
where the function $f(x,t) \equiv v(x,t) + i \pi \rho(x,t)$ satisfies a complex Burgers' equation 
with prescribed boundary conditions:
\begin{align}\label{eq:burgers_complex}
   \begin{dcases}
        \partial_t f + f \partial_x f &= 0, \\
        \rho(x,t=0) &= \rho_\bS(x), \\ 
        \rho(x,t = \Delta) &= \rho_{\bY}(x).
   \end{dcases}
\end{align}
We now evaluate the thermodynamic limit of the free entropy of eq~\eqref{eq:Zn_direct_quenched_2}, when $m,n \to \infty$ with a fixed limit of $\alpha=m/n$ :
\begin{align}
   \Psi[\rho_\bY,\rho_\bS] \equiv \lim_{n\to\infty} \Phi_{\bY,n} = \Psi_1 + \Psi_2 +\Psi_3,
\end{align}
where
\begin{align}
  \begin{dcases}
    \Psi_1 &\equiv\lim_{m \to \infty}\Big[\frac{1}{2} \ln m - \frac{1}{nm} \ln \Gamma_m(n/2) - \frac{1}{nm} \sum_{\mu=1}^m \Gamma(\mu/2)\Big], \\
    \Psi_2 &\equiv - \frac{\ln 2}{2} - \frac{\ln \alpha}{4} - \frac{3 \alpha}{8} - \frac{\alpha}{4} \ln 2, \\
    \Psi_3 &\equiv  \frac{\alpha}{2} \int \rd x \rd y \; \rho_\bS(x)\rho_\bS(y) \ln|x-y|- \frac{\alpha}{2} \int \rd x \; \rho_\bS(x) \Big[\frac{x}{\sqrt{\alpha}}+\frac{x^2}{2\Delta}\Big] \\ 
    & +\frac{1-\alpha}{2} \int dx \; \rho_\bS(x) \ln x + \frac{\alpha}{2}I_\Delta[\rho_\bY,\rho_\bS].
  \end{dcases}
\end{align}
$\Psi_1$ can be computed using the asymptotic expansion
\begin{align}
	\frac{1}{nm} \ln \Gamma_m(n/2)&= \frac{\alpha}{4} \ln \pi + \frac{1}{nm} \sum_{\mu=1}^m \ln \Gamma \Big(\frac{n + 1 - \mu}{2}\Big) + \smallO_n(1) ,
\end{align}
and Stirling's formula $\ln \Gamma(z) = z \ln z - z + \mathcal{O}(\ln z)$ as $z \to \infty$. We get:
\begin{align}
	&\frac{1}{2} \ln m - \frac{1}{nm} \ln \Gamma_m(n/2) - \frac{1}{nm} \sum_{\mu=1}^m \ln \Gamma(\mu/2) \nonumber \\ 
	&= \frac{\alpha +(\alpha -1)^2 \ln (1-\alpha )-(\alpha -2) \alpha  \ln (\alpha )}{4 \alpha } - \frac{\alpha}{4} \ln \pi + \frac{1 + \ln 2}{2}  + \smallO_n(1).
\end{align}
Collecting all pieces, we get the final expression for the free entropy density of extensive-rank symmetric matrix factorization in the Gaussian setting: 
\begin{align}
\label{eq:phi_quenched_alpha_leq1}
	\Psi[\rho_\bY,\rho_\bS] &=\frac{1}{8} \Big\{\frac{2 (\alpha -1)^2 \ln (1-\alpha )}{\alpha }-\alpha  (3+\ln 4)+2 \ln \alpha-2 \alpha  \ln \pi \alpha +6\Big\}
\nonumber \\
	&+ 	 \frac{\alpha}{4} \int \rho_\bS(\rd x) \rho_\bS(\rd y) \ln |x - y| - \frac{\sqrt{\alpha}}{2} \int \rho_\bS(\rd x) x - \frac{\alpha - 1}{2} \int \rho_\bS(\rd x) \ln x  \nonumber \\ 
	& - \frac{\alpha}{4} \int \rho_\bY(\rd x) \rho_\bY(\rd y) \ln |x - y| -\frac{\alpha}{4} \int_0^\Delta \rd t \int \rd x \rho(x,t)  \Big[\frac{\pi^2}{3}\rho(x,t)^2 + v(x,t)^2\Big]  \ .
\end{align}
As stated above, the  function $f(x,t)=v(x,t)+i\pi \rho(x,t)$ satisfies the complex Burgers' equation~\eqref{eq:burgers_complex}. 

\myskip
It turns out that we can greatly simplify eq.~\eqref{eq:phi_quenched_alpha_leq1}. 
This simplification is a consequence of two results, which we state hereafter.
While the first of these two results is not directly used in the simplification, we believe it also bears independent interest that allows to understand 
the reasons allowing for such a simplification.

\myskip
The first one is the analytical solution to eq.~\eqref{eq:burgers_complex}, using the Dyson Brownian motion. 
It was given in \cite{schmidt2018statistical}, but incompletely justified.
\begin{result}[Matytsin's solution and Dyson Brownian motion]\label{result:matytsin_solution_dyson}
	\noindent
    Recall that $\bY/\sqrt{m} = \bS + \sqrt{\Delta} \bZ / \sqrt{m}$, with $\bZ$ a Wigner matrix. For any $\bS$,
    the solution to eq.~\eqref{eq:burgers_complex} is given by the Stieltjes transform\footnotemark[1] \ of the Dyson brownian motion 
    $\bY(t)/\sqrt{m} \overset{\rd }{=} \bS + \sqrt{t} \bZ / \sqrt{m}$, with $t \in [0,\Delta]$.
    More precisely, if $g_{\bY(t)}(z) \equiv \int \rho_{\bY(t)}(\rd t) / (t-z)$, the solution is given by:
    \begin{align}\label{eq:def_rho_v_solution_dyson}
        \begin{dcases}
            \rho(x,t) &= \frac{1}{\pi}\lim_{\epsilon \downarrow 0} \Big\{ \mathrm{Im}[g_{\bY(t)}(x + i \epsilon)] \Big\}, \\
            v(x,t) &= \lim_{\epsilon \downarrow 0} \Big\{-\mathrm{Re}[g_{\bY(t)}(x + i \epsilon)]\Big\}.
        \end{dcases}
    \end{align}
\end{result}
\footnotetext[1]{For more details on the Stieltjes transform, see Appendix~\ref{sec_app:rmt}.}
\textbf{Remark --}
Note that this result, while not stated exactly as here, was used before in the mathematics literature, see e.g.\ \cite{guionnet2004first,menon2017complex}.
Further, these works relate the distribution $\rho(x,t)$ that appears in the general formula of eq.~\eqref{eq:def_matytsin_limit} of $I_\Delta[\rho_\bA,\rho_\bB]$ to the one of a \emph{Dyson Brownian Bridge}. 
More precisely, \cite{guionnet2004first} shows that \emph{there exists} a joint distribution of two random matrices $\bA'$ and $\bB'$ such that 
$\rho_A$ is the LSD of $\bA'$, $\rho_B$ the one of $\bB'$, and
$\rho(t)$ is the asymptotic eigenvalue distribution of 
$\bX(t) = (1-t/\Delta) \bA' + t \bB' / \Delta + \sqrt{t(1-t/\Delta)} \bW / \sqrt{m}$, with $\bW$ a Gaussian Wigner matrix independent of $\bA'$ and $\bB'$.
However, the joint distribution of $(\bA',\bB')$ is unknown in general: Result~\ref{result:matytsin_solution_dyson} shows that 
when $\bB = \bA + \sqrt{\Delta} \bZ / \sqrt{m}$, one can take $\bA' = \bA$ and $\bB' = \bB$\footnote[2]{Indeed, one has then $\bX(t) = \bA + [t/\sqrt{\Delta} \bZ +  \sqrt{t(1-t/\Delta)} \bW]/\sqrt{m}$. 
The last term of this equation is the sum of two independent Gaussian matrices, therefore it is itself Gaussian, with variance $t^2/\Delta + t (1-t/\Delta) = t$. 
In distribution one has thus $\bX(t) \deq \bA + \sqrt{t} \bW/\sqrt{m}$, i.e.\ $\bX(t)$ is the Dyson Brownian motion starting in $\bA$.}.
To the best of our knowledge, the only other case in which this joint distribution is known is when $\bA$ and $\bB$ are independent Gaussian Wigner matrices with arbitrary variances \cite{bun2014instanton}.

\myskip
Our second result allows for a direct simplification of eq.~\eqref{eq:phi_quenched_alpha_leq1}. It can be stated as follows: 
\begin{result}\label{result:gaussian_integral_matytsin}
\noindent
    For any well-behaved density $\rho_A$ (corresponding to a random matrix $\bA$), we have with $I = I_{\Delta = 1}$ the function of eq.~\eqref{eq:def_matytsin_limit}:
    \begin{align}\label{eq:result_gaussian_integral_matytsin}
    \sup_{\rho_L} \Big\{\frac{1}{2} \int \rho_L^{\otimes 2}(\rd x, \rd y) \ln |x-y| - \frac{1}{4} \int \rho_L(\rd x)\, x^2 + \frac{1}{2}I[\rho_L, \rho_A]\Big\}
    &= - \frac{3}{8} + \frac{1}{4} \int \rho_A(\rd x) \, x^2.
    \end{align}
    Equivalently, if $\rho(x,t)$ solves the Euler-Matytsin equation~\eqref{eq:burgers_complex} between $\rho_A$ and $\rho_L$, we have:
    \begin{align}\label{eq:result_gaussian_integral_matytsin_2}
     &\sup_{\rho_L} \Big\{\int \rho_L^{\otimes 2}(\rd x, \rd y) \ln |x-y| - \int_0^1 \rd t \int \rd x \, \rho(x,t) \big[\frac{\pi^2}{3} \rho(x,t)^2 + v(x,t)^2\big] \Big\}\\ 
      &= \int \rho_A^{\otimes 2}(\rd x, \rd y) \ln |x-y|. \nonumber
    \end{align}
    Moreover, in both previous equations the supremum is reached by the \emph{additive free convolution} $\rho_L^\star = \rho_A \boxplus \sigma_\mathrm{s.c.}$ \cite{voiculescu1986addition,anderson2010introduction}, i.e.\ the eigenvalue density of $\bA + \bZ / \sqrt{m}$, which is the Dyson Brownian motion at time $t = 1$. 
\end{result}
In Appendix~\ref{sec_app:matytsin_dyson} we give a complete justification of Results~\ref{result:matytsin_solution_dyson} and \ref{result:gaussian_integral_matytsin}.
Using Result~\ref{result:gaussian_integral_matytsin}, and the fact that $\bY/\sqrt{m} = \bS + \sqrt{\Delta} \bZ / \sqrt{m}$, eq.~\eqref{eq:phi_quenched_alpha_leq1} takes the final form
:
\begin{align}\label{eq:phi_alpha_leq1_simplified}
    \Psi[\rho_\bY,\rho_\bS] &=  \frac{1}{8} \Big\{\frac{2 (\alpha -1)^2 \ln (1-\alpha )}{\alpha }-\alpha  (3+\ln 4)+2 \ln \alpha-2 \alpha  \ln \pi \alpha +6\Big\} - \frac{\sqrt{\alpha}}{2} \int \rho_\bS(\rd x) x
  \nonumber \\
	&  - \frac{\alpha - 1}{2} \int \rho_\bS(\rd x) \ln x + 	 \frac{\alpha}{2} \int \rho_\bS(\rd x) \rho_\bS(\rd y) \ln |x - y|  
   - \frac{\alpha}{2} \int \rho_\bY(\rd x) \rho_\bY(\rd y) \ln |x - y|.
\end{align}
\textbf{Remark --} Note that the derivation of eq.~\eqref{eq:result_gaussian_integral_matytsin} (presented in Appendix~\ref{sec_app:matytsin_dyson}) does not rely on the 
expression of $I_\Delta[\rho_L,\rho_A]$ derived in \cite{matytsin1994large,guionnet2002large} that involved Burgers' equation, but is a simple consequence of the relation of the HCIZ integral to the large deviations of the Dyson Brownian motion \cite{guionnet2002large,guionnet2004first,bun2014instanton,menon2017complex}.
One can therefore recover eq.~\eqref{eq:phi_alpha_leq1_simplified} without appealing to the hydrodynamical formalism of Matytsin.

\subsubsection {The case  $\alpha > 1$}\label{subsubsec:alpha_geq_1}

When $\alpha > 1$, the calculation is very similar, but one considers the limit spectral density of $\bS' \equiv \bX^\intercal \bX / \sqrt{nm}$, which is now a full-rank 
matrix, rather than the singular matrix $\bS = \bX \bX^\intercal / \sqrt{nm}$.
Denoting  the asymptotic spectral densities of these matrices by $\rho_\bS$ and $\rho_{\bS'}$,  we have:
\begin{align*}
    \rho_\bS(x) = \alpha^{-1} \rho_\bS'(x) + (1-\alpha^{-1}) \delta(x).
\end{align*}
The issue that arises in the calculation is that $\rho_\bS(x)$ has a singular component around $x = 0$. 
In order to regularize this singularity, we consider an arbitrarily small $\epsilon > 0$, and we replace in the asymptotic result: 
\begin{align}\label{eq:def_regularization_alpha_geq1}
  \rho_\bS = \alpha^{-1} \rho_\bS'(x) + (1-\alpha^{-1}) \delta(x) \to \rho_\bS^{(\epsilon)} = (\alpha^{-1} \rho_\bS'(x) + (1-\alpha^{-1}) \delta(x)) \boxplus \sigma_\mathrm{s.c.}^{\epsilon}
\end{align}
Here, $\sigma_\mathrm{s.c.}^{\epsilon}$ is the asymptotic eigenvalue distribution of $\sqrt{\epsilon/m} \bZ$ with $\bZ$ a Gaussian Wigner matrix\footnote{In simpler terms, $\sigma_\mathrm{s.c.}^{\epsilon}$ is a Wigner semi-circle law extending from $-2\sqrt{\epsilon}$ to  $-2\sqrt{\epsilon}$.},
and $\boxplus$ is the free convolution \cite{voiculescu1986addition,anderson2010introduction}.
Assuming that the limit eigenvalue of $\bS$ is given by $\rho_\bS^{(\epsilon)}$, 
we perform calculations completely similar to what we did for $\alpha \leq 1$, and we reach, at leading order in $\epsilon$:
\begin{align}\label{eq:phi_quenched_alpha_geq1_partial}
  &\Psi[\rho_\bY,\rho_\bS] = C(\alpha) + \frac{1}{2\alpha} \int \rho_\bS'(\rd x)\rho_\bS'(\rd y) \ln |x-y| - \frac{\alpha}{4} \int (\rho_\bS^{(\epsilon)}(\rd x)\rho_\bS^{(\epsilon)}(\rd y) + \rho_\bY(\rd x)\rho_\bY(\rd y)) \ln |x-y| \nonumber \\ 
  &+ \frac{1-\alpha^{-1}}{2} \int \rho_\bS'(\rd x) \ln x -\frac{1}{2\sqrt{\alpha}} \int \rho_\bS'(\rd x) x  - \frac{\alpha}{4} \int_0^\Delta \rd t \int \rd x \rho^{(\epsilon)}(x,t) \Big[\frac{\pi^2}{3}\rho^{(\epsilon)}(x,t)^2 + v^{(\epsilon)}(x,t)^2\Big],
\end{align}
with the constant $C(\alpha)$ given by: 
\begin{align}\label{eq:def_constant_C_alpha_geq1}
  C(\alpha) = \frac{2 (\alpha -1)^2 \ln (\alpha -1)+\alpha  (-3 \alpha -2 \alpha  \ln (2 \pi \alpha) + 2 \ln \alpha+6)}{8 \alpha }.
\end{align}
We also denoted $v^{(\epsilon)}(x,t)$ and 
$\rho^{(\epsilon)}(x,t)$ to make explicit the fact that the boundary condition at $t = 0$ is determined by $\rho_\bS^{(\epsilon)}$.
Exactly as in the case $\alpha \leq 1$ we make use of Result~\ref{result:gaussian_integral_matytsin}, which allows to reach the simplification:
\begin{align}\label{eq:phi_quenched_alpha_geq1}
  \nonumber
  \Psi[\rho_\bY,\rho_\bS] &= C(\alpha) + \frac{1}{2\alpha} \int \rho_\bS'(\rd x)\rho_\bS'(\rd y) \ln |x-y| + \frac{1-\alpha^{-1}}{2} \int \rho_\bS'(\rd x) \ln x -\frac{1}{2\sqrt{\alpha}} \int \rho_\bS'(\rd x) x \\ 
  &- \frac{\alpha}{2} \int \rho_\bY(\rd x)\rho_\bY(\rd y) \ln |x-y|. 
\end{align}

\subsection{The minimum mean squared error of denoising}\label{subsec:mmse_denoising}

The information-theoretic minimum mean squared error (MMSE) for denoising is defined as:
\begin{align}\label{eq:def_y_mmse_denoising}
	\mathrm{MMSE}(\Delta) &\equiv \EE \Big[\frac{1}{m} \sum_{\mu , \nu} \Big(S^\star_{\mu \nu} - \langle S_{\mu \nu} \rangle\Big)^2  \Big],
\end{align}
where $\bS^\star = \bX^\star (\bX^\star)^\intercal / \sqrt{nm}$ is the original ``ground-truth'' matrix, and the observation is  $\bY = \sqrt{m} \bS^\star + \sqrt{\Delta} \bZ$. 
Note that adding the diagonal terms $\mu = \nu$ in this definition does not change the asymptotic MMSE, cf.\ Appendix~\ref{subsec:mmse_diagonal}.
This definition satisfies, for any $\alpha$:
\begin{align}
\lim_{\Delta \to \infty} \mathrm{MMSE}(\Delta) = 1.
\end{align}
As is usually the case for Gaussian channels \cite{guo2005mutual}, the $\mathrm{MMSE}$ is simply given by a derivative of the asymptotic free entropy:
\begin{align}\label{eq:I_mmse_denoising}
  \mathrm{MMSE}(\Delta) = \Delta -  \frac{4}{\alpha} \frac{\partial \Psi[\rho_\bY,\rho_\bS] }{\partial \Delta^{-1}}.
\end{align}
A proof of eq.~\eqref{eq:I_mmse_denoising} is given in Appendix~\ref{subsec_app:I-MMSE}.
We can now compute the MMSE explicitly using eq.~\eqref{eq:I_mmse_denoising} and eqs.~\eqref{eq:phi_quenched_alpha_leq1}, \eqref{eq:phi_quenched_alpha_geq1}. 
Therefore the MMSE takes the form:
\begin{align}\label{eq:mmse_denoising_final_expr}
  \mathrm{MMSE}(\Delta) &= \Delta + \frac{4}{\alpha} \Delta^2 \frac{\partial}{\partial \Delta}\{\Psi[\rho_\bY,\rho_\bS]\} =\Delta - 2 \Delta^2 \frac{\partial}{\partial \Delta} \int \rho_\bY(\rd x) \rho_\bY(\rd y) \ln |x - y|.
\end{align}
Recall that $\rho_\bY$ is the asymptotic eigenvalue density of $\bY /\sqrt{m}$. 
As we show below, this can be computed using free probability techniques.

\myskip
\textbf{A ``miraculously simple'' solution --} Eq.~\eqref{eq:mmse_denoising_final_expr} is particularly simple, as it only depends on the eigenvalue distribution of the \emph{observations} $\bY$.
This was noticed in \cite{bun2016rotational} for an important class of rotationally-invariant estimators.
Here, this arises as a consequence of the fact that the only dependency in $\Delta$ in the free entropy of eqs.~\eqref{eq:phi_quenched_alpha_leq1}, \eqref{eq:phi_quenched_alpha_geq1} is through the 
eigenvalue density of $\bY$.

\myskip
\textbf{Numerical evaluation of the MMSE --} 
In order to compute the Stieltjes transform $g_\bY(z)$ (and from it $\rho_\bY$ and $v_\bY$ by eq.~\eqref{eq:def_rho_v_solution_dyson}), 
we can use free probability techniques. 
Indeed, we know the $\mathcal{R}$-transform of $\bY/\sqrt{m}$, since we have
$\bY / \sqrt{m}= \bX \bX^\intercal / \sqrt{nm} + (\Delta/m)^{1/2} \bZ$:
\begin{align}
  \mathcal{R}_\bY(s) &= g_{\bY}^{-1}(-s) - \frac{1}{s} = \frac{1}{\sqrt{\alpha}(1-\sqrt{\alpha}s)} + \Delta s.
\end{align}
This implies that $s = - g_\bY(z)$ can be found as the solution to the algebraic equation:
\begin{align}\label{eq:algebraic_g}
	s^3 - \frac{\sqrt{\alpha}z+\Delta}{\sqrt{\alpha}\Delta} s^2 + \frac{1}{\sqrt{\alpha}\Delta} (z + \sqrt{\alpha} - \alpha^{-1/2}) s - \frac{1}{\sqrt{\alpha}\Delta} &= 0.
\end{align}
Eq.~\eqref{eq:algebraic_g} can be solved numerically.
We then compute the derivative of the logarithmic potential in eq.~\eqref{eq:mmse_denoising_final_expr}
using finite differences, which allows us to efficiently access the MMSE numerically.

\subsection{Bayes-optimal estimator for denoising}\label{subsec:Bouchaud_denoising}

\subsubsection{From asymptotic results to estimation}\label{subsubsec:matytsin_estimation}

While the derivations presented in Sections~\ref{subsec:fentropy_denoising} and \ref{subsec:mmse_denoising} 
allow to characterize the performance of the Bayes-optimal estimator, they did not provide an actual estimation algorithm matching this performance.
In this section, we work out such an estimator using the asymptotic limits described above.
From Bayesian statistics, the Bayes-optimal estimator of $\bS \equiv \bX \bX^\intercal / \sqrt{nm}$ 
is $\hat{\bS} \equiv \langle \bS \rangle$, with $\langle \cdot \rangle$ the average under the posterior distribution 
of eq.~\eqref{eq:Zn_direct_quenched_1}.

\medskip\noindent
Starting again from eq.~\eqref{eq:Zn_direct_quenched_1}, we have: 
\begin{align}
  \label{eq:estimator_from_fentropy}
  \begin{dcases}
  \langle S_{\mu \nu} \rangle &= \frac{Y_{\mu \nu}}{\sqrt{m}} - \Delta n \frac{\partial \Phi_{\bY,n}}{\partial (Y_{\mu \nu} / \sqrt{m})}, \\
  \langle S_{\mu \mu} \rangle &= \frac{Y_{\mu \mu}}{\sqrt{m}} - 2 \Delta n \frac{\partial \Phi_{\bY,n}}{\partial (Y_{\mu \mu} / \sqrt{m})}.
  \end{dcases}
\end{align}
Importantly, our results of eqs.~\eqref{eq:phi_quenched_alpha_leq1} and \eqref{eq:phi_quenched_alpha_geq1}
imply that, as $n \to \infty$, $\Phi_{\bY,n}$ depends on $\bY$ solely via the spectrum of $\bY/\sqrt{m}$.
This means that we can write eq.~\eqref{eq:estimator_from_fentropy} at leading order as $n,m \to \infty$ as:
\begin{align}
  \label{eq:estimator_from_fentropy_2}
  \begin{dcases}
  \langle S_{\mu \nu} \rangle &= \frac{Y_{\mu \nu}}{\sqrt{m}} - \Delta n \sum_{\rho=1}^m \frac{\partial \Phi_{\bY,n}}{\partial y_\rho} \frac{\partial y_\rho}{\partial(Y_{\mu \nu}/\sqrt{m})}, \\
  \langle S_{\mu \mu} \rangle &= \frac{Y_{\mu \mu}}{\sqrt{m}} - 2 \Delta n \sum_{\rho=1}^m \frac{\partial \Phi_{\bY,n}}{\partial y_\rho} \frac{\partial y_\rho}{\partial(Y_{\mu \mu}/\sqrt{m})},
  \end{dcases}
\end{align}
in which $\{y_\rho\}_{\rho=1}^m$ are the eigenvalues of $\bY/\sqrt{m}$.
Let us denote $\{\bv^\rho\}$ the corresponding eigenvectors, and fix $\mu_0 < \nu_0$.
By the Hellman-Feynman theorem, if we perturb the matrix $\bY/\sqrt{m}$ as:
\begin{align}
    Y_{\mu \nu} \to Y^\gamma_{\mu\nu} = Y_{\mu\nu}+\gamma \sqrt{m}\delta_{\mu\mu_0}\delta_{\nu\nu_0},
\end{align}
for a small $\gamma \in \bbR$, 
then the corresponding transformation of an eigenvalue $y_\rho$ is: 
\begin{align}\label{eq:hellman_feynman_off_diag}
    \frac{dy_\rho}{d\gamma}= \frac{\partial y_\rho}{\partial(Y_{\mu_0 \nu_0} / \sqrt{m})}
    = \sum_{\mu\nu} (\delta_{\mu\mu_0}\delta_{\nu\nu_0} + \delta_{\mu\nu_0}\delta_{\nu\mu_0})v^\rho_\mu v^\rho_\nu 
    = 2 v^\rho_{\mu_0} v^\rho_{\nu_0}.
\end{align}
In the same way, we obtain for $\mu_0 = \nu_0$:
\begin{align}\label{eq:hellman_feynman_diag}
    \frac{\partial y_\rho}{\partial(Y_{\mu_0 \mu_0} / \sqrt{m})}
    = \sum_{\mu\nu} \delta_{\mu\mu_0}\delta_{\nu\mu_0} v^\rho_\mu v^\rho_\nu 
    = (v^\rho_{\mu_0})^2.
\end{align}
Plugging eqs.~\eqref{eq:hellman_feynman_off_diag} and \eqref{eq:hellman_feynman_diag}
in eq.~\eqref{eq:estimator_from_fentropy_2} yields, for any $\mu,\nu$:
\begin{align}
  \label{eq:estimator_from_fentropy_3}
  \langle S_{\mu \nu} \rangle &= \frac{Y_{\mu \nu}}{\sqrt{m}} - 2 \Delta n \sum_{\rho=1}^m \frac{\partial \Phi_{\bY,n}}{\partial y_\rho} v_\mu^\rho v_\nu^\rho 
  = \sum_{\rho=1}^m \Big[y_\rho -  2 \Delta n \frac{\partial \Phi_{\bY,n}}{\partial y_\rho}\Big] v_\mu^\rho v_\nu^\rho.
\end{align}
Note that eqs.~\eqref{eq:phi_alpha_leq1_simplified}, \eqref{eq:phi_quenched_alpha_geq1} have the following form at finite $n$:
\begin{align}
    \Phi_{\bY,n} &= F_n(\rho_\bS) - \frac{1}{2 m n} \sum_{\mu \neq \nu} \ln |y_\mu - y_\nu| + \smallO_n(1),
\end{align}
so that 
\begin{align}
    n \frac{\partial \Phi_{\bY,n}}{\partial y_\mu} &= - \frac{1}{m} \sum_{\nu(\neq \mu)} \frac{1}{y_\mu - y_\nu} + \smallO_n(1) = \mathrm{P.V.} \int \rd y \, \rho_\bY(y) \frac{1}{y - y_\mu} + \smallO_n(1),
\end{align}
with $\mathrm{P.V.}$ the principal value of the integral. Stated differently, the Bayes-optimal estimator has the form:
\begin{align}\label{eq:expression_RIE}
  \begin{dcases}
  \hat{\bS} &= \sum_{\mu=1}^m \hat{\xi}_\mu \bv^\mu (\bv^\mu)^\intercal, \\
	\hat{\xi}_\mu &= y_\mu - 2 \Delta v_{\bY}(y_\mu,\Delta).
  \end{dcases}
\end{align}
In eq.~\eqref{eq:expression_RIE} we introduced the function:
\begin{align}\label{eq:def_vY}
  v_{\bY}(x,\Delta) \equiv - \lim_{\epsilon \downarrow 0} \mathrm{Re}[g_{\bY}(x + i\epsilon)],
\end{align}
with $g_\bY(z)$ the Stieltjes transform of $\bY/\sqrt{m}$, cf.\ Appendix~\ref{sec_app:rmt}.

\myskip
An important consequence of eq.~\eqref{eq:expression_RIE} is that the optimal estimator $\hat{\bS}$ is 
\emph{diagonal in the eigenbasis of the observation $\bY$}. 
Such estimators are usually called \emph{rotationally-invariant estimators} (RIE). 
Note that one can argue that the optimal estimator is a RIE in this problem without going through the derivation above, but simply as a consequence of the rotation-invariance of the posterior distribution 
of eq.~\eqref{eq:gibbs_denoising_xx}, which was noticed in \cite{bun2017cleaning}. 
The optimal RIE (in the $L^2$ sense) for the denoising of an extensive-rank matrix corrupted by Gaussian noise has been worked out in \cite{bun2016rotational}, in which the authors 
exactly obtained the expression of eq.~\eqref{eq:expression_RIE}.
We therefore found back their main result using our asymptotic derivation\footnote{As noticed in \cite{bun2016rotational}, eq.~\eqref{eq:expression_RIE} is in a sense a ``miracle'' similar to the one we described for the MMSE in Section~\ref{subsec:mmse_denoising}: the eigenvalues $\{\hat{\xi}_\mu\}$ of the optimal estimator are expressed solely as a function of the 
spectrum of the observations $\bY$.}.

\myskip 
As we mentioned at the beginning of Section~\ref{sec:denoising}, one can easily see that none of the results on the MMSE or the optimal denoiser are specific to the denoising of Wishart matrices, but can be easily generalized to 
any rotationally-invariant matrix, so that our derivation 
can be compared to the one of \cite{bun2016rotational} based on a direct minimization of the $L^2$ error, in the case of Gaussian additive noise. 
As we mentioned above, we leave for future work an extension of our derivation to completely generic additive rotationally-invariant noise (which is analyzed in \cite{bun2016rotational}).
Note that the mean squared error of the optimal estimator should be the MMSE of eq.~\eqref{eq:mmse_denoising_final_expr}.
We have checked this numerically (see Section~\ref{subsec:numerics}) for two types of signals.
Fig.~\ref{fig:comparison_denoising_orthogonal_wigner} shows that the two results match, as expected.

\subsubsection{Small-$\alpha$ expansion of the optimal denoiser}\label{subsubsec:small_alpha_rie_denoising}

We shall use the optimal estimator of eq.~\eqref{eq:expression_RIE} in order to test our PGY expansion for denoising in the Gaussian setting.
Recall that we are denoising a $m\times m$ Wishart matrix $\bS=\bX\bX^T/\sqrt{nm}$, where $\bX$ is $m\times n$, and $\alpha=n/m$.
We shall consider here an expansion of the denoiser in the overcomplete regime of small $\alpha= m/n$ (i.e.\ close to the limit where $\bS$ becomes Gaussian), and work out in the next section the small-$\alpha$ expansion within the PGY formalism.

\medskip\noindent
Note that in the PGY approach, we will not consider diagonal observations: we use $Y_{\mu \mu} = 0$. 
Equivalently, by rotation-invariance, we assume here that $\bY$ has asymptotically zero trace, which 
we can do without any loss of generality. Simply, we shift the observations by a deterministic term:
\begin{align}\label{eq:shifted_Y}
  \bY = \frac{1}{\sqrt{n}} \bX \bX^\intercal - \frac{1}{\sqrt{\alpha}} \Id_m + \sqrt{\Delta} \bZ. 
\end{align}
The computation detailed in Appendix~\ref{sec_app:small_alpha_rie} expands 
the Stieltjes transform of $\bY/\sqrt{m}$ and gives 
\begin{align}
	v_\bY(y) &= \frac{y}{2(\Delta+1)} + \frac{(\Delta +1-y^2)}{2 (\Delta +1)^3} \sqrt{\alpha} + \mathcal{O}(\alpha).
\end{align}
Using this expansion in the  optimal estimator of eq.~\eqref{eq:expression_RIE}, we get the result:
\begin{align}\label{eq:denoising_rie_order_2}
	\hat{\xi}_\mu &=	\frac{1}{\Delta + 1} y_\mu - \frac{\Delta }{(\Delta + 1)^2}\Big[1 - \frac{y_\mu^2}{\Delta + 1}\Big] \sqrt{\alpha} + \mathcal{O}(\alpha).
\end{align}
This expansion describes how every eigenvalue $y$ of $\sqrt{m}$ should be denoised, in the limit of small $\alpha$. The leading term,
$\hat{\xi}_\mu=y_\mu/(\Delta+1)$, is the standard expression for denoising a Wigner matrix, as expected. The term of order $\sqrt{\alpha}$ is the first non-trivial correction, we shall compare it to PGY-based denoising in the next section.

\subsection{Denoising using the PGY expansion}\label{subsec:pgy_denoising}

\subsubsection{PGY expansion for the denoising problem}\label{subsubsec:pgy_denoising_derivation}

In this section, we apply the PGY expansion formalism to the symmetric matrix denoising problem, for a Wishart matrix $\bS \equiv \bX \bX^\intercal / \sqrt{nm}$.
Recall that we observe $Y_{\mu \nu} \sim P_{\rm out}( \cdot | \sqrt{m} S_{\mu \nu})$.
We consider a general channel $P_{\rm out}$ (where the noise is applied independently to every component of the matrix), and a general prior, factorized over the components of the input matrix: $P_X(\bX)= \prod_{\mu,i} P_X(X_{\mu i})$.
Furthermore, we shall assume that the mean of $P_X$ is $0$.
Without loss of generality (as it can be absorbed in the channel), we also assume that the 
variance of the prior is $\EE_{P_X}[x^2] = 1$.

\medskip\noindent
The PGY expansion is very close to the computation that we did for the matrix factorization problem, except that,
in the denoising problem, we do not fix the first and second moments of the field $\bX$.  The free entropy thus takes a form which is simpler than the one of eq.~\eqref{eq:phi_extensive_xx_before_plefka}.
We get:
\begin{align}
	nm\Phi_{\bY,n} &= \sum_{\mu < \nu}\Big[-\omega_{\mu \nu} g_{\mu \nu} - \frac{b_{\mu \nu}}{2}\Big(-r_{\mu \nu} + g_{\mu \nu}^2\Big) \Big] + \ln \int P_H(\rd \bH) \ P_X(\rd \bx) \ e^{-S_{\rm eff}[\bx,\bH]},
\end{align}
in which we introduced an \emph{effective action}: 
\begin{align}
	S_{\rm eff}[\bx,\bH] &= \sum_{\mu < \nu} \Big[\omega_{\mu \nu} (iH)_{\mu \nu} - \frac{b_{\mu \nu}}{2} (iH)_{\mu \nu}^2\Big] + H_{\rm eff}[\bx,\bH] \ , 
\end{align}
and the (un-normalized) 
distribution and effective Hamiltonian:
\begin{align}
  \begin{dcases}
       P_H[\rd \bH] &\equiv \rd \bH \, \prod_{\mu < \nu} \int \frac{\rd \tilde{H}}{2 \pi} e^{i H_{\mu \nu} \tilde{H}} P_{\rm out}(Y_{\mu \nu}|\tilde{H}) = \prod_{\mu < \nu} \frac{1}{2 \pi} e^{-\frac{\Delta H_{\mu \nu}^2}{2} + i Y_{\mu\nu} H_{\mu \nu} } \ \rd H_{\mu \nu}, \\
       H_{\rm eff}[\bx,\bH] &\equiv \frac{1}{\sqrt{n}}\sum_{\mu < \nu} \sum_i (iH)_{\mu \nu} x_{\mu i} x_{\nu i}.
  \end{dcases}
\end{align}
We only impose the constraints $g_{\mu \nu} = - \langle (iH)_{\mu \nu}\rangle$ and $r_{\mu \nu} = - \langle (iH)_{\mu \nu}^2\rangle + g_{\mu \nu}^2$, using the Lagrange multipliers $\omega_{\mu\nu}$ and $b_{\mu\nu}$.

\medskip\noindent
The PGY expansion proceeds exactly as in Section~\ref{subsec:plefka_symmetric}:
for pedagogical reasons, we detail the PGY expansion specifically for the denoising problem in Appendix~\ref{sec_app:pgy_denoising}.
As shown there, the result for denoising can be obtained by using the PGY expansion for matrix factorization and imposing $m_{\mu i}=0$, $v_{\mu i}=1$, $\lambda_{\mu i}=0$, $\gamma_{\mu i}=0$.
The PGY free entropy for denoising at order $\eta^3$ takes the form:
\begin{align}
\label{eq:phi_pgy_order_3_denoising}
 nm\Phi_{\bY,n} &= 
        \sum_{\mu < \nu}\Big[-\omega_{\mu \nu} g_{\mu \nu} - \frac{b_{\mu \nu}}{2}\Big(-r_{\mu \nu} + g_{\mu \nu}^2\Big)  + \ln \int \rd z \, \frac{ e^{-\frac{1}{2 b_{\mu \nu}}(z - \omega_{\mu \nu})^2}}{\sqrt{2 \pi b_{\mu \nu}}} \, P_\mathrm{out}(Y_{\mu \nu}|z) \Big] \nonumber \\
       &+\frac{\eta^2}{2} \sum_{\mu<\nu} [g_{\mu\nu}^2-r_{\mu\nu}]
       +\frac{\eta^3}{6 n^{1/2}}
       \sum_{\substack{\mu_1,\mu_2,\mu_3 \\ \text{pairwise distinct}}} g_{\mu_1 \mu_2} g_{\mu_2 \mu_3} g_{\mu_3 \mu_1}.
\end{align}

\subsubsection{Small-$\alpha$ expansion of the PGY denoiser}\label{subsubsec:small_alpha_pgy_denoising}

We can now expand the denoising estimator predicted by the PGY expansion at order $3$ in $\alpha$ as we did for the optimal denoiser in Section~\ref{subsubsec:small_alpha_rie_denoising}.
We also remind what we noted above (see Section~\ref{subsubsec:higher_orders_extensive_xx}) that the PGY expansion at order $\eta^3$ corresponds to an expansion of the optimal denoiser in $\alpha$, up to order $\sqrt{\alpha}$.

\medskip\noindent
We start from the expression of the free entropy in eq.~\eqref{eq:phi_pgy_order_3_denoising}, evaluated at $\eta=1$.
Using the fact that the diagonal elements of $\bg$ are zero, the last term is just proportional to $\mathrm{Tr}[\bg^3]$.
Using that $\partial_{g_{\mu \nu}} \mathrm{Tr}[\bg^k] = 2 k (\bg^{k-1})_{\mu \nu}$, the TAP equations corresponding to eq.~\eqref{eq:phi_pgy_order_3_denoising} are:
\begin{align}\label{eq:TAP_denoising_order_3_full}
  \begin{dcases}
  b_{\mu \nu} = 1  \hspace{0.5cm} &; \hspace{0.5cm} 
  g_{\mu \nu} = \frac{Y_{\mu \nu} - \omega_{\mu \nu}}{\Delta + b_{\mu \nu}},  \\
  \omega_{\mu \nu} = -b_{\mu \nu} g_{\mu \nu} + g_{\mu \nu} + \frac{1}{\sqrt{n}} (\bg^2)_{\mu \nu} \hspace{0.5cm} &; \hspace{0.5cm}
  r_{\mu \nu} = \frac{1}{\Delta + b_{\mu \nu}}.
  \end{dcases}
\end{align}
Therefore this can be reduced to a single equation on $\bg$: 
\begin{align}\label{eq:TAP_order_3_denoising}
	\frac{\bY}{\sqrt{m}} &= (\Delta + 1) \frac{\bg}{\sqrt{m}} + \sqrt{\alpha} \Big(\frac{\bg}{\sqrt{m}}\Big)^2 - \frac{\sqrt{\alpha}}{m}\mathrm{Tr}\Big[\frac{\bg^2}{m}\Big] \Id_m.
\end{align}
Note that in eq.~\eqref{eq:TAP_order_3_denoising} we added a last term to impose the zero-trace condition on $\bg$ and 
$\bY$, since eq.~\eqref{eq:TAP_denoising_order_3_full} only describes the off-diagonal terms\footnote{Since these are rotationally-invariant matrices, having a zero trace is asymptotically equivalent to having a zero diagonal.}.
We solve eq.~\eqref{eq:TAP_order_3_denoising} perturbatively in $\alpha$\footnote{Note that one could also solve eq.~\eqref{eq:TAP_denoising_order_3_full} iteratively, as what would usually be done for TAP equations. Here, we rather leverage an analytical expression for the spectrum of the solution. 
In the more general matrix factorization case, there is no such expression for the solution to the TAP equations~\eqref{eq:fixedpoint_lagrange_extensive_xx},\eqref{eq:tap_order_3_extensive_xx}, and we found ``naive'' iterations to be numerically very ill-behaved, see our remark 
at the end of Section~\ref{subsec:plefka_symmetric}.}.
We denote by $y$ a generic eigenvalue of $\bY / \sqrt{m}$, and $g_y$ the corresponding eigenvalue of $\bg / \sqrt{m}$ (as is clear from eq.~\eqref{eq:TAP_order_3_denoising}, the two matrices will be diagonal in the same eigenbasis).
The first order is trivial, and given by:
\begin{align}\label{eq:gy_pgy_order_1}
	g_y = \frac{y}{\Delta + 1} + \mathcal{O}(\sqrt{\alpha}).
\end{align}
Plugging this back into eq.~\eqref{eq:TAP_order_3_denoising} and considering the order $\sqrt{\alpha}$,
we reach (imposing the zero-trace condition):
\begin{align}\label{eq:gy_pgy_order_2}
	g_y = \frac{y}{\Delta + 1} + \frac{(\Delta + 1 - y^2)}{(\Delta+1)^3} \sqrt{\alpha} + \mathcal{O}(\alpha).
\end{align}
This is the same result as the one obtained with the rotationally-invariant estimator of eq.~\eqref{eq:denoising_rie_order_2}, and this is the expression plotted in Fig.~\ref{fig:denoising_intro} in orange.
Therefore the PGY expansion 
does capture the asymptotic optimal estimator of the problem (in the overcomplete regime), as it should.

\medskip\noindent
\textbf{Effect of the prior on the denoising problem --}
Note that the PGY expansion suggests that the Bayes-optimal estimator for denoising is independent of the particularities 
of the prior distribution $P_X$, and only depends on its first two moments, assumed here to be fixed. 
While this observation is only made here up to order $\eta^3$, understanding if (and how) an i.i.d.\ prior on $\bX$ 
might enhance the denoising of the matrix $\bX \bX^\intercal$ by going to higher orders in the PGY expansion is beyond our results.
Note finally that the universality suggested by the PGY expansion at order $3$ might be better understood in terms of random matrix theory, 
in which many such universality results are known in the asymptotic limit, e.g.\ for the bulk spectral statistics of $\bX \bX^\intercal$ \cite{marchenko1967distribution}.

\subsection{Numerical simulations}\label{subsec:numerics}

In this section, we present numerical results illustrating our contributions.
The code is available in a public \href{https://github.com/sphinxteam/perturbative_mean_field_matrix_factorization}{GitHub repository} \cite{github_repo}.

\medskip\noindent
\textbf{Denoising MMSE for Wishart denoising --}
\begin{figure}[t]
  \includegraphics[width=\textwidth]{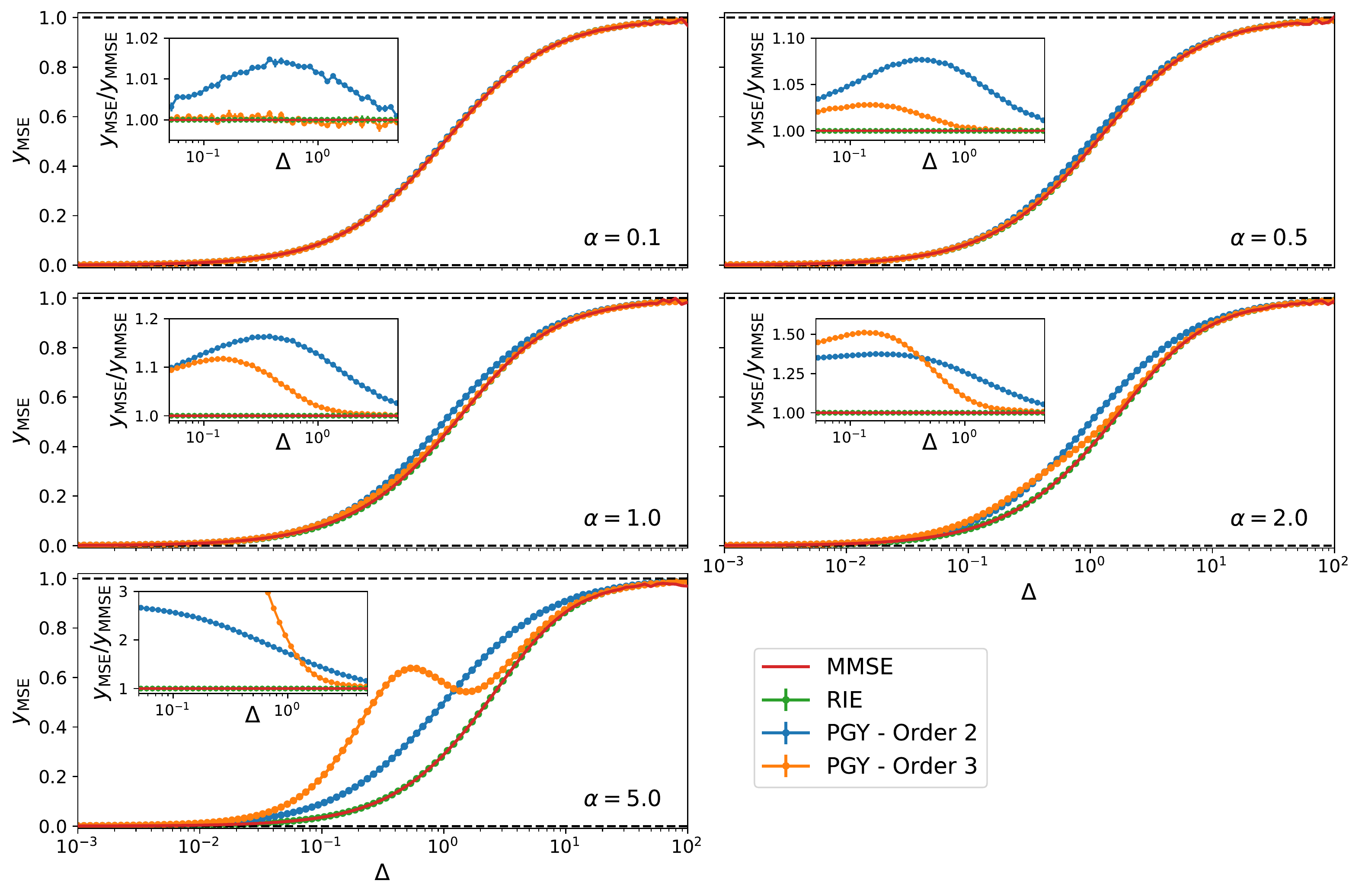}
	\caption{(Extended version of Fig. \ref{fig:denoising_intro}) Denoising of a Wishart matrix, for different values of the aspect ratio $\alpha$ as a function of the additive noise variance $\Delta$: comparison of the analytic MMSE with the RIE estimator of \cite{bun2016rotational}, and the 
  TAP-based estimator obtained from the PGY expansion truncated at orders $2$ and $3$. 
	\label{fig:comparison_denoising}}
\end{figure}
In Fig.~\ref{fig:comparison_denoising}, we compare the denoising errors reached by different procedures: 
\begin{itemize}[leftmargin=*]
  \item The MMSE, that is the analytical value of the Bayes-optimal error predicted in eq.~\eqref{eq:mmse_denoising_final_expr}, 
  is shown in the full red line.
  We discretize both parameters $(t,x)$, to obtain precise estimates of the involved integrals via Romberg's method \cite{romberg1955vereinfachte}, and a polynomial solver
  to compute the Stieltjes transform of $\bY(t)$ above the real axis by using eq.~\eqref{eq:algebraic_g}.
  \item The denoising MSE reached by the optimal RIE of eq.~\eqref{eq:expression_RIE}. 
  For each value of $\Delta$ (shown as green points), we compute the MSE achieved on $3$ different instances of $\bY$, and 
  error bars are too small to be visible. We use $m = 3000$ to generate the matrices.
  \item In blue points, we show the error reached by the PGY equations at order $\eta^2$, that is simply via the denoiser of eq.~\eqref{eq:gy_pgy_order_1}.
  As for the RIE, we use $m = 3000$ and average over $3$ instances, while error bars are in practice invisible.
  Note that the truncation at order $2$ corresponds to assuming that the underlying matrix is a Wigner matrix, as clearly shown in eq.~\eqref{eq:gy_pgy_order_1}.
  \item Finally, orange points correspond to the PGY equations truncated at order $\eta^3$, that is 
  eq.~\eqref{eq:gy_pgy_order_2}. We use the same physical parameters as for the PGY at order $2$.
\end{itemize}
In the inset of all figures shown in Fig.~\ref{fig:comparison_denoising}, we show the ratio of the error 
with the analytical MMSE.
It is clear that the optimal RIE achieves the Bayes-optimal error within numerical accuracy. 
Moreover, the PGY at order $3$ significantly outperforms the order-$2$ method when $\alpha$ is not too large (recall that the order of the PGY expansion 
can be understood as an order of perturbation in $\alpha$).
This strengthens our claim that the PGY expansion captures the optimal estimator, although one would have to take into account 
all orders of perturbation to turn it into an estimator for the factorization problem, something which is beyond the scope of this paper.
Interestingly, for large values of $\alpha$ (e.g. $\alpha = 5$ in Fig.~\ref{fig:comparison_denoising}) and small values of $\Delta$, 
the order-$3$ method can perform worse than the order-$2$.
Such a peculiar behavior of the PGY expansion series was already noticed for a class of Ising models 
in \cite{ricci2012bethe}. We remark that our evidence combined with the one of \cite{ricci2012bethe} 
suggests that such a phenomenon might also be present in the general PGY expansion of many rotationally-invariant inference problems 
which was derived in \cite{maillard2019high}: this interesting question is however not in the scope of the present paper and would deserve further studies.

\medskip\noindent
\textbf{Estimating the asymptotic free entropy --}
\begin{figure}[t]
  \includegraphics[width=\textwidth]{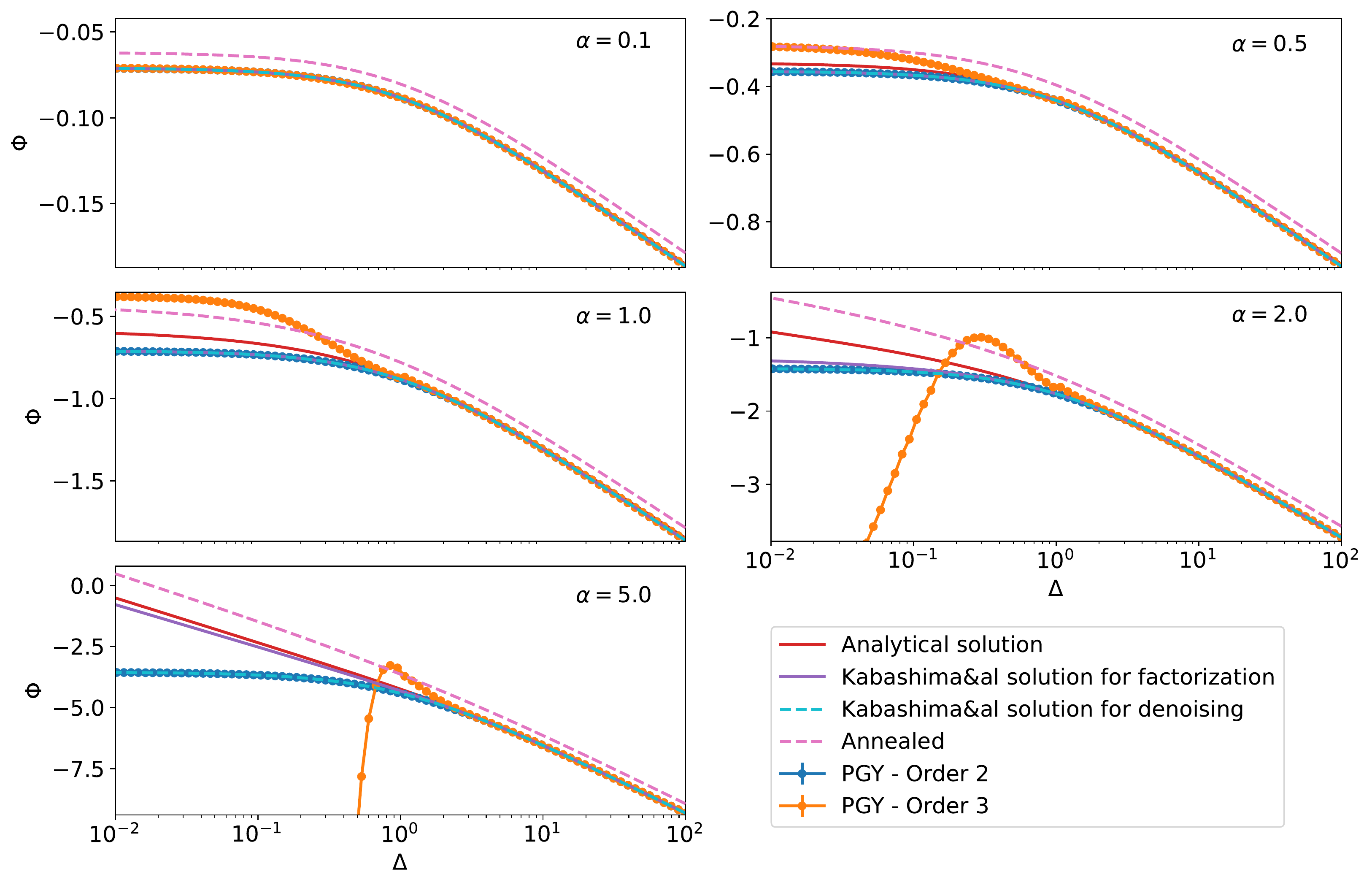}
  \caption{Denoising of a Wishart matrix: comparison of the analytical free entropy computed 
  using Matytsin's equations (red curve), the one computed with the TAP approach, truncated at orders $2$ (blue points) and $3$ (orange points), the annealed free energy (dashed pink line) computed in Appendix \ref{sec_app:annealed}, 
  the denoising solution of \cite{kabashima2016phase} in dashed cyan, and the (incorrect) solution for the factorization problem proposed in \cite{kabashima2016phase} in purple.
	\label{fig:comparison_free_energies}}
\end{figure}
Similar observations can be made by computing the free entropy predicted by the 
solution to the Matytsin problem and comparing it to the predictions of the truncated PGY expansion.
We present our results in Fig.~\ref{fig:comparison_free_energies}, in which we show the following curves for different values of $\alpha$:
\begin{itemize}[leftmargin=*]
  \item The red curve is obtained by analytically estimating eqs.~\eqref{eq:phi_quenched_alpha_leq1} and \eqref{eq:phi_quenched_alpha_geq1}, 
  depending on the regime of $\alpha$. The discretization technique we use was detailed in the description of Fig.~\ref{fig:comparison_denoising}.
  \item The blue and oranges points correspond to truncating the PGY expansion at order $2$ and $3$, as in Fig.~\ref{fig:comparison_denoising}, and the physical parameters 
  used in the simulations are the same.
  Note that truncating the PGY series at order $3$ does not necessarily yield a better approximation (to the free entropy or the MSE, see Fig.~\ref{fig:comparison_denoising}) than order $2$ for low values of $\Delta$: indeed, 
  this truncation does not correspond to any physical approximation, while we saw that the truncation at order $2$ is equivalent to assuming that the matrix to denoise 
  is Wigner.
  \item
  In dashed cyan, we show the free entropy prediction of \cite{kabashima2016phase} for the matrix denoising problem, matching exactly the PGY expansion at order $2$.
  Note that \cite{kabashima2016phase} also tackled the more involved matrix factorization problem. Since the free energies of factorization and denoising are the same, 
  we show as well in solid purple their (incorrect) prediction for the free entropy when analyzing the matrix factorization problem. Recall that we analytically disproved their approach in Section~\ref{subsec:previous_approaches}.
  The computation of these two curves is elementary, and described in Appendix~\ref{subsec_app:kabashima_sol_gaussian}.
  \item For completeness, we compare all these curves to an annealed bound (in dashed pink) that can be computed analytically, 
  using a calculation presented in Appendix~\ref{sec_app:annealed}. 
  In particular, the PGY expansion at order $3$ clearly violates the annealed bound for small enough $\Delta$ and large enough $\alpha$, 
  indicating that it does not correspond to a physical free entropy and that the PGY expansion would have to be carried to further orders 
  to correct this behavior.
\end{itemize}

\medskip\noindent
\textbf{Denoising of other rotationally-invariant matrices --}
\begin{figure}[t]
     \centering
		\includegraphics[width=\textwidth]{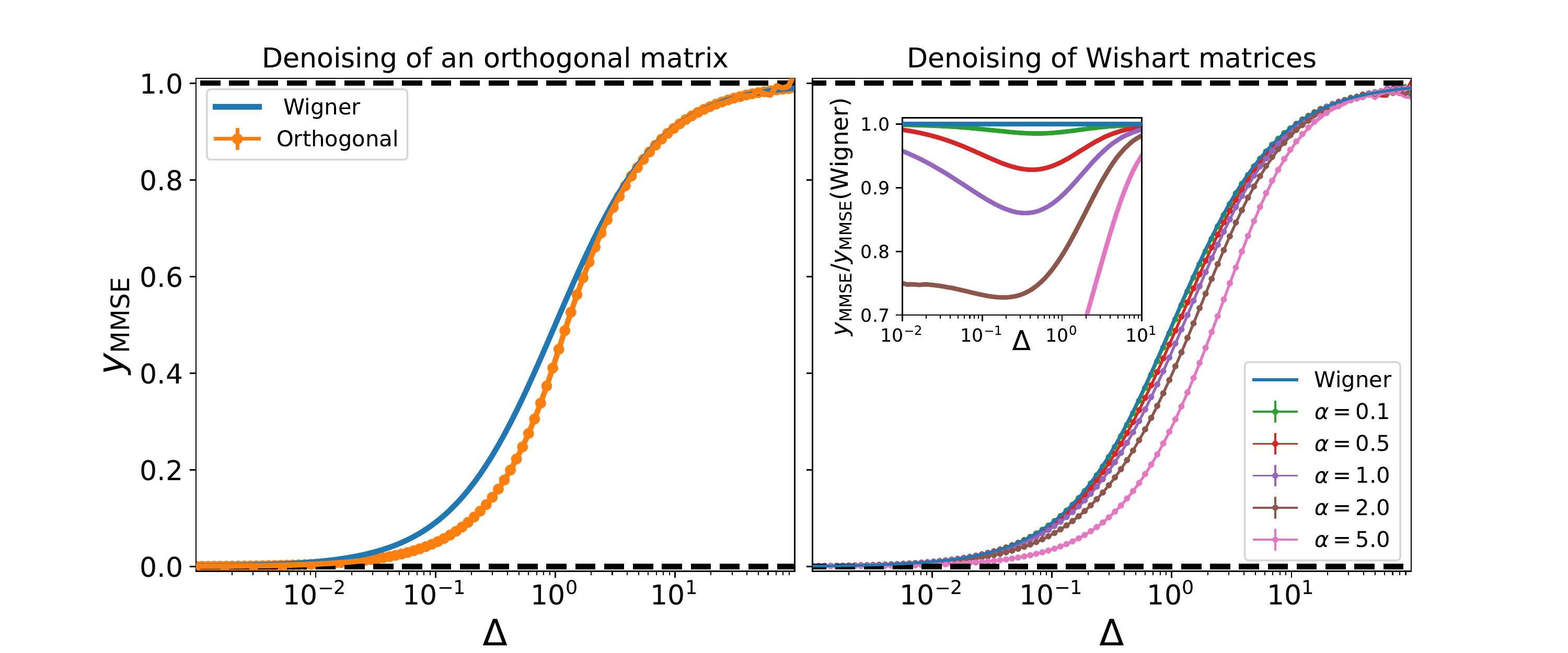}
     	\caption{
       Comparison of the denoising MMSE computed from the optimal rotation invariant estimator of \cite{bun2016rotational}, i.e.\ eq.~\eqref{eq:expression_RIE}, for different 
       ground-truth matrices $\bS^\star$. We use $m = 3000$ and average over $3$ instances, with no visible error bars.
       Left: denoising of a uniformly-sampled symmetric orthogonal matrix. Right: denoising of a Wishart matrix with various values of $m/n=\alpha$. The well known denoising of a Wigner random matrix (which is the overcomplete limit of a Wishart matrix when $\alpha\to 0$) is shown for reference.
	\label{fig:comparison_denoising_orthogonal_wigner}}
\end{figure}
In Fig.~\ref{fig:comparison_denoising_orthogonal_wigner}, we illustrate the influence of the structure of the signal on the performance of the optimal denoiser. 
We show how the structure of a symmetric orthogonal matrix (i.e.\ $\bS^\star = \bO \bD \bO^\intercal$, with a uniformly-sampled orthogonal matrix $\bO$, and $D_\mu \overset{\mathrm{i.i.d.}}{=} \pm 1$ with probability $1/2$), or of a Wishart matrix with different values of $\alpha$, 
allows to significantly improve the denoising performance over simple Wigner (scalar) denoising.

\medskip\noindent 
\textbf{Large-$\alpha$ behavior and the BBP transition --}
\begin{figure}[t]
     \centering
		\includegraphics[width=\textwidth]{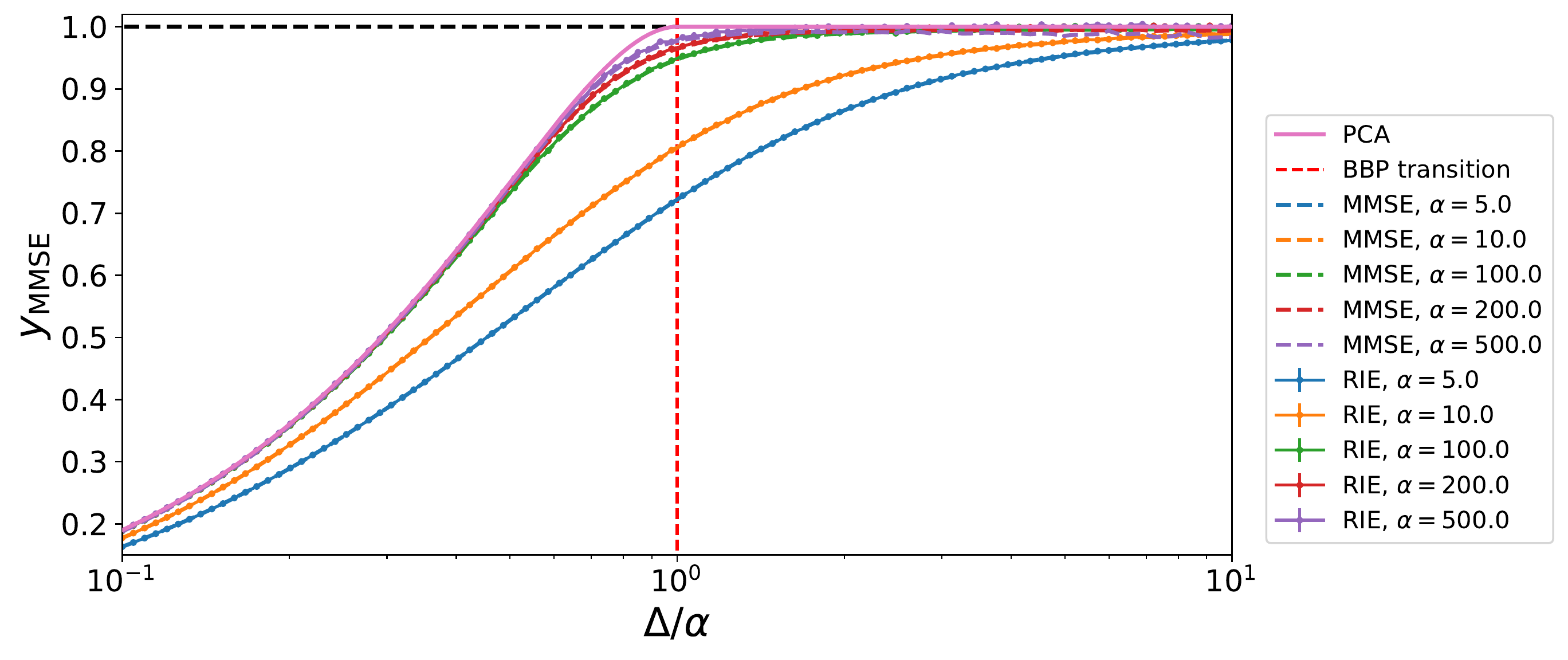}
     	\caption{
         Denoising MMSE for Wishart denoising at large $\alpha$, and different values of $\Delta$.
         We show both analytical predictions using eq.~\eqref{eq:mmse_denoising_final_expr} (dashed lines) and 
         numerical applications of the RIE of eq.~\eqref{eq:expression_RIE} (using $m = 20000$ and averaging over $3$ instances, error bars are invisible, dashed lines are covered by the corresponding data points). 
         Numerical errors arise in the MMSE calculation, due to the difficulty of estimating the derivative in eq.~\eqref{eq:expression_RIE} for very large values of the noise.
         We compare it with the performance of vanilla PCA, which is optimal in the low-rank limit $n = \mathcal{O}(1)$.
	\label{fig:denoising_large_alpha}}
\end{figure}
We conclude this presentation of numerical results by Fig.~\ref{fig:denoising_large_alpha}, in which we 
investigate the denoising MMSE in the large $\alpha$ (or undercomplete) regime. 
While this regime is not reachable by our PGY expansion, we can evaluate the MMSE with the exact formula of eq.~\eqref{eq:mmse_denoising_final_expr} (dashed lines), 
or by sampling large matrices and using the optimal RIE of eq.~\eqref{eq:expression_RIE}. 
The agreement between these two methods is again extremely good, showing in particular that the RIE calculation is not affected strongly by finite-size effects (e.g.\ at $\alpha = 200$, we have 
$n = m / \alpha = 100$, so that finite-size effects could have been present).

\medskip\noindent
Interestingly, we can compare the performance in this regime to the low-rank limit in which $n = \mathcal{O}(1)$. 
In the extreme case in which $n = 1$ (so that $\bX = \bx$ is a vector), it is known that the optimal procedure is PCA, which estimates $\bx$ as the leading eigenvector of $\bY/\sqrt{\alpha} = \bx \bx^\intercal / \sqrt{m} + \sqrt{\Delta/\alpha} \bZ$.
Its performance has been described very precisely, and exhibits a transition, known as ``BBP'' \cite{baik2005phase} at $\Delta/\alpha = 1$.
More precisely, if $\bv$ is the leading eigenvector of $\bY$, we have 
\begin{align}
   \lim_{m \to \infty} \frac{1}{m} |\bv^\intercal \bx|^2 &= 
   \begin{dcases}
      0 & \textrm{ if } \Delta \geq \alpha \\ 
      1 - \frac{\Delta}{\alpha} &\textrm{ if } \Delta \leq \alpha.
   \end{dcases}
\end{align}
We show the corresponding denoising MMSE as a pink line in Fig.~\ref{fig:denoising_large_alpha}. 
It is clear that the optimal denoiser approaches the PCA performance at large $\alpha$. We note that in the extensive-rank regime we observe no phase transition in the MMSE. We will further comment  
on the description of the transition between the low-rank and extensive-rank regimes in the coming conclusion.

%% file: conclusion.tex
\section*{Conclusion and openings}
\addcontentsline{toc}{section}{Conclusion and open questions}

In this paper, we have shown how the PGY formalism of high-temperature expansions at fixed order parameters can be applied to
the problems of extensive-rank matrix factorization and denoising.
The previous approaches to this problem \cite{sakata2013statistical,krzakala2013phase,kabashima2016phase,parker2014bilinear,parker2014bilinear2,zou2021multi} turn out to be relying on a hidden assumption that amounts to truncating the PGY series at second order.
Both in the matrix factorization and in the denoising problem, we have shown, by an explicit computation of the third order contribution, that the higher-order corrections are not negligible.
 
\medskip\noindent
The denoising of rotationally-invariant matrices provides a good test of our approach.
On the one hand it can be solved analytically using matrix integration techniques.
On the other hand there exists an optimal denoising algorithm that relies only on the observed matrix, and denoises each eigenvalue.
We have shown that these two approaches give the same mean squared error, and that our PGY-based expansion, truncated at order three, also agrees with these other approaches for denoising Wishart matrices in an overcomplete (or small-$\alpha$) regime close to the Wigner limit.
 
\medskip\noindent
The interest of the formalism that we have presented here is that it opens the way to an optimality-matching algorithm for matrix factorization in the extensive rank case accounting for general prior and output channel.
Indeed, the optimization of the PGY free entropy leads to TAP equations which, properly iterated in time, develop into AMP-type algorithms.
An important step in that direction would of course be to work out the PGY expansion to all orders, and re-sum it.
On the other hand, as a first step, it will be interesting to find out if the TAP equations that we have derived at order 3 can be transformed into a useful algorithm for overcomplete matrix factorization.

\medskip\noindent
It would also be interesting to perform a perturbative calculation of the free entropy and MMSE in the opposite regime of large $\alpha$, to analytically understand 
the limit behavior observed in Fig.~\ref{fig:denoising_large_alpha}.
Such a calculation is related to an important open problem in random matrix theory, that is a sharp description of the transition between the low-rank and extensive-rank spherical HCIZ integrals.
Solving this later problem would help to describe the transition between low-rank and extensive-rank results in the PGY expansion.
    
\medskip\noindent
Another interesting open problem is to relate the PGY-expansion to other approaches.
The recent conjecture proposed in \cite{barbier2021statistical} for solving matrix factorization with replicas should be compared to the present approach.
In general it is well known that the two main approaches to disordered systems, the replica approach and the cavity approach (that derives TAP equations) have the same physical content, the second one having the advantage that it can be transformed into an algorithm.
The challenge that we face here is the extension of these two methods to the case where the order parameter is a large matrix, and the contact that can be made between the two.
In particular, in both of these approaches, we have found that the effective order parameter turns out to be a distribution of eigenvalues.
Understanding the generality of this statement is an important open challenge.
On the other hand, it has been shown in \cite{maillard2019high} in finite-rank problems that these expansions are equivalent to other techniques, e.g.\ the Expectation Consistency or adaptive TAP approaches.
This equivalence does not seem to easily transfer to the extensive-rank case: understanding how to apply these approaches here is another interesting direction of research.

%% file: appendix_rmt.tex
\section{Some definitions of probability theory}\label{sec_app:rmt}

We introduce here a few notations and quantities particularly useful to study the asymptotic spectra of 
random matrices.
For more complete mathematical or physical introductions to random matrix theory, the reader can refer to \cite{mehta2004random, anderson2010introduction,livan2018introduction}.
Let us consider a symmetric random matrix $\bM \in \bbR^{n \times n}$, with eigenvalues $(\lambda_1,\cdots,\lambda_n)$. 
In general, we will consider random matrices that admit an \emph{asymptotic spectral density}, here denoted $\mu(x)$, such that 
\begin{align}
    \mu(x) &= \lim_{n \to \infty} \frac{1}{n} \sum_{i=1}^n \delta(x - \lambda_i).
\end{align}
Here the limit is to be understood as the weak limit of probability measures.
We denote $\lambda_{\rm max} \equiv \max \mathrm{supp}(\mu)$ and $\lambda_{\rm min} \equiv \min \mathrm{supp}(\mu)$, 
assuming that the support of $\mu$ is bounded from above and below.
Letting $\bbC_+ \equiv \{z \in \bbC \ \mathrm{s.t.} \ \mathrm{Im}(z) > 0\}$,
one can introduce the \emph{Stieltjes} transform of $\mu$ as:
\begin{align}\label{eq:def_stieltjes}
g_\mu (z) \equiv \mathbb{E} \Big[\frac{1}{X-z}\Big] = \int_{\mathbb{R}} \mathrm{d}\lambda \, \mu(\lambda) \, \frac{1}{\lambda - z} = \lim_{n \to \infty} \frac{1}{\lambda_i - z},
\end{align}
for all $z \in \bbC_+$. Note that then we also have $g_{\mu}(z) \in \bbC_+$.
Moreover, on $(\lambda_{\rm max},+\infty)$, $g_\mu$ induces a strictly increasing $\mathcal{C}^\infty$ diffeomorphism ${\cal S}_\mu : (\lambda_{\rm max},\infty) \hookrightarrow (-\infty,0)$, and 
we denote its inverse $g^{-1}_\mu$. 
One can then introduce the \emph{${\cal R}$-transform} of $\mu$ as: 
\begin{align}
\forall z > 0, \quad \mathcal{R}_\mu(z) &\equiv g_\mu^{-1}(-z) - \frac{1}{z}.
\end{align}
$\mathcal{R}_\mu(z)$ is \emph{a priori} defined for $-z \in g_\mu[(\lambda_{\rm min},\lambda_{\rm max})^c]$
and admits an analytical expansion around $z = 0$.
We can write this expansion as:
\begin{align}\label{eq:expansion_R_freecum}
\mathcal{R}_\mu(z) = \sum_{k=0}^\infty c_{k+1}(\mu) \, z^k.
\end{align} 
The elements of the sequence $\{c_k(\mu)\}_{k \in \bbN^\star}$ are called the \emph{free cumulants} of $\mu$.
In particular, one can show that $c_1(\mu) = \EE_\mu(X)$ and $c_2(\mu) = \EE_\mu(X^2) - (\EE_\mu X)^2$.
The free cumulants can be recursively computed from the moments of the measure using the so-called \emph{free cumulant equation}: 
\begin{align}\label{eq:free_cum_formula}
\forall k \in \mathbb{N}^*, \quad \mathbb{E}_\mu X^k = \sum_{m = 1}^k c_m(\mu) \sum_{\substack{\{k_i\}_{i \in [|1,m|]} \\ \text{s.t } \sum_i k_i = k}} \prod_{i=1}^m \mathbb{E}_\mu X^{k_i -1}.
\end{align}

%% file: appendix_plefka_technicalities.tex
\section{Technicalities of the PGY expansion}\label{sec_app:technicalities_plefka}

\subsection{Setting of the Plefka-Georges-Yedidia expansion}\label{subsec_app:setting_pgy}

We describe here in more details the formalism we used to derive Result~\ref{result:pgy_order_3_extensive_xx}: we follow the formalism of \cite{georges1991expand} to perform a Plefka expansion \cite{plefka1982convergence}.
Note that some parts of the derivation are very similar to what is done in \cite{maillard2019high}.
We start from eq.~\eqref{eq:phi_extensive_xx_before_plefka}:
\begin{align*}
   \begin{dcases}
      nm\Phi_{\bY,n} &= \sum_{\mu,i}\Big[\lambda_{\mu i} m_{\mu i} + \frac{\gamma_{\mu i}}{2}\Big(v_{\mu i} + (m_{\mu i})^2\Big) \Big] + \sum_{\mu < \nu}\Big[-\omega_{\mu \nu} g_{\mu \nu} - \frac{b_{\mu \nu}}{2}\Big(-r_{\mu \nu} + g_{\mu \nu}^2\Big) \Big] \\
      & + \ln \int P_H(\rd \bH) P_X(\rd \bX) \, e^{-S_{\rm eff}[\bX,\bH]}, \\
         S_{\rm eff}[\bX,\bH] &\equiv \sum_{\mu, i} \Big[\lambda_{\mu i} X_{\mu i} + \frac{\gamma_{\mu i}}{2} X_{\mu i}^2\Big] + \sum_{\mu<\nu} \Big[\omega_{\mu \nu} (iH)_{\mu \nu} - \frac{b_{\mu \nu}}{2} (iH)_{\mu \nu}^2\Big] + H_{\rm eff}[\bX,\bH], \\
         H_{\rm eff}[\bX,\bH] &\equiv \frac{1}{\sqrt{n}}\sum_{\mu <\nu} \sum_{i} (iH)_{\mu \nu} \,X_{\mu i} \,X_{\nu i}.
   \end{dcases}
\end{align*}
For clarity, we will keep the dependency of all the Lagrange parameters on $\eta$
explicit when needed. For a given $\eta$ and a given $\bY$, one defines the operator $U$ of Georges-Yedidia \cite{georges1991expand}:
\begin{align}\label{eq:def_U_xx}
   U_{\bY,\eta} \equiv  H_{\rm eff}- \langle H_{\rm eff}\rangle_{\bY,\eta} &+ \sum_{\mu,i} \partial_\eta\lambda_{\mu i} (X_{\mu i} - m_{\mu i}) + \frac{1}{2} \sum_{\mu,i}\partial_\eta \gamma_{\mu i}[X_{\mu i}^2 - v_{\mu i} - (m_{\mu i})^2 ] \nonumber \\
   &+ \sum_{\mu < \nu} \partial_\eta \omega_{\mu \nu} (iH_{\mu \nu} + g_{\mu \nu}) - \frac{1}{2} \sum_{\mu<\nu}\partial_\eta b_{\mu \nu}[(iH)_{\mu \nu}^2 + r_{\mu \nu} - g_{\mu \nu}^2].
\end{align}

\subsection{First orders of perturbation}\label{subsec_app:first_orders_pgy_xx}

\noindent
\textbf{Order $1$ in $\eta$ --}
At order $1$, we have directly:
\begin{align}\label{eq:phi_pgy_xx_order_1}
   \Big(\frac{\partial \Phi_{\bY,n}}{\partial \eta}\Big)_{\eta = 0} &= \frac{1}{nm} \langle H_\mathrm{eff} \rangle_0 = \frac{1}{n^{3/2}m} \sum_i \sum_{\mu < \nu} g_{\mu \nu} m_{\mu i} m_{\nu i}.
\end{align}
We can then use ``Maxwell-like'' relations \cite{georges1991expand} to compute the derivatives of the Lagrange parameters at $\eta = 0$.
For instance, for $\lambda_{\mu i}$ and $\gamma_{\mu i}$ they can be read from eq.~\eqref{eq:phi_extensive_xx_before_plefka}:
\begin{subnumcases}{\label{eq:maxwell}}
   \label{eq:maxwell_1}
   \lambda_{\mu i} + m_{\mu i} \gamma_{\mu i} = nm \, \frac{\partial \Phi_{\bY,n}}{\partial m_{\mu i}}, &\\
   \label{eq:maxwell_2}
   \gamma_{\mu i} = 2 nm \, \frac{\partial \Phi_{\bY,n}}{\partial v_{\mu i}}, & 
\end{subnumcases}
so that we can compute :
\begin{align}\label{eq:maxwell_relations_x}
   \begin{dcases}
   \partial_\eta \gamma_{\mu i}(\eta = 0) &= 2 n m \frac{\partial}{\partial v_{\mu i}} \Big[\frac{\partial \Phi_{\bY,n}}{\partial \eta}(\eta = 0)\Big] = 0, \\
   \partial_\eta \lambda_{\mu i}(\eta = 0) &= n m \frac{\partial}{\partial m_{\mu i}} \Big[\frac{\partial \Phi_{\bY,n}}{\partial \eta}(\eta = 0)\Big] = \frac{1}{\sqrt{n}} \sum_{\nu} g_{\mu \nu} m_{\nu i}. 
   \end{dcases}
\end{align}
In a similar fashion, we reach
\begin{align}\label{eq:maxwell_relations_h}
   \begin{dcases}
   \partial_\eta b_{\mu \nu}(\eta = 0) &= 2 n m \frac{\partial}{\partial r_{\mu \nu}} \Big[\frac{\partial \Phi_{\bY,n}}{\partial \eta}(\eta = 0)\Big] = 0, \\
   \partial_\eta \omega_{\mu \nu}(\eta = 0) &= - n m \frac{\partial}{\partial g_{\mu \nu}} \Big[\frac{\partial \Phi_{\bY,n}}{\partial \eta}(\eta = 0)\Big] = -\frac{1}{\sqrt{n}} \sum_{i} m_{\mu i}m_{\nu i}. 
   \end{dcases}
\end{align}
Using eqs.~\eqref{eq:maxwell_relations_x} and \eqref{eq:maxwell_relations_h}, we can compute $U_{\bY,0}$ from eq.~\eqref{eq:def_U_xx}.
For clarity of the notation, we will denote by lowercase letters \emph{centered variables}, i.e.\ $x_{\mu i}\equiv X_{\mu i} - m_{\mu i}$ and $(ih)_{\mu \nu} \equiv (iH)_{\mu \nu} + g_{\mu \nu}$. 
We obtain after a straightforward calculation:
\begin{align}\label{eq:U0_xx}
    U_{\bY,\eta = 0} &= \frac{1}{\sqrt{n}} \sum_{i} \sum_{\mu < \nu} [(ih)_{\mu \nu} x_{\mu i} x_{\nu i} - g_{\mu \nu} x_{\mu i} x_{\nu i} + (ih)_{\mu \nu} m_{\mu i} x_{\nu i} + (ih)_{\mu \nu} x_{\mu i} m_{\nu i}].
\end{align}

\medskip\noindent
\textbf{Order $2$ in $\eta$ --}
Relying on the formulas of \cite{georges1991expand}, 
we have:
\begin{align}\label{eq:pgy_method_order_2}
   \frac{1}{2}\Big(\frac{\partial^2 \Phi_{\bY,n}}{\partial \eta^2}\Big)_{\eta = 0} &= \frac{1}{2nm} \langle U_\bY^2\rangle_0.
\end{align}
Recall that at $\eta = 0$, all the variables $\{h_{\mu \nu},x_{\rho i}\}$ are independent and have zero mean. 
Using this fact alongside with the expression of $U_{\bY,\eta = 0}$ given in eq.~\eqref{eq:U0_xx} 
yields:
\begin{align}\label{eq:phi_pgy_xx_order_2}
   \frac{1}{2}\Big(\frac{\partial^2 \Phi_{\bY,n}}{\partial \eta^2}\Big)_{\eta = 0} &= \frac{1}{2 n^2 m} \sum_{i,j} \sum_{\substack{\mu < \nu \\ \mu' < \nu'}} \Big\langle [(ih)_{\mu \nu} x_{\mu i} x_{\nu i} - g_{\mu \nu} x_{\mu i} x_{\nu i} + (ih)_{\mu \nu} m_{\mu i} x_{\nu i} + (ih)_{\mu \nu} x_{\mu i} m_{\nu i}]
   \nonumber \\ 
   & \hspace{1cm} \times 
   [(ih)_{\mu' \nu'} x_{\mu' j} x_{\nu' j} - g_{\mu' \nu'} x_{\mu' j} x_{\nu' j} + (ih)_{\mu' \nu'} m_{\mu' j} x_{\nu' j} + (ih)_{\mu' \nu'} x_{\mu' j} m_{\nu' j}]
   \Big\rangle_0,  \nonumber \\ 
   &= \frac{1}{2 n^2 m} \sum_i \sum_{\mu < \nu} [-r_{\mu \nu} v_{\mu i} v_{\nu i} + g_{\mu \nu}^2 v_{\mu i} v_{\nu i} - r_{\mu \nu} m_{\mu i}^2 v_{\nu i} - r_{\mu \nu} v_{\mu i} m_{\nu i}^2].
\end{align}

\medskip\noindent
\textbf{Order $3$ in $\eta$ --}
At order $3$, the formula of Appendix~A of \cite{georges1991expand} yields:
\begin{align}\label{eq:pgy_method_order_3}
   \frac{1}{3!}\Big(\frac{\partial^3 \Phi_{\bY,n}}{\partial \eta^3}\Big)_{\eta = 0} &= -\frac{1}{6nm} \langle U_\bY^3\rangle_0.
\end{align}
In order to compute the right hand-side of eq.~\eqref{eq:pgy_method_order_3}, we decompose the operator $U$ of eq.~\eqref{eq:U0_xx} as 
$U_\bY = U_c + U_g + U_m$, with 
\begin{align}\label{eq:decomposition_U0_xx}
   \begin{dcases}
      U_c &\equiv \frac{1}{\sqrt{n}} \sum_{i} \sum_{\mu < \nu} (ih)_{\mu \nu} x_{\mu i} x_{\nu i}, \\
      U_g &\equiv -\frac{1}{\sqrt{n}} \sum_{i} \sum_{\mu < \nu} g_{\mu \nu} x_{\mu i} x_{\nu i}, \\
      U_m &\equiv  \frac{1}{\sqrt{n}} \sum_{i} \sum_{\mu < \nu} [(ih)_{\mu \nu} m_{\mu i} x_{\nu i} + (ih)_{\mu \nu} x_{\mu i} m_{\nu i}].
   \end{dcases}
\end{align}
Therefore, we obtain (dropping the $0$ subscript in the averages to lighten the notations): 
\begin{align}\label{eq:decomposition_U3}
   \langle U_\bY^3 \rangle_0 &= \langle U_c^3 \rangle + \langle U_g^3 \rangle + \langle U_m^3 \rangle \\
  \nonumber &+  3 \langle U_c^2 U_g \rangle + 3 \langle U_c^2 U_m \rangle + 3 \langle U_g^2 U_c \rangle + 3 \langle U_g^2 U_m \rangle 
   + 3 \langle U_m^2 U_c \rangle + 3 \langle U_m^2 U_g \rangle + 6 \langle U_c U_g U_m \rangle.
\end{align}
In eq.~\eqref{eq:decomposition_U3}, the terms $\langle U_g^2 U_c \rangle$ and $\langle U_g^2 U_m\rangle$  are trivially zero since $\langle \bh \rangle = 0$.
Let us now argue that, in eq.~\eqref{eq:decomposition_U3}, all the terms except $\langle U_g^3\rangle$ are negligible in the thermodynamic limit.
We can for instance consider $\langle U_c^3 \rangle$, which can easily be written from eq.~\eqref{eq:decomposition_U0_xx} as (using again $\langle \bh \rangle = 0$): 
\begin{align}
   \langle U_c^3 \rangle &= \frac{1}{n^{3/2}} \sum_{\mu < \nu} \langle (ih)_{\mu \nu}^3 \rangle \sum_{i,j,k} \langle x_{\mu i} x_{\nu i} x_{\mu j} x_{\nu j} x_{\mu k} x_{\nu k} \rangle, \\ 
   &= \frac{1}{n^{3/2}} \sum_{\mu < \nu} \langle (ih)_{\mu \nu}^3 \rangle \sum_{i=1}^n \langle x_{\mu i}^3 \rangle \langle x_{\nu i}^3 \rangle.
\end{align}
In particular, this directly implies that $\langle U_c^3 \rangle = \mathcal{O}(n^{3/2})$, and therefore that it will not contribute to the asymptotic free energy, 
which is in the scale $\Theta(n^2)$. In the same way, one can show: 
\begin{subnumcases}{\label{eq:all_terms_order_3}}
  \langle U_m^3 \rangle = \frac{1}{n^{3/2}} \sum_{\mu < \nu} \langle (ih)_{\mu \nu}^3 \rangle \sum_{i=1}^n [ m_{\mu i}^3 \langle x_{\nu i}^3 \rangle +  m_{\nu i}^3 \langle x_{\mu i}^3 \rangle ] = \mathcal{O}(n^{3/2}) , & \\
  \langle U_c^2 U_g \rangle = \frac{1}{n^{3/2}} \sum_{\mu < \nu} g_{\mu \nu} r_{\mu \nu} \sum_{i=1}^n  \langle x_{\mu i}^3 \rangle \langle x_{\nu i}^3 \rangle = \mathcal{O}(n^{3/2}) , & \\
  \langle U_c^2 U_m \rangle = \frac{1}{n^{3/2}} \sum_{\mu < \nu} \langle (ih)_{\mu \nu}^3 \rangle \sum_{i=1}^n [\langle x_{\mu i}^3 \rangle m_{\nu i} v_{\nu i} + \langle x_{\nu i}^3 \rangle m_{\mu i} v_{\mu i}] = \mathcal{O}(n^{3/2}) , & \\
  \langle U_m^2 U_c \rangle = \frac{2}{n^{3/2}} \sum_{\mu < \nu} \langle (ih)_{\mu \nu}^3 \rangle \sum_{i=1}^n m_{\mu i} v_{\mu i} m_{\nu i} v_{\nu i} = \mathcal{O}(n^{3/2}) , & \\
  \langle U_m^2 U_g \rangle = \frac{2}{n^{3/2}} \sum_{\mu < \nu} g_{\mu \nu} r_{\mu \nu} \sum_{i=1}^n m_{\mu i} v_{\mu i} m_{\nu i} v_{\nu i} = \mathcal{O}(n^{3/2}) , & \\
  \langle U_c U_g U_m \rangle = \frac{1}{n^{3/2}} \sum_{\mu < \nu} g_{\mu \nu} r_{\mu \nu} \sum_{i=1}^n [\langle x_{\mu i}^3 \rangle m_{\nu i} v_{\nu i} + \langle x_{\nu i}^3 \rangle m_{\mu i} v_{\mu i}] = \mathcal{O}(n^{3/2}). &
\end{subnumcases}
All in all, eqs.~\eqref{eq:pgy_method_order_3} and \eqref{eq:all_terms_order_3} imply that: 
\begin{align}
   \frac{1}{3!}\Big(\frac{\partial^3 \Phi_{\bY,n}}{\partial \eta^3}\Big)_{\eta = 0} &= -\frac{1}{6 nm} \langle U_g^3 \rangle + \smallO_n(1).
\end{align}
This remaining term can be computed again by expanding the sums and using the independence of all variables $\{h_{\mu \nu}, x_{\rho i}\}$ at $\eta = 0$.
Elementary combinatorics allow to count the number of terms appearing in this expansion, and we reach:
\begin{align}\label{eq:expansion_Ug3}
   \langle U_g^3 \rangle &= -\frac{1}{8 n^{3/2}} \Big\{ \prod_{a=1}^3\sum_{\mu_a \neq \nu_a} g_{\mu_a \nu_a}\Big\} \sum_{i_1,i_2,i_3} \Big\langle \prod_{a=1}^3 x_{\mu_a i_a} x_{\nu_a i_a} \Big\rangle, \nonumber \\
    &= -\frac{1}{n^{3/2}} \Big\{ \prod_{a=1}^3\sum_{\mu_a \neq \nu_a} g_{\mu_a \nu_a}\Big\} \delta_{\mu_1 \nu_2} \delta_{\mu_2 \nu_3} \delta_{\mu_3 \nu_1} \sum_{i=1}^n \prod_{a=1}^3 v_{\mu_a i} - \frac{1}{2n^{3/2}} \sum_{\mu \neq \nu} g_{\mu \nu}^3 \sum_{i=1}^n \langle x_{\mu i}^3 \rangle \langle x_{\nu i}^3 \rangle.
\end{align}
Note that the factor $8$ in eq.~\eqref{eq:expansion_Ug3} disappeared since all possible pairings of indices $\{\mu_a,\nu_a\}_{a=1}^3$ are equivalent, and one can easily count that there 
are eight such pairings.
The second term in eq.~\eqref{eq:expansion_Ug3} is again $\mathcal{O}(n^{3/2})$, and we reach:
\begin{align}
   \label{eq:phi_pgy_xx_order_3}
   \frac{1}{3!}\Big(\frac{\partial^3 \Phi_{\bY,n}}{\partial \eta^3}\Big)_{\eta = 0} &= \frac{1}{6 n^{5/2}m} \sum_i \sum_{\substack{\mu_1,\mu_2,\mu_3 \\ \text{pairwise distinct}}} g_{\mu_1 \mu_2} g_{\mu_2 \mu_3} g_{\mu_3 \mu_1} \prod_{a=1}^3 v_{\mu_a i} + \smallO_n(1).
\end{align}
Note that the constraint of having pairwise distinct indices is a simple consequence of the constraint $\mu_a \neq \nu_a$ in eq.~\eqref{eq:expansion_Ug3}, 
associated with the form of the pairing.

\medskip\noindent
\textbf{Remark on higher-order cumulants --}
An important remark that one can already conjecture by generalizing from eqs.~\eqref{eq:phi_pgy_xx_order_2},\eqref{eq:phi_pgy_xx_order_3} is that at any given 
order of perturbation in $\eta$, \emph{only the first two moments of the fields $\bH,\bX$ will appear at dominant order}.
This conjecture arises as a consequence of a simple scaling argument : the higher-order moments constraint too much indices on which we can sum, and thus the terms involving them can be neglected.
This was already formulated for similar Plefka expansions in symmetric and bipartite finite-rank models in \cite{maillard2019high}.

\subsection{PGY expansion in the non-symmetric Model~\ref{model:extensive_factorization_fx}}\label{subsec_app:pgy_fx}

In this section, we detail the derivation of Result~\ref{result:pgy_order_4_extensive_fx} from the PGY expansion formalism. 
As many things are similar to the derivation of Result~\ref{result:pgy_order_3_extensive_xx}, we will shorten some arguments that can be easily transposed 
from the symmetric setting.

\subsubsection{The method}\label{subsubsec_app:method_pgy_fx}

As in the symmetric case, we start from the original expression of the TAP free entropy, very similar to eq.~\eqref{eq:phi_extensive_xx_before_plefka}:
\begin{align}
   \begin{dcases}
    n(m+p)\Phi_{\bY,n}(\eta) &= \sum_{\mu,i}\Big[\lambda^F_{\mu i} m^F_{\mu i} + \frac{\gamma^F_{\mu i}}{2}(v^F_{\mu i} + (m^F_{\mu i})^2) \Big] + \sum_{i,l}\Big[\lambda^X_{i l} m^X_{i l} + \frac{\gamma^X_{i l}}{2}(v^X_{i l} + (m^X_{i l})^2) \Big] , \\
    & \hspace{-6pt}+ \sum_{\mu,l}\Big[-\omega_{\mu l} g_{\mu l} - \frac{b_{\mu l}}{2}(-r_{\mu l} + g_{\mu l}^2) \Big] + \ln \int P_H(\mathrm{d}\bH) \, P_F(\mathrm{d}\bF) \, P_X(\mathrm{d}\bX) \, e^{-S_{\rm eff,\eta}[\bF,\bX,\bH]}, \\
    S_{\rm eff,\eta}[\bF,\bX,\bH] &\equiv \sum_{\mu, i} \Big[\lambda^F_{\mu i} F_{\mu i} + \frac{\gamma^F_{\mu i}}{2} F_{\mu i}^2\Big] +   \sum_{i,l} \Big[\lambda^X_{i l} X_{il} + \frac{\gamma^X_{i l}}{2} X_{i l}^2\Big]\\
    &\quad + \sum_{\mu,l} \Big[\omega_{\mu l} (iH)_{\mu l} - \frac{b_{\mu l}}{2} (iH)_{\mu l}^2\Big] + \frac{\eta}{\sqrt{n}}\sum_{\mu,i,l} (iH)_{\mu l} \,F_{\mu i} \,X_{i l}, \\
    H_{\rm eff}[\bF,\bX,\bH] &\equiv \frac{1}{\sqrt{n}} \sum_{\mu,i,l} (iH)_{\mu l} F_{\mu i} X_{il}. 
   \end{dcases}
\end{align}
The operator $U$ of Georges-Yedidia is defined similarly as the one of eq.~\eqref{eq:def_U_xx}:
\begin{align}\label{eq:def_U_fx}
   U_{\bY,\eta} \equiv  H_{\rm eff}- \langle H_{\rm eff}\rangle_{\bY,\eta} &+ \sum_{\mu,i} \partial_\eta\lambda^F_{\mu i} (F_{\mu i} - m^F_{\mu i}) + \frac{1}{2} \sum_{\mu,i}\partial_\eta \gamma^F_{\mu i}[F_{\mu i}^2 - v^F_{\mu i} - (m^F_{\mu i})^2 ] \nonumber \\
   &+ \sum_{i,l} \partial_\eta\lambda^X_{il} (X_{i l} - m^X_{i l}) + \frac{1}{2} \sum_{i,l}\partial_\eta \gamma^X_{i l}[X_{i l}^2 - v^X_{i l} - (m^X_{i l})^2] \nonumber \\ 
   &+ \sum_{\mu,l} \partial_\eta \omega_{\mu l} (iH_{\mu l} + g_{\mu l}) - \frac{1}{2} \sum_{\mu,l}\partial_\eta b_{\mu l}[(iH)_{\mu l}^2 + r_{\mu l} - g_{\mu l}^2].
\end{align}
\textbf{Order $1$ in $\eta$ --} At order $1$, we have directly:
\begin{align}
   \label{eq:phi_order1_fx}
   \Big(\frac{\partial \Phi_{\bY,n}}{\partial \eta}\Big)_{\eta = 0} &= -\frac{1}{n(m+p)} \langle H_{\rm eff}\rangle_0
   = \frac{1}{n^{3/2} (m+p)} \sum_{\mu,i,l} g_{\mu l} m^F_{\mu i} m^X_{il}.
\end{align}
We can then use the Maxwell relations we already described in eq.~\eqref{eq:maxwell}, 
to compute e.g.\:
\begin{align}\label{eq:maxwell_f_fx}
   \begin{dcases}
      \partial_\eta \gamma^F_{\mu i}(\eta = 0) &= 2 n (m+p) \frac{\partial}{\partial v^F_{\mu i}} \Big[\frac{\partial \Phi_{\bY,n}}{\partial \eta}(\eta = 0)\Big] = 0, \\
      \partial_\eta \lambda^F_{\mu i}(\eta = 0) &= n (m+p) \frac{\partial}{\partial m^F_{\mu i}} \Big[\frac{\partial \Phi_{\bY,n}}{\partial \eta}(\eta = 0)\Big] = \frac{1}{\sqrt{n}} \sum_l g_{\mu l} m^X_{il}. 
   \end{dcases}
\end{align}
Applying this technique to all Lagrange multipliers, we compute $U_{\bY,0}$ from eq.~\eqref{eq:def_U_fx}.
Again we will denote by lowercase letters \emph{centered variables}, e.g.\ $x_{\mu i}\equiv X_{\mu i} - m^X_{\mu i}$. 
We obtain:
\begin{align}\label{eq:U0_fx}
   U_{\bY,0} &= \frac{1}{\sqrt{n}} \sum_{\mu,i,l}[(ih)_{\mu l} f_{\mu i} x_{il} - g_{\mu l} f_{\mu i} x_{i l} + (ih)_{\mu l} m^F_{\mu i} x_{il} + (ih)_{\mu l} f_{\mu i} m^X_{il}].
\end{align}
\textbf{Order $2$ in $\eta$ --}
One can then compute, using the formulas of \cite{georges1991expand}, and very similarly to Model~\ref{model:extensive_factorization_xx}:
\begin{align}
   \label{eq:phi_order2_fx}
   \frac{1}{2}\Big(\frac{\partial^2 \Phi_{\bY,n}}{\partial \eta^2}\Big)_{\eta = 0} &= \frac{1}{2n(m+p)} \langle U^2\rangle_0,\nonumber \\
   &= \frac{1}{2 n^2 (m+p)} \sum_{\mu,i,l} [ - r_{\mu l} v^F_{\mu i} v^X_{il} + g_{\mu l}^2 v^F_{\mu i} v^X_{il}- r_{\mu l} (m^F_{\mu i})^2 v^X_{il} - r_{\mu l} v^F_{\mu i} (m^X_{il})^2 ].
\end{align}
Orders $3$ and $4$ are more tedious to compute, and we detail them in separate paragraphs.

\subsubsection{Order $3$ of the expansion}\label{subsubsec:order3_fx}

In order to compute this order of the PGY expansion, we use again the formulas of Appendix~A of \cite{georges1991expand}: 
\begin{align}
\frac{1}{3!}\Big(\frac{\partial^3 \Phi_{\bY,n}}{\partial \eta^3}\Big)_{\eta = 0} &= -\frac{1}{6n(m+p)} \langle U^3 \rangle_0.
\end{align}
We will need to introduce the (centered) third moments of the distributions 
of the independent variables $\{iH_{\mu l}, F_{\mu i}, X_{il}\}$ at $\eta = 0$.
These moments are denoted $\{c^{(3,H)}_{\mu l}, c^{(3,F)}_{\mu i} , c^{(3,X)}_{i l} \}$.
We then decompose the operator of eq.~\eqref{eq:U0_fx} as follows:
\begin{align}\label{eq:U_decomposition_fx}
   U_{\bY,0} &= \frac{1}{\sqrt{n}} \sum_{\mu,i,l}\Big[\underbrace{(ih)_{\mu l} f_{\mu i} x_{il}}_{A} + \underbrace{(- g_{\mu l} f_{\mu i} x_{i l})}_{B_H} + \underbrace{(ih)_{\mu l} m^F_{\mu i} x_{il}}_{B_F} + \underbrace{(ih)_{\mu l} f_{\mu i} m^X_{il}}_{B_X}\Big].
\end{align}
Since all the variables in the equation above are \emph{centered}, we get easily:
\begin{align}\label{eq:U3_first_term_fx}
   &\langle A^3 + B_H^3 + B_F^3 + B_X^3\rangle_0  \\
   &= \frac{1}{n^{3/2}} \sum_{\mu,i,l}[c^{(3,H)}_{\mu l} c^{(3,F)}_{\mu i} c^{(3,X)}_{i l} - g_{\mu l}^3 c^{(3,F)}_{\mu i} c^{(3,X)}_{i l}+ c^{(3,H)}_{\mu l} (m^{F}_{\mu i})^3 c^{(3,X)}_{i l}+ c^{(3,H)}_{\mu l} c^{(3,F)}_{\mu i} (m^{X}_{i l})^3]. \nonumber
\end{align}
Using again the centering of the variables and the decomposition above, we get that the only non-zero terms of the type $\langle X^2 Y\rangle_0$ with $X,Y \in \{A,B_H,B_F,B_X\}$ yields the contribution:
\begin{align}\label{eq:U3_second_term_fx}
   3& \langle A^2(B_H+B_F+B_X)\rangle_0 = \\
   & \frac{3}{n^{3/2}} \sum_{\mu,i,l}[g_{\mu l} r_{\mu l} c^{(3,F)}_{\mu i} c^{(3,X)}_{i l} + c^{(3,H)}_{\mu l} m^{F}_{\mu i} v^F_{\mu i} c^{(3,X)}_{i l} + c^{(3,H)}_{\mu l} c^{(3,F)}_{\mu i} m^{X}_{i l} v^X_{il}]\nonumber.
\end{align}
Finally, the last contribution to $\langle U^3\rangle_0$ comes from the term:
\begin{align}\label{eq:U3_third_term_fx}
   6& \langle A B_H B_F + A B_H B_X + A B_F B_X + B_H B_F B_X\rangle_0 = \\
   & \frac{6}{n^{3/2}} \sum_{\mu,i,l}[g_{\mu l} r_{\mu l} m^F_{\mu i} v^F_{\mu i} c^{(3,X)}_{i l} + g_{\mu l} r_{\mu l}  c^{(3,F)}_{\mu i}m^X_{i l} v^X_{i l}   + c^{(3,H)}_{\mu l} m^F_{\mu i} v^F_{\mu i} m^{X}_{i l} v^X_{il}  + g_{\mu l} r_{\mu l} m^F_{\mu i} v^F_{\mu i} m^{X}_{i l} v^X_{il}]\nonumber.
\end{align}
Summing the contributions from eqs.~\eqref{eq:U3_first_term_fx},\eqref{eq:U3_second_term_fx},\eqref{eq:U3_third_term_fx} yields $\langle U^3 \rangle_0$, which 
then yields:
\begin{align}\label{eq:phi_order3_fx}
\frac{1}{3!}\Big(\frac{\partial^3 \Phi_{\bY,n}}{\partial \eta^3}\Big)_{\eta = 0} &= \frac{-1}{6n^{5/2} (m+p)} \sum_{\mu,i,l} [c^{(3,H)}_{\mu l} c^{(3,F)}_{\mu i}c^{(3,X)}_{i l} - g_{\mu l}^3 c^{(3,F)}_{\mu i}c^{(3,X)}_{i l}+ c^{(3,H)}_{\mu l} c^{(3,F)}_{\mu i}(m^{X}_{i l})^2 \nonumber \\
&  + c^{(3,H)}_{\mu l} (m^{F}_{\mu i})^3 c^{(3,X)}_{i l} + 3 g_{\mu l} r_{\mu l} c^{(3,F)}_{\mu i} c^{(3,X)}_{il} + 3 c^{(3,H)}_{\mu l} m^F_{\mu i} v^F_{\mu i} c^{(3,X)}_{il} + 3 c^{(3,H)}_{\mu l} c^{(3,X)}_{\mu i} m^X_{il} v^X_{il}\nonumber \\
& + 6 g_{\mu l} r_{\mu l} m^F_{\mu i} v^F_{\mu i} c^{(3,X)}_{il} + 6 g_{\mu l} r_{\mu l} c^{(3,F)}_{\mu i} m^X_{i l} v^X_{il} + 6 c^{(3,H)}_{\mu l} m^F_{\mu i} v^F_{\mu i} m^X_{i l} v^X_{il} \nonumber \\
& + 6 g_{\mu l} r_{\mu l} m^F_{\mu i} v^F_{\mu i} m^X_{il} v^X_{il}].
\end{align}
From eq.~\eqref{eq:phi_order3_fx}, since all involved terms inside the sum are of order $\mathcal{O}_N(1)$, it is clear that the third order
is subdominant:
\begin{align}\label{eq:phi_order3_final_fx}
\frac{1}{3!}\Big(\frac{\partial^3 \Phi_{\bY,n}}{\partial \eta^3}\Big)_{\eta = 0} &= \smallO_n(1).
\end{align}
\textbf{Higher-order moments --} An important remark that one can already conjecture by generalizing from eqs.~\eqref{eq:phi_order3_fx},\eqref{eq:phi_order3_final_fx} is that at any given 
order of perturbation in $\eta$, \emph{only the first two moments of the fields $\bH,\bF,\bX$ will appear at dominant order}.
This conjecture arises as a consequence of a simple scaling argument: the higher-order moments constraint too much indices on which we can sum, and thus the terms involving them can be neglected.
Note that we noticed already a completely similar behavior in the symmetric case in Appendix~\ref{subsec_app:first_orders_pgy_xx}.

\subsubsection{Order $4$ of the expansion}\label{subsubsec:order4_fx}
At order $4$, one can again use Appendix~A of \cite{georges1991expand}. The general (and quite heavy) formula is:
\begin{align}
n(m+p) \Big(\frac{\partial^4 \Phi_{\bY,n}}{\partial \eta^4}\Big)_{\eta = 0} &= \langle U^4\rangle_0 - 3 \langle U^2\rangle^2_0 \nonumber \\ 
&- 3 \sum_{\mu,i} \Big[\partial^2_\eta \lambda^F_{\mu i} \langle U^2 f_{\mu i}\rangle_0 + \frac{\partial^2_\eta \gamma^F_{\mu i}}{2} \langle U^2(f_{\mu i}^2 + 2 m^F_{\mu i} f_{\mu i} - v^F_{\mu i})\rangle_0\Big] \nonumber \\
& - 3 \sum_{i,l} \Big[\partial^2_\eta \lambda^X_{i l} \langle U^2 x_{i l}\rangle_0 + \frac{\partial^2_\eta \gamma^X_{i l}}{2} \langle U^2(x_{i l}^2 + 2 m^X_{i l} x_{i l} - v^X_{i l})\rangle_0\Big] \nonumber \\
\label{eq:order4_general}
&- 3 \sum_{\mu,l} \Big[\partial^2_\eta \omega_{\mu l} \langle U^2 (ih)_{\mu l}\rangle_0 + \frac{\partial^2_\eta b_{\mu l}}{2} \langle U^2((ih)_{\mu l}^2 - 2 g_{\mu l} (ih)_{\mu l} + r_{\mu l})\rangle_0\Big].
\end{align}
This section describes the calculation of the order $4$ perturbation of the free entropy for Model~\ref{model:extensive_factorization_fx}.
It is particularly tedious and lengthy, but the techniques involved are not conceptually complicated.
For simplicity, we will not consider terms involving cumulants of order $3$ and $4$ of the variables $ih_{\mu l}$, $f_{\mu i}$ and $x_{il}$. 
One can check that the terms involving these moments cancel out, and in the end only yield sub-dominant contributions in the thermodynamic limit.
From eq.~\eqref{eq:phi_order2_fx} and the Maxwell relations of eq.~\eqref{eq:maxwell}, we obtain the derivatives of the Lagrange multipliers at leading order and at $\eta = 0$:
\begin{align}\label{eq:lagrange_order_2_fx}
   \begin{cases}
   \partial^2_\eta \, \omega_{\mu l} &= \frac{2}{n} g_{\mu l} \sum_i [(m^F_{\mu i})^2 v^X_{il} + v^F_{\mu i} (m^X_{il})^2], \\
   \partial^2_\eta  \, b_{\mu l} &= -\frac{2}{n} \sum_i [v^F_{\mu i} v^X_{il} + (m^F_{\mu i})^2 v^X_{il} + v^F_{\mu i} (m^X_{il})^2], \\
   \partial^2_\eta \, \lambda^F_{\mu i} &= \frac{2}{n} m^F_{\mu i} \sum_l [r_{\mu l} (m^X_{il})^2 - g_{\mu l}^2 v^X_{il}], \\
   \partial^2_\eta  \, \gamma^F_{\mu i} &= -\frac{2}{n} \sum_l [r_{\mu l} v^X_{il} + r_{\mu l} (m^X_{il})^2 - g_{\mu l}^2 v^X_{il}], \\
   \partial^2_\eta \, \lambda^X_{i l} &= \frac{2}{n} m^X_{i l} \sum_\mu [r_{\mu l} (m^F_{\mu i})^2 - g_{\mu l}^2 v^F_{\mu i}], \\
   \partial^2_\eta  \, \gamma^X_{i l} &= -\frac{2}{n} \sum_\mu [r_{\mu l} v^F_{\mu i} + r_{\mu l} (m^F_{\mu i})^2 - g_{\mu l}^2 v^F_{\mu i}].
   \end{cases}
\end{align}
Discarding as we mentioned the cumulants of order greater than $3$, we can compute:
\begin{align*}
   \begin{dcases}
   \langle U^2 x_{il}\rangle_0 &= - \frac{2}{n} m^X_{il} v^X_{il} \sum_\mu r_{\mu l} v^F_{\mu i}, \\
   \langle U^2(x_{il}^2 + 2 m^X_{il} x_{il} - v^X_{il})\rangle_0 &= \frac{2 (v^X_{il})^2 }{n} \sum_\mu [-r_{\mu l} v^F_{\mu i} + g_{\mu l}^2 v^F_{\mu i} - r_{\mu l}(m^F_{\mu i})^2] - \frac{4}{n} v^X_{il} (m^X_{il})^2 \sum_\mu r_{\mu l} v^F_{\mu i}.
   \end{dcases}
\end{align*}
From this and eq.~\eqref{eq:lagrange_order_2_fx}, one can obtain the term involving the derivatives of the Lagrange parameters $\lambda^X$ and $\gamma^X$ in eq.~\eqref{eq:order4_general}:
\begin{align}\label{eq:order4_lagrange_x_fx}
   - 3 \sum_{i,l} &\Big[\partial^2_\eta \lambda^X_{i l} \langle U^2 x_{i l} \rangle_0 + \frac{\partial^2_\eta \gamma^X_{i l}}{2} \langle U^2(x_{i l}^2 + 2 m^X_{i l} x_{i l} - v^X_{i l})\rangle_0\Big] = \\
&- \frac{12}{n^2} \sum_{i,l} v^X_{il} (m^X_{il})^2 \Big[\sum_\mu r_{\mu l}v^F_{\mu i}\Big]^2 - \frac{6}{n^2} \sum_{i,l} (v^X_{il})^2 \Big[\sum_\mu [r_{\mu l} v^F_{\mu i} + r_{\mu l}(m^F_{\mu i})^2 - g_{\mu l}^2 v^F_{\mu i}]\Big]^2. \nonumber
\end{align}
Doing similarly for $(ih)_{\mu l}$ and $f_{\mu i}$, we obtain all the terms involving derivatives of the Lagrange multipliers in eq.~\eqref{eq:order4_general}:
\begin{align}\label{eq:order4_lagrange_f_fx}
   - 3 \sum_{\mu,i} &\Big[\partial^2_\eta \lambda^F_{\mu i} \langle U^2 f_{\mu i}\rangle_0 + \frac{\partial^2_\eta \gamma^F_{\mu i}}{2} \langle U^2(f_{\mu i}^2 + 2 m^F_{\mu i} f_{\mu i} - v^F_{\mu i})\rangle_0\Big] = \\
&- \frac{12}{n^2} \sum_{\mu, i} v^F_{\mu i} (m^F_{\mu i})^2 \Big[\sum_l r_{\mu l}v^X_{i l}\Big]^2 - \frac{6}{n^2} \sum_{\mu,i} (v^F_{\mu i})^2 \Big[\sum_l [r_{\mu l} v^X_{i l} + r_{\mu l}(m^X_{i l})^2 - g_{\mu l}^2 v^X_{i l}]\Big]^2, \nonumber\\
   \label{eq:order4_lagrange_h_fx}
   - 3 \sum_{\mu,l} &\Big[\partial^2_\eta \omega_{\mu l} \langle U^2 (ih)_{\mu l}\rangle_0 + \frac{\partial^2_\eta b_{\mu l}}{2} \langle U^2((ih)_{\mu l}^2 - 2 g_{\mu l} (ih)_{\mu l} + r_{\mu l})\rangle_0\Big] = \\
& \frac{12}{n^2} \sum_{\mu,l} r_{\mu l} g_{\mu l}^2 \Big[\sum_i v^F_{\mu i}v^X_{i l}\Big]^2 - \frac{6}{n^2} \sum_{\mu,l} r_{\mu l}^2 \Big[\sum_i [v^F_{\mu i} v^X_{il} + (m^F_{\mu i})^2 v^X_{il} + v^F_{\mu i}(m^X_{il})^2]\Big]^2. \nonumber
\end{align}
The calculation at order $2$ already gave (cf.\ eq.~\eqref{eq:phi_order2_fx}):
\begin{align}\label{eq:order4_U2_term_fx}
   -3 \langle U^2\rangle^2_0 = -\frac{3}{n^2} \sum_{\substack{\mu,i,l \\ \mu' i' l'}} &[ - r_{\mu l} v^F_{\mu i} v^X_{il} + g_{\mu l}^2 v^F_{\mu i} v^X_{il}- r_{\mu l} (m^F_{\mu i})^2 v^X_{il} - r_{\mu l} v^F_{\mu i} (m^X_{il})^2] \\ 
   & \hspace{0.5cm} \times [ - r_{\mu' l'} v^F_{\mu' i'} v^X_{i'l'} + g_{\mu' l'}^2 v^F_{\mu' i'} v^X_{i'l'}- r_{\mu' l'} (m^F_{\mu' i'})^2 v^X_{i'l'} - r_{\mu' l'} v^F_{\mu' i'} (m^X_{i'l'})^2]. \nonumber
\end{align}
We finally have to compute $\langle U^4 \rangle_0$, whose calculation is very tedious, but not conceptually difficult.
We again make use of the decomposition of eq.~\eqref{eq:U_decomposition_fx}.
A first simplification arises when using that the variables are centered and that we can neglect their moments of odd order.
This implies:
\begin{align}\label{eq:decomposition_U4_fx}
   \langle U^4\rangle_0 &= \underbrace{\langle A^4 \rangle_0}_{I_1} + \underbrace{6 \langle A^2(B_H^2+B_F^2+B_X^2) \rangle_0}_{I_2} +  \underbrace{\langle (B_H+B_F+B_X)^4\rangle_0}_{I_3} + \smallO_n(1).
\end{align}
We now compute the three terms $I_1,I_2,I_3$ independently.
An explicit calculation gives 
\begin{align}\label{eq:U4_I1_fx}
      I_1 &= \frac{6}{n^2} \sum_{\mu, i, l} r_{\mu l}^2 (v^F_{\mu i})^2 (v^X_{i l})^2 + \frac{3}{n^2} \sum_{\substack{\mu,i,l \\ \mu' i' l'}} r_{\mu l} r_{\mu' l'} v^F_{\mu i} v^F_{\mu' i'} v^X_{il} v^X_{i'l'} \\ 
      &+\frac{6}{n^2} \sum_{\mu, i, l} r_{\mu l} v^F_{\mu i} v^X_{il} \Big(\sum_{\mu'} r_{\mu' l} v^F_{\mu' i} v^X_{il} + \sum_{i'} r_{\mu l} v^F_{\mu i'} v^X_{i' l} + \sum_{l'} r_{\mu l'} v^F_{\mu i} v^X_{il'}\Big) + \smallO_n(1)\nonumber.
\end{align}
Importantly, the indices are not supposed to be pairwise distinct unless explicitly stated so.
The terms $I_2,I_3$ can be explicitly computed as well, and are very lengthy:
\begin{align}\label{eq:U4_I2_fx}
      I_2 &= \frac{6}{n^2}\sum_{\substack{\mu,i,l \\ \mu' i' l'}} r_{\mu l} v^F_{\mu i} v^X_{i l} [-g_{\mu' l'}^2 v^F_{\mu' i'} v^X_{i'l'} + r_{\mu'l'} v^F_{\mu' i'} (m^X_{i'l'})^2 + r_{\mu'l'}(m^F_{\mu'i'})^2 v^X_{il}]
      \\
      &+\frac{12}{n^2} \sum_{\mu,i,l} r_{\mu l}v^F_{\mu i} v^X_{il}[-g_{\mu l}^2 v^F_{\mu i} v^X_{il} + r_{\mu l} v^F_{\mu i} (m^X_{il})^2 + r_{\mu l} (m^F_{\mu i})^2 v^X_{il}]\nonumber \\ 
      &+\frac{12}{n^2} \sum_{\mu,i,l} \sum_{\mu'} r_{\mu l}v^F_{\mu i} v^X_{il}[-g_{\mu' l}^2 v^F_{\mu' i} v^X_{il} + r_{\mu' l} v^F_{\mu' i} (m^X_{il})^2 + r_{\mu' l} (m^F_{\mu' i})^2 v^X_{il}] \nonumber\\ 
      &+\frac{12}{n^2} \sum_{\mu,i,l} \sum_{i'} r_{\mu l}v^F_{\mu i} v^X_{il}[-g_{\mu l}^2 v^F_{\mu i'} v^X_{i'l} + r_{\mu l} v^F_{\mu i'} (m^X_{i'l})^2 + r_{\mu l} (m^F_{\mu i'})^2 v^X_{i'l}] \nonumber\\ 
      &+\frac{12}{n^2} \sum_{\mu,i,l} \sum_{l'} r_{\mu l}v^F_{\mu i} v^X_{il}[-g_{\mu l'}^2 v^F_{\mu i} v^X_{il'} + r_{\mu l'} v^F_{\mu i} (m^X_{il'})^2 + r_{\mu l'} (m^F_{\mu i})^2 v^X_{il'}] \nonumber  + \smallO_n(1), \\ 
      \label{eq:U4_I3_fx}
      I_3 &= \frac{3}{n^2} \sum_{\substack{\mu,i,l \\ \mu' i' l'}} [g_{\mu l}^2 g_{\mu' l'}^2 v^F_{\mu i} v^F_{\mu' i'} v^X_{il} v^X_{i'l'} + r_{\mu l} r_{\mu' l'} (m^F_{\mu i})^2 (m^F_{\mu' i'})^2 v^X_{il} v^X_{i'l'} + r_{\mu l} r_{\mu' l'} v^F_{\mu i} v^F_{\mu' i'} (m^X_{il})^2 (m^X_{i'l'})^2] \nonumber \\ 
      &- \frac{6}{n^2} \sum_{\substack{\mu,i,l \\ \mu' i' l'}} [-r_{\mu l}r_{\mu'l'} (m^F_{\mu' i'})^2 v^F_{\mu i} (m^X_{il})^2 v^X_{i'l'} + g_{\mu l}^2 r_{\mu'l'} (m^F_{\mu' i'})^2 v^F_{\mu i} v^X_{il} v^X_{i'l'} + g_{\mu l}^2 r_{\mu'l'} v^F_{\mu' i'} v^F_{\mu i} (m^X_{i'l'})^2 v^X_{il}] \nonumber \\ 
      &- \frac{12}{n^2} \sum_{\mu,i,l} \Big[\sum_{\mu'} g_{\mu l}^2 r_{\mu' l} (m^F_{\mu i})^2 v^F_{\mu i} (v^X_{il})^2 - \sum_{i'} r_{\mu l}^2 (m^F_{\mu i'})^2 v^F_{\mu i} (m^X_{il})^2 v^X_{i'l} + \sum_{l'} g_{\mu l}^2 r_{\mu l'}  (v^F_{\mu i})^2 (m^X_{il'})^2 v^X_{il}  \Big] \nonumber \\ 
      &+ \frac{6}{n^2} \sum_{\mu,i,l} \Big[\sum_{\mu',i'} m^F_{\mu i} m^F_{\mu' i} m^F_{\mu' i'} m^F_{\mu i'} r_{\mu l} r_{\mu' l} v^X_{il} v^X_{i'l} + \sum_{i',l'} m^X_{i l} m^X_{i' l} m^X_{i' l'} m^X_{i l'} v^F_{\mu i} v^F_{\mu i'} r_{\mu l} r_{\mu l'}\Big] \nonumber, \\
      &+ \frac{6}{n^2} \sum_{\mu,i,l} \Big[\sum_{\mu',l'} g_{\mu l} g_{\mu' l} g_{\mu' l'} g_{\mu l'} v^F_{\mu i} v^F_{\mu' i} v^X_{il} v^X_{il'}\Big] + \smallO_n(1).
\end{align}
Many simplifications occur in the terms of \cref{eq:order4_lagrange_x_fx,eq:order4_lagrange_f_fx,eq:order4_lagrange_h_fx,eq:order4_U2_term_fx,eq:U4_I1_fx,eq:U4_I2_fx,eq:U4_I3_fx}.
Two type of terms are for instance negligible:
\begin{itemize}[leftmargin=*]
   \item Terms of the type $n^{-4} \sum_{\mu i l} A_{\mu i l}$, with $A_{\mu i l}$ typically of order $1$. These terms are negligible by a simple scaling argument.
   \item Terms involving $m^F$ or $m^X$. For instance, the term:
   \begin{align*}
      \sum_{\mu} \sum_{i \neq i'} \sum_{l \neq l'} m^X_{i l} m^X_{i l'} m^X_{i' l'} m^X_{i' l} r_{\mu l} r_{\mu l'} v^F_{\mu i} v^F_{\mu i'}.
   \end{align*}
   By Hypothesis~\ref{hyp:means_uncorrelated_extensive_xx}, the variables $m^F,m^X$ behave like uncorrelated variables, so that all these terms will be negligible. 
   A more detailed explanation of how uncorrelated variables leads to all these terms being negligible can be found e.g.\ in \cite{maillard2019high}.
   This is precisely the sort of terms that are not negligible when involving $g_{\mu l}$, because of the structure described in \ref{hyp:structure_g_extensive_xx_new}.
\end{itemize}
We can now sum all \cref{eq:order4_lagrange_x_fx,eq:order4_lagrange_f_fx,eq:order4_lagrange_h_fx,eq:order4_U2_term_fx,eq:U4_I1_fx,eq:U4_I2_fx,eq:U4_I3_fx}, simplifying the terms that are negligible by the arguments above, and checking that almost all non-negligible terms are cancelling each other.
This is a lengthy but straightforwards calculation, and we reach the result:
\begin{align}\label{eq:phi_order4_fx}
\frac{1}{4!}\Big(\frac{\partial^4 \Phi_{\bY,n}}{\partial \eta^4}\Big)_{\eta=0} &= \frac{1}{4 n^3(m+p)} \sum_i \sum_{\mu_1 \neq \mu_2} \sum_{l_1 \neq l_2} g_{\mu_1 l_1} g_{\mu_1 l_2} g_{\mu_2 l_2} g_{\mu_2 l_1} v^F_{\mu_1 i} v^F_{\mu_2 i} v^X_{i l_1} v^X_{i l_2} + \smallO_n(1).
\end{align}
This ends the derivation of Result~\ref{result:pgy_order_4_extensive_fx}.

%% file: appendix_solution_matytsin.tex
\section{On the solutions to Matytsin's equations}
\label{sec_app:matytsin_dyson}

\subsection{Derivation of Result~\ref{result:matytsin_solution_dyson}}

Note first that this result is shown in the very particular case in which $\bS$ is a Wigner matrix in \cite{bun2014instanton}.
Let us recall the complete form of Burgers' equation:
\begin{align}\label{eq:matytsin_pde_original}
    \begin{dcases}
        \partial_t \rho(x,t) + \partial_x[\rho(x,t) v(x,t)] &= 0, \\ 
        \partial_t v(x,t) + v(x,t) \partial_x v(x,t) &= \pi^2 \rho(x,t) \partial_x \rho(x,t).
    \end{dcases}
\end{align}
Recall that $g_{\bY(t)}(z)$ is the Stieltjes transform of the asymptotic eigenvalue distribution of $\bY(t) / \sqrt{m}$, 
and similarly $g_\bS$ is the Stieltjes transform of $\rho_\bS$.
It is known by simple free probability that for every $z$ such that $\mathrm{Im}[z] > 0$, we have that $\mathrm{Im}[g_{\bY(t)}(z)] > 0$ and that $g_{\bY(t)}$ satisfies 
the free convolution identity:
\begin{align}\label{eq:free_prob_additive_noise}
	g_{\bY(t)}(z) &= g_\bS[z + t g_{\bY(t)}(z)].
\end{align}
Note that defining $\rho$ and $v$ by eq.~\eqref{eq:def_rho_v_solution_dyson}, the boundary conditions of eq.~\eqref{eq:burgers_complex} 
are directly satisfied by the Stieltjes-Perron theorem.
In the following, we denote by $(\rho_\epsilon,v_\epsilon)$ the functions defined by eq.~\eqref{eq:def_rho_v_solution_dyson} 
without taking the $\epsilon \downarrow 0$ limit.
By eq.~\eqref{eq:free_prob_additive_noise}, we have for any $x$, with $f_\epsilon(x,t) \equiv - \overline{g_{\bY(t)}}(x + i\epsilon) = -g_{\bY(t)}(x - i\epsilon) = v_\epsilon(x,t) + i \pi \rho_\epsilon(x,t)$:
\begin{align*}
	\partial_t f_\epsilon(x,t) &= - \big\{- f_\epsilon(x,t) - t \partial_t f_\epsilon(x,t) \big\}g_{\bS}'[x - i \epsilon - t f_\epsilon(x,t)]\big\}, \\
	\partial_x f_\epsilon(x,t) &= -\big\{1  - t \partial_x f_\epsilon(x,t) \big\}g_{\bS}'[x - i \epsilon - t f_\epsilon(x,t)]\big\}.
\end{align*}
This implies
\begin{align*}
	\partial_t f_\epsilon + f_\epsilon \partial_x f_\epsilon &= -\big\{-f_\epsilon(x,t) - t \partial_t f_\epsilon(x,t) + f_\epsilon(x,t) - t f_\epsilon(x,t) \partial_x f_\epsilon(x,t) \big\} g_{\bS}'[x - i \epsilon - t f_\epsilon(x,t)], \\
	&= t\big\{\partial_t f_\epsilon(x,t) + f_\epsilon(x,t) \partial_x f_\epsilon(x,t) \big\} g_{\bS}'[x - i \epsilon - t f_\epsilon(x,t)]
\end{align*}
This implies, that \emph{for any $\epsilon > 0$}, $f_\epsilon(x,t)$ satisfies the complex Burgers' equation for all $t \in (0,\Delta)$.
Moreover, as $\epsilon \downarrow 0$, this solution satisfies the proper boundary conditions: this ends our justification 
of Result~\ref{result:matytsin_solution_dyson}.

\medskip\noindent
\textbf{Remark: the derivation of \cite{schmidt2018statistical} --}
In the appendix of his PhD thesis \cite{schmidt2018statistical}, C.~Schmidt states a result very similar to Result~\ref{result:matytsin_solution_dyson}.
However, there is an essential issue in his justification: 
the argument is based on the use of the methods of characteristics, which is in general wrong for complex PDEs. Fortunately, the solution found remains correct, as we show above.

\subsection{Derivation of Result~\ref{result:gaussian_integral_matytsin}}\label{appendix:bird}

\myskip
In this part, we ``shoot birds with cannons'' (in the words of \cite{schmidt2018statistical}), in order 
to derive Result~\ref{result:gaussian_integral_matytsin}.
The idea is quite simple:
consider $\bA$ a symmetric $n \times n$ matrix, with a well-defined asymptotic spectral distribution $\rho_A$.
We define the Gaussian integral:
\begin{align}\label{eq:def_In}
    I_n(\bA) &\equiv \frac{1}{n^2} \ln \int \prod_{i\leq j}\frac{\mathrm{d}H_{ij}}{\sqrt{2(1+\delta_{ij})\pi/n}} \exp\Big\{-\frac{n}{4} \mathrm{Tr}[\bH^2] + \frac{n}{2} \mathrm{Tr}[\bH \bA]\Big\}.
\end{align}
It is trivial to compute it, and if $I(\rho_A) \equiv \lim_{n \to \infty} I_n(\bA)$, we have directly
\begin{align}\label{eq:I_rho_explicit}
    I(\rho_A) &= \lim_{n\to \infty} \frac{1}{4 n} \sum_{i,j} A_{ij}^2 = \frac{1}{4} \int \rho_A(\rd x) \, x^2.
\end{align}
Now comes the cannon part. Can we find back eq.~\eqref{eq:I_rho_explicit} using Matytsin's formalism and HCIZ integrals?
We start from eq.~\eqref{eq:def_In}, introducing the change of variable to the eigenvalues of $\bH$, exactly as we did in eq.~\eqref{eq:Zn_direct_quenched_2}:
\begin{align}
    I_n(\bA) &= \frac{1}{n^2} \ln \mathcal{C}_n + \frac{1}{n^2} \ln \int_{\bbR^n} \rd \bL \, \prod_{i < j} |l_i - l_j| \, e^{-\frac{n}{4} \sum_i l_i^2} \int_{\mathcal{O}(n)} \mathcal{D}\bO e^{\frac{n}{2} \mathrm{Tr}[\bO \bL \bO^\intercal \bA]}, \\
    \nonumber
    \mathcal{C}_n &=  \frac{\pi^{n(n+1)/4}}{2^{n/2} \Gamma(n+1) \prod_{i=1}^n \Gamma(i/2)} \frac{1}{(2\pi/n)^{n(n-1)/4}} \frac{1}{(4\pi/n)^{n/2}}.
\end{align}
We worked out the asymptotics of $\mathcal{C}_n$ in Section~\ref{subsec:fentropy_denoising}, and we derived:
\begin{align}
    \frac{1}{n^2} \ln \mathcal{C}_n &= \frac{3}{8} + \smallO_n(1).
\end{align}
We can now use Laplace's method on $\bL$, and we reach: 
\begin{align}\label{eq:In_Mat_1}
    I(\rho_A) = \frac{3}{8} + \sup_{\rho_L} \Big\{\frac{1}{2} \int \rho_L^{\otimes 2}(\rd x, \rd y) \ln |x-y| - \frac{1}{4} \int \rho_L(\rd x)\, x^2 + \frac{1}{2}I[\rho_L, \rho_A]\Big\},
\end{align}
with $I[\rho_L, \rho_A] = I_{\Delta=1}[\rho_L, \rho_A]$ the Matytsin function of eq.~\eqref{eq:def_matytsin_limit}.
Using this explicit form into eq.~\eqref{eq:In_Mat_1} we reach:
\begin{align}
    I(\rho_A) &= \frac{1}{4} \int \rho_A(\rd x) \, x^2 + \sup_{\rho_L} \Big\{\frac{1}{4} \int \rho_L^{\otimes 2}(\rd x, \rd y) \ln |x-y| \\ 
    &- \frac{1}{4} \int \rho_A^{\otimes 2}(\rd x, \rd y) \ln |x-y| - \frac{1}{4} \int_0^1 \rd t \int \rd x \, \rho(x,t) \big[\frac{\pi^2}{3} \rho(x,t)^2 + v(x,t)^2\big] \Big\}. \nonumber
\end{align}
Comparing with eq.~\eqref{eq:I_rho_explicit} shows the equations presented in Result~\ref{result:gaussian_integral_matytsin}.

\myskip
The last part of Result~\ref{result:gaussian_integral_matytsin} can be justified using known results on high-dimensional HCIZ integrals and their links with large deviations theory \cite{guionnet2004large,bun2014instanton}.
Indeed, these works show (see e.g.\ eqs.~(14-15) of \cite{bun2014instanton}) that the following function:
\begin{align}
    S[\rho_L] \equiv 
    - \frac{3}{8} + \frac{1}{4} \int \rho_A(\rd x) \, x^2  + \frac{1}{4} \int \rho_L(\rd x)\, x^2 - \frac{1}{2} \int \rho_L^{\otimes 2}(\rd x, \rd y) \ln |x-y| - \frac{1}{2}I[\rho_L, \rho_A]
\end{align}
is the rate function for the large deviations of the Dyson Brownian Motion starting from $\rho_A$, at time $t = 1$.
In particular, the minimum of this rate function is reached in the expected eigenvalue density of the Dyson Brownian motion, i.e.\ $\rho_L^\star = \rho_A \boxplus \sigma_\mathrm{s.c.}$, 
which shows the last part of Result~\ref{result:gaussian_integral_matytsin}.

\myskip 
\textbf{A long remark: an alternative derivation of Result~\ref{result:gaussian_integral_matytsin} --}
Finally, we note that we can show eq.~\eqref{eq:result_gaussian_integral_matytsin} (or equivalently eq.~\eqref{eq:result_gaussian_integral_matytsin_2})
in an alternative manner that does not appeal to the Gaussian integral calculation we described.
To see this, recall that we know that $\rho_L^\star = \rho_A \boxplus \sigma_\mathrm{s.c.}$, 
and we denote $\rho(x,t)$ the density of the Dyson Brownian Motion (that solves the Euler-Matytsin equations between $\rho_A$ and $\rho_L^\star$, cf.\ Result~\ref{result:matytsin_solution_dyson}).
Let us define
\begin{align}
    G(t) &\equiv \int \rd x \, \rd y \, \rho( x, t) \, \rho(y, t) \, \ln |x-y| - \int_0^t \rd u \int \rd x \, \rho(x,u) \big[\frac{\pi^2}{3} \rho(x,u)^2 + v(x,u)^2\big].
\end{align}
We will obtain that $G(t)$ is constant by showing $G'(t) = 0$. From the definition of $G$, we have: 
\begin{align}
   G'(t) &=  2 \int \partial_t \rho(\rd x, t) \, \rho(\rd y, t) \, \ln |x-y| - \int \rd x \, \rho(x,t) \big[\frac{\pi^2}{3} \rho(x,t)^2 + v(x,t)^2\big].
\end{align}
Using eq.~\eqref{eq:burgers_complex} and integration by parts, we have, with $\mathrm{P.V.}$ the principal value of the integral:
\begin{align}\label{eq:Gprime}
   G'(t) &=  2 \int \rd x\, \rho(x,t) v(x,t) \mathrm{P.V.} \Big\{\int \rd y \rho(y, t) \frac{1}{x-y} \Big\} - \int \rd x \, \rho(x,t) \big[\frac{\pi^2}{3} \rho(x,t)^2 + v(x,t)^2\big].
\end{align}
However, we know by Result~\ref{result:matytsin_solution_dyson} that here we have: 
\begin{align}
    v(x,t) &= \mathrm{P.V.} \int \frac{\rho(y,t)}{x-y} \rd y,
\end{align}
Let us emphasize that this equality, shown in Result~\ref{result:matytsin_solution_dyson}, is really specific to the fact that we are 
considering the Matytsin equations between $\rho_A$ and $\rho_L^\star = \rho_A \boxplus \sigma_\mathrm{s.c.}$.
In general, there is a remainder term in $v(x,t)$ beyond the Hilbert transform of $\rho(x,t)$ \cite{guionnet2002large,guionnet2004first,menon2017complex}. 
This identity implies that eq.~\eqref{eq:Gprime} becomes:
\begin{align}\label{eq:Gprime_2}
   G'(t) &=  \int \rd x\, \rho(x,t) v(x,t)^2  - \frac{\pi^2}{3} \int \rd x \, \rho(x,t)^3.
\end{align}
The result then follows by a simple consequence of the theory of Hilbert transformation, 
which we give here as a lemma.
It can be found in different forms e.g.\ in \cite{tricomi1985integral} or \cite{guionnet2002large,guionnet2004first}, and we give a short proof here for completeness.
\begin{lemma}[Properties of Hilbert transformation]\label{lemma:hilbert}
    \noindent
    Let $\rho$ be a (well-behaved) probability density, and denote $\mathcal{H}[\rho]$ its Hilbert transform: 
    \begin{align}
        \mathcal{H}[\rho](x) \equiv \frac{1}{\pi} \mathrm{P.V.} \int \frac{\rho(y)}{x-y} \rd y = -\frac{1}{\pi}\lim_{\epsilon \downarrow 0} \mathrm{Re}[g_\rho(x+i \epsilon)].
    \end{align}
    Then one has the identity: 
    \begin{align}
        \frac{1}{3}\int \rd x \, \rho(x)^3 &= \int \rd x \, \rho(x) \, \mathcal{H}[\rho](x)^2.
    \end{align}
\end{lemma}
Applying Lemma~\ref{lemma:hilbert} with $v(t) = \pi \mathcal{H}[\rho(t)]$ ends our alternative derivation of Result~\ref{result:gaussian_integral_matytsin}.

\myskip
\begin{proof}[Proof of Lemma~\ref{lemma:hilbert} --]
    This simple fact can be shown in (at least) two ways. The more pedestrian way is to use the Fourier transform of the Hilbert transformation, 
    which reads 
    \begin{align}
        \hat{\mathcal{H}}[\rho](u) &\equiv \int e^{-2 i \pi x u} \mathcal{H}[\rho](x) \, \rd x = - i \mathrm{sign}(u) \hat{\rho}(u).
    \end{align} 
    Here we report a more clever proof that uses the general identity:
    \begin{align}\label{eq:identity_hilbert}
        \mathcal{H}[\rho]^2 &= \rho^2 + 2 \mathcal{H}[\rho \mathcal{H}[\rho]].
    \end{align}
    This can be shown by an argument of complex analysis, based on the following theorem:
    let $f(z)$ be an analytic function in the upper-half of the complex plane. 
    Then if $a(x) \equiv \lim_{\epsilon \downarrow 0} \mathrm{Re}[f(x+i \epsilon)]$ and $b(x) \equiv \lim_{\epsilon \downarrow 0} \mathrm{Im} [f(x+i \epsilon)]$, 
    we have $b = \mathcal{H}[a]$.
    Using this result for $f^2$ (with $f = \rho + i H[\rho]$) we reach eq.~\eqref{eq:identity_hilbert}. 
    Therefore we have:
    \begin{align}
        \int \mathcal{H}[\rho](x)^2 \, \rho(x) \, \rd x &= \int \rho(x)^3 \, \rd x + 2 \int \rd x \, \rho(x) \, \mathcal{H}[\rho \mathcal{H}[\rho]](x).
    \end{align}
    It is also a common property of the Hilbert transform (for all the mentioned properties see e.g.\ \cite{tao2004lecture}) that 
    $\int g \mathcal{H}[f] =  - \int f \mathcal{H}[g]$. Applying it in the last equation yields:
    \begin{align}
        \int \rd x \, \rho(x) \, \mathcal{H}[\rho](x)^2 &= \int \rho(x)^3 \, \rd x - 2 \int \rd x \, \rho(x) \, \mathcal{H}[\rho](x)^2, 
    \end{align}
    which is the sought result.
\end{proof}

%% file: appendix_small_alpha_RIE.tex
\section{Small-\texorpdfstring{$\alpha$}{alpha} expansion of the optimal RIE}\label{sec_app:small_alpha_rie}

Here we study the behavior of the optimal rotationally-invariant denoiser of \cite{bun2016rotational} in the overcomplete limit of small $\alpha$. 
Starting from the optimal denoiser found in eq.~\eqref{eq:expression_RIE}, we need to understand the behavior of the Stieltjes transform of $\bY/\sqrt{m}$, with $\bY$ 
the shifted observation matrix of eq.~\eqref{eq:shifted_Y}.

\medskip\noindent
At leading order as $\alpha \to 0$, by the random matrix equivalent of the central limit theorem (and somewhat informally), 
$\bY / \sqrt{m}$ should behave like $\bW / \sqrt{m} + \sqrt{\Delta} \bZ / \sqrt{m}$, with $\bW$ and $\bZ$ independent Gaussian Wigner matrices. 
The Stieltjes transform of a Wigner matrix is well-known, and we get for any $y$ eigenvalue of $\bY/\sqrt{m}$ that: 
\begin{align}
  v_\bY(y) = \frac{y}{2 (\Delta + 1)} + \smallO_\alpha(1).
\end{align}
This yields:
\begin{align}
	\hat{\xi_\mu} = \frac{1}{\Delta + 1} y_\mu + \smallO_\alpha(1),
\end{align}
which corresponds to the limit of scalar denoising \cite{bun2016rotational}. 
Let us now go to higher order in $\alpha$. 
We need to consider the corrections of the 
spectrum of $\bY$ at first non-trivial order in $\alpha$. 
This can be done by recalling the $\mathcal{R}$-transform of the shifted Wishart matrix $\bR/\sqrt{m} = \bX \bX^\intercal / \sqrt{nm} - \alpha^{-1/2} \Id_m$:
\begin{align}
	\mathcal{R}_\bR(s) &= \frac{1}{\sqrt{\alpha}(1 -  s \sqrt{\alpha})} - \frac{1}{\sqrt{\alpha}} = s + \sqrt{\alpha} s^2 + \mathcal{O}(\alpha).
\end{align}
Therefore by free addition the $\mathcal{R}$-transform of $\bY/\sqrt{m}$ is:
\begin{align*}
	\mathcal{R}_\bY(s) &= s (1 + \Delta) + \sqrt{\alpha} s^2 + \mathcal{O}(\alpha) = g_{\bY}^{-1}(-s) - \frac{1}{s}.
\end{align*}
Letting $s = -g_\bY(z)$, we reach the expansion:
\begin{align}
	g_{\bY}(z) &= \overbrace{\frac{\sqrt{-4 \Delta -4 +z^2}-z}{2 (\Delta + 1}}^{\mathrm{Wigner}} \\ 
	\nonumber & - \Big[\frac{(z-\sqrt{z^2-4 (\Delta + 1)})^4}{8 (\Delta +1)^3 \big(4 (\Delta +1)+z (\sqrt{z^2-4 (\Delta +1)}-z)\big)} \Big]\sqrt{\alpha} + \mathcal{O}(\alpha).
\end{align}
For $|y| \leq 2 \sqrt{\Delta + 1}$, this gives the expansion of $v_\bY(y)$ as:
\begin{align}
	v_\bY(y) &= \frac{y}{2(\Delta+1)} + \frac{(\Delta +1-y^2)}{2 (\Delta +1)^3} \sqrt{\alpha} + \mathcal{O}(\alpha).
\end{align}
In the end, we reach the first non-trivial order:
\begin{align}\label{eq_app:denoising_rie_order_2}
	\hat{\xi}_\mu &=	\frac{1}{\Delta + 1} y_\mu - \frac{\Delta }{(\Delta + 1)^2}\Big[1 - \frac{y_\mu^2}{\Delta + 1}\Big] \sqrt{\alpha} + \mathcal{O}(\alpha).
\end{align}

%% file: appendix_pgy_denoising.tex
\section{PGY expansion for denoising}\label{sec_app:pgy_denoising}

As for the factorization problem (cf.\ Section~\ref{sec:plefka}),
we replace $H_{\rm eff}$ by $\eta H_{\rm eff}$, and we do the Plefka-Georges-Yedidia expansion in $\eta$, and then take $\eta = 1$.
The order $1$ and the $U$ operator of \cite{georges1991expand} at $\eta = 0$ are very simple:
\begin{align}
  \begin{dcases}
    n m \frac{\partial \Phi}{\partial \eta} &= - \langle H_{\rm eff} \rangle = 0, \\
    U_{\eta = 0} &= \frac{1}{\sqrt{n}}\sum_{\mu < \nu} \sum_i (i H_{\mu \nu}) x_{\mu i} x_{\nu i} =  \frac{1}{2\sqrt{n}}\sum_{\mu \neq \nu} \sum_i (i H_{\mu \nu}) x_{\mu i} x_{\nu i}.
  \end{dcases}
\end{align}
Recall that here we consider the denoising problem, so $x_{\mu i} \overset{\mathrm{i.i.d.}}{\sim} P_X$, with $\EE_{P_X}[x] = 0$ and $\EE_{P_X}[x^2] = 1$.
For generality, we will not assume that $P_\mathrm{out}$ is necessarily Gaussian
so that we can consider any noise applied component-wise, and not only additive Gaussian noise. 
 
\medskip\noindent
\textbf{Order 2 --}
From the Georges-Yedidia formalism, we can obtain the order $2$:
\begin{align*}
	\frac{1}{2} \frac{\partial^2 \Phi}{\partial \eta^2} &= \frac{1}{2 nm}\langle U^2 \rangle, \\
	&= \frac{1}{8n^2 m} \sum_{\mu_1 \neq \nu_1} \sum_{\mu_2 \neq \nu_2} \sum_{i,j} \langle (iH)_{\mu_1 \nu_1} (iH)_{\mu_2 \nu_2} \rangle \langle x_{\mu_1 i} x_{\nu_1 i} x_{\mu_2 j} x_{\nu_2 j} \rangle, \\ 
	&=  \frac{1}{8n^2 m} \sum_{\substack{\mu_1 \neq \nu_1 \\\mu_2 \neq \nu_2}} \sum_{i,j} (g_{\mu_1 \nu_1} g_{\mu_2 \nu_2} - r_{\mu_1 \nu_1} [\delta_{\mu_1 \mu_2} \delta_{\nu_1 \nu_2} + \delta_{\mu_1 \nu_2} \delta_{\nu_1 \mu_2}]) (\delta_{\mu_1 \mu_2} \delta_{\nu_1 \nu_2} \delta_{ij} + \delta_{\mu_1 \nu_2} \delta_{\nu_1 \mu_2} \delta_{ij}), \\ 
	&=  \frac{1}{4n m} \sum_{\mu \neq \nu} (g_{\mu \nu}^2 - r_{\mu \nu}) = \frac{1}{2 n m} \sum_{\mu < \nu} [g_{\mu \nu}^2 - r_{\mu \nu}].
\end{align*}

\medskip\noindent
\textbf{Order 3 --}
The order $3$ can also be computed using \cite{georges1991expand}. 
We can actually use the calculation that we described in Appendix~\ref{subsec_app:first_orders_pgy_xx}, and decompose  
the operator $U_\bY$ as $U_\bY = U_c + U_g$, with terms given in eq.~\eqref{eq:decomposition_U0_xx}.
As we described very precisely in Appendix~\ref{subsec_app:first_orders_pgy_xx}, we reach: 
\begin{align}
    \langle U^3 \rangle &= - \frac{1}{\sqrt{n}} \sum_{\substack{\mu_1,\mu_2,\mu_3 \\ \text{pairwise distinct}}} g_{\mu_1 \mu_2} g_{\mu_2 \mu_3} g_{\mu_3 \mu_1} + \mathcal{O}(n^{3/2}),
\end{align}
in which the term $\mathcal{O}(n^{3/2})$ contains all the dependency on the third-order moments of the i.i.d.\ variables $\{x_{\mu i}\}$, 
and of the independent variables $\{(ih)_{\mu \nu}\}$.
All in all, we reach the expression of eq.~\eqref{eq:phi_pgy_order_3_denoising}:
\begin{align}
	&\frac{1}{3!} \frac{\partial^3 \Phi}{\partial \eta^3} = -\frac{1}{3! nm}\langle U^3 \rangle = \frac{1}{6 n^{3/2} m} \sum_{\substack{\sigma_1, \sigma_2, \sigma_3 \\ \textrm{pairwise distincts}}} g_{\sigma_1 \sigma_2} g_{\sigma_2 \sigma_3} g_{\sigma_1 \sigma_1} + \mathcal{O}(n^{-1/2}).
\end{align}

%% file: appendix_annealed.tex
\section{Annealed calculations in the Gaussian setting}\label{sec_app:annealed}

\subsection{The symmetric case}\label{subsec_app:annealed_xx}

\noindent
\textbf{Outline of the calculation --}
We consider the inference problem of Model~\ref{model:extensive_factorization_xx}, in the Gaussian setting.
Equivalently, it can be written as (with $\lambda \equiv \Delta^{-1}$ and rescaling the observation $\bY$):
\begin{align}\label{eq:model_annealed_xx}
    \bY &= \sqrt{\frac{\lambda}{n}} \bX \bX^\intercal + \bZ,
\end{align}
in which $\bX \in \bbR^{m \times n}$, with a signal-to-noise ratio $\lambda > 0$, and $m/n \to \alpha > 0$.
Without loss of generality, we can assume that $X_{\mu i} \overset{\mathrm{i.i.d.}}{\sim} \mathcal{N}(0,1)$.
Note that the scaling of eq.~\eqref{eq:model_annealed_xx} implies that the empirical spectral distribution of $\bY/\sqrt{n}$ converges (as $n \to \infty$)
to a limit measure $\mu_Y$.
Since we assumed a standard Gaussian prior on $\bX$, one can write the partition function for this model as:
\begin{align*}
   \mathcal{Z}_{\bY,n} &\equiv \int_{\bbR^{m \times n}} \frac{{\rm d} \bx}{(2 \pi)^{mn/2}} \, \exp\Big\{-\frac{1}{2} \mathrm{Tr} \, (\bx\bx^\intercal) \Big\} \frac{\exp\Big\{- \frac{1}{4} \mathrm{Tr} \Big(\bY - \sqrt{\frac{\lambda}{n}}\bx \bx^\intercal\Big)^2\Big\}}{(2\pi)^{m(m+1)/4}},
\end{align*}
Note that compared to eq.~\eqref{eq:def_gibbs_xx} we added the diagonal terms $\mu = \nu$. 
As there are $\mathcal{O}(n)$ such terms, and the free energy is of order $\mathcal{O}(n^2)$, this addition does not affect the limit free entropy.
We compute here the annealed free entropy $\Phi_{\rm an}(\alpha) \equiv \lim_{n \to \infty} \{(nm)^{-1} \ln \EE \mathcal{Z}_{\bY,n} \}$.

\medskip\noindent
\textbf{Averaging the first moment --}
Using the Bayes-optimality assumption, we have:
\begin{align*}
\EE_\bY [\mathcal{Z}_{\bY,n}] = \int \prod_{\mu \leq \nu}\rd y_{\mu \nu}\, \prod_{a=0,1} \Big[ \int_{\bbR^{m \times n}} \frac{\rd \bX^a}{(2\pi)^{\frac{nm}{2}}} e^{-\frac{1}{2} \underset{\mu,i}{\sum} (x^a_{\mu i})^2} \frac{e^{- \frac{1}{4}\underset{\mu,\nu}{\sum}\Big(y_{\mu \nu} - \sqrt{\frac{\lambda}{n}} \underset{i}{\sum} x^a_{\mu i} x^a_{\nu i}\Big)^2}}{(2\pi)^{\frac{m(m+1)}{4}}} \Big].
\end{align*}
We can now integrate over $\bY$. We get:
\begin{align*}
\EE_\bY [\mathcal{Z}_{\bY, n}] &= \prod_{a \in \{0,1\}} \Big[ \int_{\bbR^{m \times n}} \frac{\rd \bX^a}{(2\pi)^{\frac{nm}{2}}} e^{-\frac{1}{2} \underset{\mu,i}{\sum} (x^a_{\mu i})^2}\Big] \prod_{\mu \leq \nu} \int \frac{\rd y_{\mu \nu}}{2\pi}e^{- \frac{1}{2(1+\delta_{\mu\nu})}\underset{a}{\sum}\Big(y_{\mu \nu} - \sqrt{\frac{\lambda}{n}} \underset{i}{\sum} x^a_{\mu i} x^a_{\nu i}\Big)^2}, \\
&= D_n(\alpha) \prod_{a \in \{0,1\}} \Bigg\{\int_{\bbR^{m \times n}} \frac{\rd \bX^a}{(2\pi)^{\frac{mn}{2}}} e^{-\frac{1}{2} \underset{\mu,i}{\sum} (x^a_{\mu i})^2}\Bigg\} e^{-\frac{\lambda}{4 n}  \underset{\mu,\nu}{\sum}\underset{a}{\sum} \Big(\underset{i}{\sum} x^a_{\mu i} x^a_{\nu i}\Big)^2 + \frac{\lambda}{8 n} \underset{\mu,\nu}{\sum} \Big(\underset{a}{\sum} \underset{i}{\sum} x^a_{\mu i} x^a_{\nu i} \Big)^2},
\end{align*}
with a constant $D_n(\alpha) > 0$. 
Note that the above equation can be rewritten, using the identity:
\begin{align*}
    e^{-\frac{\lambda}{4 n}  \underset{\mu,\nu}{\sum}\,\underset{a \in \{0,1\} }{\sum} \Big(\underset{i}{\sum} x^a_{\mu i} x^a_{\nu i}\Big)^2 + \frac{\lambda}{8 n} \underset{\mu,\nu}{\sum} \Big(\underset{a \in \{0,1\}}{\sum} \, \underset{i}{\sum} x^a_{\mu i} x^a_{\nu i} \Big)^2} &= \exp\Bigg\{-\frac{\lambda}{8n}\sum_{\mu,\nu} \Big[\sum_{i} x^0_{\mu i} x^0_{\nu i} - \sum_i x^1_{\mu i} x^1_{\nu i}\Big]^2\Bigg\},
\end{align*}
so that we decouple the replicas (up to a constant that only depends on $n$ and $\alpha$):
\begin{align*}
    \exp\Big\{-\frac{\lambda}{8n}\sum_{\mu,\nu} \Big[\sum_{i} x^0_{\mu i} x^0_{\nu i} - \sum_i x^1_{\mu i} x^1_{\nu i}\Big]^2\Big\} &\simeq \int \prod_{\mu \leq \nu} \rd Q_{\mu \nu} e^{-\frac{n}{4} \sum_{\mu,\nu} Q_{\mu \nu}^2} \nonumber \\
    & \qquad \exp\Big\{\frac{i \sqrt{\lambda}}{2 \sqrt{2}} \sum_{\mu,\nu} Q_{\mu \nu} \Big[\sum_i x^0_{\mu i} x^0_{\nu i} - \sum_i x^1_{\mu i} x^1_{\nu i}\Big]\Big\}.
\end{align*}
We can integrate over the prior distribution on $\{\bX^0,\bX^1\}$ and reach, up to a constant $B_n(\alpha)$ independent of $\lambda$ :
\begin{align}
\EE_\bY [\mathcal{Z}_{\bY, n}] &= B_n(\alpha) \int \prod_{\mu \leq \nu} \rd Q_{\mu \nu} \,e^{-\frac{n}{4} \underset{\mu,\nu}{\sum} Q_{\mu \nu}^2} e^{-\frac{n}{2} \ln \det \Big[\Id_m + i \sqrt{\frac{\lambda}{2}} \bQ\Big] - \frac{n}{2} \ln \det \Big[\Id_m - i \sqrt{\frac{\lambda}{2}} \bQ\Big]},  \nonumber\\
\label{eq:annealed_sym_OLO}
&= B_n(\alpha) \int \prod_{\mu \leq \nu} \rd Q_{\mu \nu} \,e^{-\frac{n}{4} \underset{\mu,\nu}{\sum} Q_{\mu \nu}^2} e^{-\frac{n}{2} \ln \det \Big[\Id_m + \frac{\lambda}{2} \bQ^2\Big] }.
\end{align}
\noindent
\textbf{The saddle point --}
 We note that in eq.~\eqref{eq:annealed_sym_OLO} the integrand only depends on the \emph{spectrum} of the matrix $\bQ$. Since $\bQ$ is integrated over the set of symmetric matrices, 
 we can write $\bQ = \bO \bD \bO^\intercal$ and use the classical change of variables to this representation (see e.g.\ Proposition~4.1.1 of \cite{anderson2010introduction}):
\begin{align}
\label{eq:annealed_sym_before_saddle}
\EE_\bY [\mathcal{Z}_{\bY, n}] &= C_n(\alpha) \int_{\bbR^m} \rd \bL \exp\Bigg\{- \frac{n}{4} \sum_{\mu=1}^m l_\mu^2 + \frac{1}{2} \sum_{\mu \neq \nu} \ln |l_\mu - l_\nu| - \frac{n}{2} \sum_{\mu = 1}^m \ln\Big(1 + \frac{\lambda}{2} l_\mu^2\Big)\Bigg\},
\end{align}
with a constant $C_n(\alpha)$ that we will compute at the end by taking the $\lambda \downarrow 0$ limit.
In the end we obtain the following variational principle for the annealed free entropy:
\begin{align}\label{eq:fe_annealed}
   \Phi_{\rm an}(\alpha) &= c(\alpha) + \sup_{\nu \in \mathcal{M}^+_1(\bbR)} \Big\{- \frac{1}{4} \EE_\nu[X^2]+ \frac{\alpha}{2}\EE_{X,Y \sim \nu} \ln |X-Y| - \frac{1}{2} \EE_\nu \Big[\ln \Big(1 + \frac{\lambda}{2} x^2\Big)\Big]\Big\},
\end{align}
in which $c(\alpha) \equiv \lim_{n \to \infty} \{(nm)^{-1}\ln C_n(\alpha) \}$, and $\mathcal{M}_1^+(\bbR)$ is the set of real probability measures.

\medskip\noindent
\textbf{Computing $c(\alpha)$ --}
When $\lambda=0$ the solution to the variational problem is easily known, as it is related to the celebrated 
Wigner semi-circle law.
More precisely, the solution to the variational problem of eq.~\eqref{eq:fe_annealed} at $\lambda=0$ is:
\begin{align*}
   \nu_\alpha^\star(\rd x) &= \frac{1}{2 \pi \alpha} \mathds{1}_{|x|\leq 2 \sqrt{\alpha}} \sqrt{4 \alpha-x^2}\rd x.  
\end{align*}
And it satisfies:
\begin{align*}
    - \frac{1}{4} \int \nu_\alpha^\star(\rd x)\, x^2 + \frac{\alpha}{2} \int \nu_\alpha^\star(\rd x) \nu_\alpha^\star(\rd y) \ln |x-y| &= -\frac{3 \alpha}{8} + \frac{\alpha}{4} \ln \alpha. 
\end{align*}
Moreover, a direct calculation at $\lambda = 0$ gives $\Phi_{\rm an}(\alpha)_{|\lambda=0} = - (\alpha/4) \ln 4 \pi$.
This yields 
\begin{align}\label{eq:constant_annealed}
   c(\alpha) &= \frac{3 \alpha}{8} - \frac{\alpha}{4} \ln 4 \pi - \frac{\alpha}{4} \ln \alpha.
\end{align}
\noindent
\textbf{Exact resolution of the variational principle --}
We now detail an exact resolution of the variational principle of eq.~\eqref{eq:fe_annealed}. This calculation 
is very much inspired by a similar derivation, for the Wigner semi-circle law, done in \cite{livan2018introduction}.
We define the potential
\begin{align*}
V_\alpha(x) \equiv \frac{1}{4 \alpha} x^2 + \frac{1}{2\alpha} \ln\Big(1+ \frac{\lambda}{2} x^2\Big), 
\end{align*}
Going back to eq.~\eqref{eq:annealed_sym_before_saddle}, prior to taking the $n \to \infty$ limit, we 
can apply the saddle point method to the set of eigenvalues $\{l_\mu\}_{\mu=1}^m$.
The saddle-point equation for the eigenvalues reads:
\begin{align}\label{eq:saddle_ev_annealed_sym}
\frac{1}{m} \sum_{\nu (\neq \mu)} \frac{1}{l_\mu - l_\nu} &= \frac{1}{2\alpha} \Big(l_\mu  +  \frac{\lambda l_\mu}{1+ \frac{\lambda}{2} l_\mu^2} \Big) = V_\alpha'(l_\mu).
\end{align}
A classical way to solve equations of the type of eq.~\eqref{eq:saddle_ev_annealed_sym} is to use Tricomi's theorem \cite{tricomi1985integral}, as explained 
for instance in \cite{livan2018introduction}. Here we follow an alternative (equivalent) way by writing a self-consistent equation on the Stieltjes transform.
We will multiply both sides of eq.~\eqref{eq:saddle_ev_annealed_sym} by $(z-l_\mu)^{-1}$, before summing the result over $\mu$.
Recall the definition of the Stieltjes transform for any $z \in \bbC_+$:
\begin{align*}
  {\cal S}(z) &\equiv \lim_{n \to \infty} \frac{1}{m} \sum_{\mu=1}^m \frac{1}{l_\mu - z}.
\end{align*}
As can be found e.g.\ in \cite{livan2018introduction} (or very quickly derived), we have at large $n,m$:
\begin{align*}
 \frac{1}{m^2} \sum_{\mu \neq \nu} \frac{1}{l_\mu - l_\nu} \frac{1}{z-l_\mu} = \frac{1}{2} {\cal S}(z)^2 + {\cal O}(1/n).
\end{align*}
We now focus on the right hand side of eq.~\eqref{eq:saddle_ev_annealed_sym}, multiplied by $(z-l_\mu)^{-1}$, and summed over $\mu$.
The first term can be simplified using the relation:
\begin{align*}
\frac{1}{2 \alpha} \frac{1}{m} \sum_{\mu=1}^m \frac{l_\mu}{z-l_\mu} = - \frac{1 + z {\cal S}(z) }{2\alpha}+ {\cal O}(1/n).
\end{align*}
The other term can also be simplified:
\begin{align}
  \Gamma(z) &\equiv  \frac{1}{2 \alpha} \frac{1}{m} \sum_{\mu=1}^m \frac{\lambda l_\mu}{1 + \frac{\lambda}{2} l_\mu^2}  \frac{1}{z-l_\mu}, \nonumber \\
  \label{eq:gamma_density_annealed}
  &= \frac{1}{2 \alpha z} \Big(\frac{1}{m} \sum_{\mu=1}^m \frac{\lambda l_\mu}{1 + \frac{\lambda}{2} l_\mu^2} + \frac{1}{m} \sum_{\mu=1}^m \frac{2}{z-l_\mu}\Big) - \frac{1}{2\alpha z} \frac{1}{m} \sum_{\mu=1}^m \frac{2}{(1+\frac{\lambda}{2} l_\mu^2)(z-l_\mu)}.
\end{align}
Note that by the symmetry $x \to -x$ in the potential $V_\alpha(x)$, the optimizing measure should be symmetric with respect to $0$. This implies that 
\begin{align*}
  \lim_{n \to \infty} \frac{1}{m} \sum_{\mu=1}^m \frac{\lambda l_\mu}{1+\frac{\lambda}{2} l_\mu^2} &= 0. 
\end{align*}
Let us denote $L > 0$ the limit:
\begin{align*}
   \frac{1}{m} \sum_{\mu=1}^m \frac{1}{1+\frac{\lambda}{2} l_\mu^2} &\equiv L + \smallO_n(1). 
\end{align*}
Finally, we have the relation
\begin{align*}
 \frac{1}{m} \sum_{\mu=1}^m \frac{2}{(1+\frac{\lambda}{2} l_\mu^2)(z-l_\mu)} &= \frac{2}{\lambda z} \frac{1}{m} \sum_{\mu=1}^m \frac{\lambda l_\mu}{(1+\frac{\lambda}{2} l_\mu^2)(z-l_\mu)} + \frac{2}{z} \frac{1}{m} \sum_{\mu=1}^m \frac{1}{1+\frac{\lambda}{2} l_\mu^2}. 
\end{align*}
Using these tricks, we reach for the term we were treating:
\begin{align*}
  \Gamma(z) &= \frac{1}{2 \alpha z} (- 2 {\cal S}(z) + \smallO_n(1)) - \frac{1}{2\alpha z} \Big[\frac{2}{\lambda z} 2 \alpha \Gamma(z) + \frac{2}{z}L + \smallO_n(1)\Big],
\end{align*}
so that 
\begin{align*}
   \Gamma(z) &= - \frac{\lambda}{2 \alpha} \frac{L + z {\cal S}(z)}{1+\frac{\lambda}{2} z^2} + \smallO_n(1).
\end{align*}
Finally going back to eq.~\eqref{eq:saddle_ev_annealed_sym} and taking the $n \to \infty$ limit, we obtain:
\begin{align*}
  {\cal S}(z)^2 &= - \frac{1}{\alpha} [1+ z {\cal S}(z)] - \frac{\lambda}{\alpha} \frac{L + z {\cal S}(z)}{1+ \frac{\lambda}{2}z^2}.
\end{align*}
We can solve this quadratic equation to obtain:
\begin{align*}
   {\cal S}(z) &= -\frac{1}{2} \Bigg[\frac{\lambda  z}{\alpha  (\frac{\lambda  z^2}{2}+1)}+\frac{z}{\alpha } - \sqrt{\Big(\frac{\lambda  z}{\alpha  (\frac{\lambda  z^2}{2}+1)}+\frac{z}{\alpha }\Big)^2-4 \Big(\frac{1}{\alpha }+\frac{\lambda  L}{\alpha (\frac{\lambda  z^2}{2}+1)}\Big)}\Bigg].
\end{align*}
The solution is taken such that ${\rm Im}[{\cal S}(z)] > 0$ for $z \in \bbC_+$ and ${\cal S}(z) \sim -1/z$ for $z \to + \infty$.
Using the Stieltjes-Perron inversion formula, the density of the maximizing measure $\nu_\alpha^\star$ is thus:
\begin{align}\label{eq:result_annealed_xx}
   \frac{{\rm d}\nu_\alpha^\star}{{\rm d}x} &= \lim_{\epsilon \to 0^+} \frac{1}{\pi} {\rm Im} [{\cal S}(x+i\epsilon)] =  \frac{1}{2\pi \alpha} \sqrt{4 \Bigg(\alpha +\frac{\lambda  \alpha L}{1 + \frac{\lambda x^2}{2}}\Bigg) - \Bigg(\frac{\lambda x}{1+\frac{\lambda x^2}{2}}+x\Bigg)^2}.
\end{align}
The constant $L > 0$ has to be chosen in order to ensure the proper normalization of the probability distribution.

\subsection{The non-symmetric case}\label{subsec_app:annealed_fx}

\noindent 
\textbf{The calculation --}
We consider now the non-symmetric setting, i.e.\ Model~\ref{model:extensive_factorization_fx}, 
with $m/n \to \alpha > 0$ and $p/n \to \psi > 0$, and in the Gaussian setting:
\begin{align*}
    \bY &= \sqrt{\frac{\lambda}{n}} \bF \bX + \bZ,
\end{align*}
in which $\bY \in \bbR^{m \times p}$, $\bF \in \bbR^{m \times n}$ and $\bX \in \bbR^{n \times p}$.
Without loss of generality, we assume an i.i.d.\ standard Gaussian prior for both $\bF,\bX$, and that the quenched noise $\bZ$ is Gaussian as well.
The partition function of this model is:
\begin{align*}
   \mathcal{Z}_{\bY,n} &= \int {\rm d}\bbf \ {\rm d}\bx \ \frac{e^{-\frac{1}{2} \sum_{\mu,i} f_{\mu i}^2}}{(2 \pi)^{\frac{mn}{2}}} \frac{e^{-\frac{1}{2} \sum_{i,l} x_{i l}^2}}{(2 \pi)^{\frac{np}{2}}} \frac{e^{-\frac{1}{2} \sum_{\mu,l}\Big(y_{\mu l} - \sqrt{\frac{\lambda}{n}} \sum_i f_{\mu i} x_{il}\Big)^2}}{(2 \pi)^{\frac{mp}{2}}}.
\end{align*}
We can write the first moment of the partition function using Bayes-optimality:
\begin{align*}
   \EE_\bY[\mathcal{Z}_{\bY,n}] &= \int \prod_{\mu,l} \rd y_{\mu l} \prod_{a \in \{0,1\}} \Big[\int \rd \bbf^a \rd \bx^a \frac{e^{-\frac{1}{2} \mathrm{Tr} \bx^a (\bx^a)^\intercal - \frac{1}{2} \mathrm{Tr} \bbf^a (\bbf^a)^\intercal}}{(2\pi)^{\frac{mn}{2} + \frac{pn}{2}}}\Big] \frac{e^{-\frac{1}{2} \sum\limits_{a \in \{0,1\}} \sum\limits_{\mu,l}\Big(y_{\mu l} - \sqrt{\frac{\lambda}{n}} \sum_i f^a_{\mu i} x^a_{il}\Big)^2}}{(2 \pi)^{mp}}.
\end{align*}
Integrating over $\bY$, we reach:
\begin{align*}
   \EE_\bY[\mathcal{Z}_{\bY,n}] &= \prod_{a \in \{0,1\}} \Big[\int \rd \bbf^a \rd \bx^a \frac{e^{-\frac{1}{2} \mathrm{Tr} \bx^a (\bx^a)^\intercal - \frac{1}{2} \mathrm{Tr} \bbf^a (\bbf^a)^\intercal}}{(2\pi)^{\frac{mn}{2} + \frac{pn}{2}}}\Big] \frac{\exp\Big\{-\frac{\lambda}{4n} \sum_{\mu,l} \Big[\sum_i f^0_{\mu i} x^0_{il} - \sum_i f^1_{\mu i} x^1_{il}\Big]^2\Big\}}{2^{\frac{mp}{2}}(2 \pi)^{\frac{mp}{2}}}.
\end{align*}
We can then use a Gaussian transformation in the last term:
 \begin{align*}
     \exp\Big\{-\frac{\lambda}{4n} \sum_{\mu,l} \Big[\sum_i f^0_{\mu i} x^0_{il} - \sum_i f^1_{\mu i} x^1_{il}\Big]^2\Big\}
     &= \Big(\frac{n}{2\pi}\Big)^{\frac{mp}{2}} \int \prod_{\mu,l} \rd Q_{\mu l} \, e^{-\frac{n}{2} \sum_{\mu,l} Q_{\mu l}^2} \nonumber \\
     & \qquad \exp\Big\{i \sqrt{\frac{\lambda}{2}} \sum_{\mu,l} Q_{\mu l} \Big[\sum_i f^0_{\mu i} x^0_{i l} - \sum_i f^1_{\mu i} x^1_{i l}\Big]\Big\}.
 \end{align*}
One integrates now over the prior distributions on $\bbf^a,\bx^a$. The two replicas yield the same contribution and we reach:
\begin{align}\label{eq:ZN_annealed}
   \EE_\bY[\mathcal{Z}_{\bY,n}] &= D_n(\alpha,\psi) \int \prod_{\mu,l} \rd Q_{\mu l} \, e^{-\frac{n}{2} \sum_{\mu,l} Q_{\mu l}^2} \, e^{-n \ln \det [\Id_p + \frac{\lambda}{2} \bQ^\intercal \bQ]},
\end{align}
with $D_n$ a constant that only depends on $n,\alpha,\psi$.
To compute this integral, we use a Weyl-type formula for integrating a function that only depends on the singular values of the integrated matrix. It is stated for instance in Proposition~4.1.3 of \cite{anderson2010introduction}. 
We need to separate the two possibilities $\alpha \leq \psi$ and $\alpha > \psi$.

\medskip\noindent
\textbf{The case $\alpha \leq \psi$ --}
We directly obtain from eq.~\eqref{eq:ZN_annealed} and the Weyl-type formula:
\begin{align*}
   \EE_\bY \mathcal{Z}_n(\bY) &= C_n(\alpha,\psi) \int_{\bbR_+^m} \prod_{\mu=1}^m \rd l_\mu \, e^{-\frac{n}{2} \sum_\mu l_\mu^2 - n \sum_\mu \ln (1 + \frac{\lambda}{2} l_\mu^2) + (p-m) \sum_\mu \ln l_\mu + \frac{1}{2} \sum_{\mu \neq \nu} \ln |l_\mu^2 - l_\nu^2|} ,
\end{align*}
in which $C_n(\alpha,\psi)$ is a constant and we denote $C(\alpha,\psi) \equiv \lim_{n \to \infty}\{\ln C_n(\alpha,\psi) / (n(m+p))\}$.
Performing a saddle point finally yields:
\begin{align}\label{eq_app:phi_annealed_fx_alpha}
   \Phi_{\rm an}(\alpha,\psi) = C(\alpha,\psi) + \sup_{\nu \in {\cal M}_1^+(\bbR_+)} &\Big\{-\frac{\alpha}{2(\alpha+\psi)}\EE_\nu[X] - \frac{\alpha}{\alpha+\psi} \EE_\nu\Big[\ln\Big(1+\frac{\lambda}{2}X\Big)\Big] \nonumber \\
   & + \frac{\alpha(\psi-\alpha)}{2(\alpha+\psi)} \EE_\nu[\ln X]+ \frac{\alpha^2}{2(\alpha+\psi)} \EE_{X,Y\sim \nu}\ln |X-Y|\Big\}. 
\end{align}
\noindent
\textbf{The case $\alpha > \psi$ --}
This time, we obtain from eq.~\eqref{eq:ZN_annealed} and the Weyl-type formula:
\begin{align*}
   \EE_\bY \mathcal{Z}_n(\bY) &= D_n(\alpha,\psi) \int_{\bbR_+^p} \prod_{l=1}^p \rd a_l \, e^{-\frac{n}{2} \sum_l a_l^2 - n \sum_l \ln (1 + \frac{\lambda}{2} a_l^2) + (m-p) \sum_l \ln a_l + \frac{1}{2} \sum_{l \neq l'} \ln |a_l^2 - a_{l'}^2|} ,
\end{align*}
in which $D_n(\alpha,\psi)$ is a constant and we denote $D(\alpha,\psi) \equiv \lim_{n \to \infty}[\ln D_n(\alpha,\psi)/ (n(m+p))]$.
Performing a saddle point finally yields:
\begin{align}\label{eq_app:phi_annealed_fx_psi}
   \Phi_{\rm an}(\alpha,\psi) = D(\alpha,\psi) + \sup_{\nu \in {\cal M}_1^+(\bbR_+)} &\Big\{-\frac{\psi}{2(\alpha+\psi)}\EE_\nu[X] - \frac{\psi}{\alpha+\psi} \EE_\nu\Big[\ln\Big(1+\frac{\lambda}{2}X\Big)\Big] \nonumber \\
   & + \frac{\psi(\alpha-\psi)}{2(\alpha+\psi)} \EE_\nu[\ln X]+ \frac{\psi^2}{2(\alpha+\psi)} \EE_{X,Y\sim \nu}\ln |X-Y|\Big\}. 
\end{align}
\noindent
\textbf{Computing the constants --}
The constants can again be found using the $\lambda \downarrow 0$ limit.
If $\alpha \leq \psi$, the extremizing measure $\nu^\star_{\alpha,\psi}$ is a scaled Marchenko-Pastur distribution \cite{marchenko1967distribution}
with ratio $\alpha /\psi \leq 1$:
\begin{align}\label{eq:lambda_0_uv}
  \nu^\star_{\alpha,\psi}({\rm d}x) &= \frac{1}{2 \pi \alpha} \frac{\sqrt{(\lambda_+ - x) (x - \lambda_-)}}{x} \mathds{1}\{\lambda_- \leq x \leq \lambda_+\} {\rm d} x, 
\end{align}
with $\lambda_\pm = (\sqrt{\psi} \pm \sqrt{\alpha})^2$.
Indeed, $\nu^\star_{\alpha,\psi}$ is the asymptotic eigenvalue distribution of $\bQ\bQ^\intercal$, in which $\{Q_{\mu l}\}$ are independent Gaussian 
random variables with zero mean and variance $n^{-1} = \psi p^{-1}$.
As an easy calculation gives $\Phi_{\rm an}(\alpha,\psi) = -[\alpha \psi/(2(\alpha+\psi))]\ln 4 \pi $ for $\lambda = 0$, this yields after some computation:
\begin{align*}
  C(\alpha,\psi) &= -\frac{\alpha ^2 \ln (\alpha )+\alpha  \psi  (-3+4\ln (2)+2 \ln (\pi ))-(\alpha -\psi )^2 \ln (\psi -\alpha )+\psi ^2 \ln (\psi )}{4 (\alpha +\psi )}.
\end{align*}
In the same way, one can get $D(\alpha,\psi) = C(\psi,\alpha)$.

\medskip\noindent
\textbf{Solving for the maximizing density --}
We can solve the variational principles of eqs.~\eqref{eq_app:phi_annealed_fx_alpha},\eqref{eq_app:phi_annealed_fx_psi} in a similar way as what we did in Appendix~\ref{subsec_app:annealed_xx} 
for the symmetric case.
In the limit $n\to\infty$, the saddle point equation for $x_\mu \equiv l_{\mu}^2$ gives:
\begin{align*}
- \frac{1}{2} - \frac{\lambda}{2} \frac{1}{1+\lambda x_{\mu}/2} + \frac{1}{n} \sum_{\nu (\neq \mu)} \frac{1}{x_{\mu}-x_{\nu}} +  \frac{(\psi-\alpha)}{2}  \frac{1}{ x_{\mu} } = 0.
\end{align*}
As before we multiply by $(z-x_{\mu})^{-1}$ and we sum over $\mu=1,\dots, m$.
In this way we get:
\begin{align*}
 \frac{1}{n m} \sum_{\mu=1}^m \sum_{\nu (\neq \mu) =1}^m \frac{1}{x_\mu-x_\nu} \frac{1}{z-x_{\mu}} =  \frac{\alpha {\cal S}(z)^2}{2} + {\cal O}(1/n).
\end{align*}
Next, we add an arbitrary parameter $\gamma > 0 $ in the calculation, before taking the limit $\gamma \downarrow 0$,  as done in \cite{livan2018introduction}.
In this way we can compute:
\begin{align*} 
\frac{1}{m}\sum_{\mu=1}^m \Big( \frac{(\psi-\alpha)}{2} + \frac{\gamma}{2} \Big) \frac{1}{(z-x_{\mu})} \frac{1}{x_{\mu}} = \Big( \frac{(\psi-\alpha)}{2} + \frac{\gamma}{2} \Big) \Big( \frac{Q - {\cal S}(z)}{z} \Big),
\end{align*}
with the constant
\begin{align*}
Q \equiv \frac{1}{m} \sum_{\mu=1}^m \frac{1}{x_\mu} = \int {\rm d} x \ \frac{\nu(x)}{x}.
\end{align*}
Next we have:
\begin{align*}
\frac{\lambda}{2}  \frac{1}{m} \sum_{\mu=1}^m \frac{1}{1+\lambda x_{\mu}/2}  \frac{1}{z-x_{\mu}} 
= \frac{\lambda}{2} \frac{1}{\lambda z/2 +1} (Y - {\cal S}(z)),
\end{align*}
with
\begin{align*}
Y \equiv \frac{\lambda}{2} \int {\rm d} x \ \frac{\mu(x)}{1+\lambda x/2} = {\cal S}(-2/\lambda).
\end{align*}
Altogether the equation for ${\cal S}(z)$ reads:
\begin{align*}
{\cal S}(z) - \frac{\lambda}{\lambda z/2 +1} (Y - {\cal S}(z)) + (\psi-\alpha+ \gamma ) \frac{Q - {\cal S}(z)}{z} + \alpha {\cal S}(z)^2 = 0,
\end{align*}
from which we have (after taking the limit $\gamma \downarrow 0$), for $z > 0$:
\begin{align*}
&{\cal S}(z) \\ 
&= \frac{1}{2\alpha} \Big[ - 1+ \frac{(\psi-\alpha)}{z} - \frac{\lambda}{1+\lambda z/2} 
+  \frac{1}{z}\sqrt{\Big( z- \psi-\alpha + \frac{z\lambda}{1+\lambda z/2}  \Big)^2 - 4 \alpha z \Big((\psi-\alpha) Q- \frac{z\lambda}{1+\lambda z/2} Y \Big)} \Big]
\end{align*}
The behavior at $z\to\infty$ fixes the constant $Y$:
\begin{align*}
Y = \frac{1}{2} ((\psi-\alpha)Q-1).
\end{align*}
From this one can deduce the final density which reads:
\begin{align}\label{eq:result_annealed_fx}
\frac{{\rm d}\nu_{\alpha,\psi}^\star}{{\rm d}x}  = \frac{1}{2 \alpha \pi x} \sqrt{-\Big( x- \psi-\alpha + \frac{x\lambda}{1+\lambda x/2}  \Big)^2 + 4 \alpha x \Big((\psi-\alpha) Q- \frac{x\lambda}{1+\lambda x/2}  \frac{1}{2} ((\psi-\alpha)Q-1) \Big)},
\end{align}
where the value of $Q$ should be determined by imposing the normalization of $\nu^\star_{\alpha,\psi}$.

%% file: appendix_misc.tex
\section{Other technicalities}\label{sec_app:misc}

\subsection{The free entropy of \texorpdfstring{\cite{kabashima2016phase}}{kab16} in the Gaussian setting}\label{subsec_app:kabashima_sol_gaussian}

In this symmetric Gaussian setting, with prior of variance $1$ and noise $\Delta > 0$ in the channel, 
the prediction of \cite{kabashima2016phase} for the quenched free entropy is:
\begin{align}\label{eq:kab_free_entropy}
	\Phi_{\mathrm{sec.}} &= \underset{q,\hat{q}}{\text{extr}} \Big\{\frac{1}{2} q \hat{q} - \frac{\hat{q}}{2} - \frac{1}{2} \ln (1-\hat{q}) - \frac{\alpha}{4}(1+\ln 2 \pi) - \frac{\alpha}{4} \ln [\Delta+(1-q^2)] \Big\}.
\end{align}
We denote it $\Phi_\mathrm{sec.}$, since we saw in Section~\ref{sec:plefka} that it corresponds to a second-order truncation of the PGY expansion for the factorization problem.
The extremum in eq.~\eqref{eq:kab_free_entropy} is given by the set of equations:
\begin{align*}
   \hat{q} = - \frac{\alpha q}{\Delta + (1-q^2)} 
   \hspace{0.5cm} ; \hspace{0.5cm}
   q = - \frac{\hat{q}}{1-\hat{q}}. 
\end{align*}
In the paramagnetic phase (which corresponds to the denoising solution) we moreover have $q = \hat{q} = 0$ and thus:
\begin{align*}
	\Phi_\mathrm{Kab.}(q=0) &= - \frac{\alpha}{4}(1+\ln 2 \pi) - \frac{\alpha}{4} \ln(1+\Delta).
\end{align*}
Coming back to the factorization problem,
the equation on $q$ is $\hat{q} = -q / (1-q)$ and (if $q \neq 0$):
\begin{align*}
	\alpha(1 - q) = \Delta + 1-q^2,
\end{align*}
which easily solves into:
\begin{align}\label{eq:sol_q_kabashima}
	q &= \frac{\alpha \pm \sqrt{(2-\alpha)^2 + 4 \Delta}}{ 2}
\end{align}
Computing $\Phi_\mathrm{sec.}$ can then easily be done: one must consider three possible solutions: 
the paramagnetic one, and the two solutions of eq.~\eqref{eq:sol_q_kabashima}. 
One then checks if said solutions are physical (i.e.\ if $q \in [0,1)$). 
To finish, one must compute the free entropies of all physical solutions, and the actual 
prediction for the free entropy is given by the largest of those.

\subsection{Absence of contribution of diagonal terms to the MMSE}\label{subsec:mmse_diagonal}

Recall that we consider the denoising model:
\begin{align}\label{eq_app:denoising_model}
  \bY &= \frac{1}{\sqrt{n}} \bX^\star (\bX^\star)^\intercal + \sqrt{\Delta} \bZ.
\end{align}
We let $\bS^\star \equiv \bX^\star (\bX^\star)^\intercal / \sqrt{n}$, 
and we define the associated asymptotic free energy:
\begin{align}
  \Phi_{\bY,n} &\equiv \frac{1}{nm} \EE \ln \int \mathcal{D}\bX e^{-\frac{1}{4\Delta} \sum_{\mu,\nu} \Big(Y_{\mu \nu} - \frac{1}{\sqrt{n}} \sum_{i=1}^n X_{\mu i} X_{\nu i}\Big)^2}.
\end{align}
And the asymptotic MMSE is:
\begin{align}\label{eq_app:mmse_denoising_full}
  \mathrm{MMSE} &\equiv \frac{1}{m^2} \EE \sum_{\mu,\nu} \Big(S^\star_{\mu \nu} - \langle S_{\mu \nu} \rangle\Big)^2.
\end{align}
In this paragraph, we also consider
a definition of the MMSE ``not taking into account the diagonal'':
\begin{align}\label{eq_app:mmse_denoising_no_diagonal}
  \mathrm{MMSE} &\equiv \frac{1}{m^2} \EE \sum_{\mu \neq \nu} \Big(S^\star_{\mu \nu} - \langle S_{\mu \nu} \rangle\Big)^2,
\end{align}
In the limit $n \to \infty$, we now show that eqs.~\eqref{eq_app:mmse_denoising_full} and \eqref{eq_app:mmse_denoising_no_diagonal} are equivalent.
Equivalently, we will show:
\begin{align}\label{eq_app:mmse_diagonal_vanishing}
 \mathrm{MMSE}_\mathrm{diag} \equiv \frac{1}{m^2} \EE \Bigg\{ \sum_{\mu=1}^m \Big(S^\star_{\mu \mu} - \langle S_{\mu \mu} \rangle\Big)^2 \Bigg\} \to_{n \to \infty} 0.
\end{align}
\begin{proof}[Proof of eq.~\eqref{eq_app:mmse_diagonal_vanishing} --]
	We will make use of the \emph{Nishimori identity}, a consequence of Bayes-optimality.
Physically speaking, the Nishimori identity shows that the planted solution $\bS^\star$ behaves like a replica of the system 
under the posterior distribution, and we refer the reader to \cite{barbier2019optimal} for its elementary proof:
\begin{proposition}[Nishimori identity] \label{prop:nishimori}
    \noindent
	Let $(X,Y) \in \mathbb{R}^{n_1} \times \mathbb{R}^{n_2}$ be a couple of random variables. Let $k \geq 1$ and let $X^{(1)}, \dots, X^{(k)}$ be $k$ i.i.d.
	samples (given $Y$) from the conditional distribution $\bbP(X=\cdot\, | Y)$. Let us denote 
	$\langle - \rangle$ the expectation operator w.r.t.\ $\bbP(X= \cdot\, | Y)$ and $\mathbb{E}$ the expectation w.r.t.\ $(X,Y)$. Then, for all continuous bounded function $g$ we have
	\begin{align*}
	\mathbb{E} \langle g(Y,X^{(1)}, \dots, X^{(k)}) \rangle
	=
	\mathbb{E} \langle g(Y,X^{(1)}, \dots, X^{(k-1)}, X) \rangle\,.	
	\end{align*}
\end{proposition}
By Proposition~\ref{prop:nishimori}, we have
\begin{align}
 \mathrm{MMSE}_\mathrm{diag} &= \frac{1}{m^2} \EE \Bigg\{ \sum_{\mu=1}^m \big(S^\star_{\mu \mu}\big)^2 - \langle S_{\mu \mu} \rangle^2 \Bigg\}.
\end{align}
Using the simple variance bound $\EE X^2 \geq (\EE X)^2$, and again the Nishimori identity, we have:
\begin{align*}
 \mathrm{MMSE}_\mathrm{diag} &\leq \frac{1}{m^2} \sum_{\mu=1}^m \Big\{\EE\big[\big(S^\star_{\mu \mu}\big)^2\big] - \big(\EE \langle S_{\mu \mu} \rangle \big)^2 \Big\}, \\
 &\leq \frac{1}{m^2} \sum_{\mu=1}^m \Big\{\EE\big[\big(S^\star_{\mu \mu}\big)^2\big] - \big(\EE S^\star_{\mu \mu} \big)^2 \Big\}, \\
 &\leq \frac{1}{m^2 n} \sum_{\mu=1}^m \Big\{\sum_{i,j}\EE \big[(X^\star_{\mu i})^2 (X^\star_{\mu j})^2\big] - \sum_{i,j} \EE \big[(X^\star_{\mu i})^2\big] \EE \big[(X^\star_{\mu j})^2\big] \Big\}, \\
 &\leq \frac{1}{m^2 n} \sum_{\mu=1}^m \{n^2 + 2n - n^2\} \leq \frac{2}{m},
\end{align*}
since $\EE[(X^\star_{\mu i})^2] = 1$. This ends the proof of eq.~\eqref{eq_app:mmse_diagonal_vanishing}.
\end{proof}

\subsection{Proof of the I-MMSE theorem for denoising}\label{subsec_app:I-MMSE}

We shall derive eq.~\eqref{eq:I_mmse_denoising}, making use several times of the Nishimori identity (Proposition~\ref{prop:nishimori}).
We start with:
\begin{align}
  \frac{\partial \Psi[\rho_\bY,\rho_\bS]}{\partial \Delta^{-1}} &= \frac{\alpha \Delta}{4} + \frac{1}{4 n m} \sum_{\mu,\nu} \Big\{-2 m \Big[\EE (S^\star_{\mu \nu})^2 - \EE (\langle S_{\mu \nu} \rangle^2)\Big] 
  + \sqrt{\Delta m} \EE [Z_{\mu \nu} \langle S_{\mu \nu} \rangle]\Big\}.
\end{align}
Moreover we have (taking into account both the $\mu \leq \nu$ and $\mu > \nu$ terms:)
\begin{align}
 \EE [Z_{\mu \nu} \langle S_{\mu \nu} \rangle] &= \sqrt{\frac{m}{\Delta}} \Big[\EE (S^\star_{\mu \nu})^2 - \EE (\langle S_{\mu \nu} \rangle^2)\Big].
\end{align}
So in the end we reach:
\begin{align}
  \frac{\partial \Psi[\rho_\bY,\rho_\bS]}{\partial \Delta^{-1}} &= \frac{\alpha \Delta}{4} - \frac{1}{4 n} \sum_{\mu,\nu} \Big[\EE (S^\star_{\mu \nu})^2 - \EE (\langle S_{\mu \nu} \rangle^2)\Big] = 
  \frac{\alpha \Delta}{4} - \frac{\alpha}{4} \mathrm{MMSE}(\Delta),
\end{align}
which ends the proof.